\documentclass[runningheads]{llncs}

\usepackage{amsmath}
\usepackage{amssymb}

\usepackage{amsthm}

\usepackage{silence}
\usepackage{complexity}
\usepackage{graphicx}
\usepackage{enumitem}
\usepackage{tikz}
\usepackage[htt]{hyphenat}
\usepackage{hyperref}
\usepackage[noend]{algpseudocode}
\usepackage{algorithm}

\usepackage{thmtools}
\usepackage{thm-restate}

\declaretheorem[name=Lemma]{lem}
\declaretheorem[name=Proposition]{prop}
\declaretheorem[style=remark, numbered=no, name=Proof sketch]{proofsketch}

\usepackage{eqparbox}
\newcommand{\matheqbox}[2]{\eqmakebox[#1][c]{$\displaystyle#2$}}

\algdef{S}[WHILE]{WhileNoDo}[1]{\algorithmicwhile\ #1}%
\algnewcommand{\LineComment}[1]{\State \(\triangleright\) #1}
\algnewcommand{\LineCommentPhantom}[1]{\State \phantom{\(\triangleright\)} #1}
\algrenewcommand\algorithmicindent{1.0em}

\usepackage{xpatch}
\makeatletter
\xpatchcmd{\algorithmic}
  {\ALG@tlm\z@}{\leftmargin\z@\ALG@tlm\z@}
  {}{}
\makeatother

\usetikzlibrary{arrows,decorations.markings,decorations.pathreplacing,patterns,matrix,calc,positioning,backgrounds,arrows.meta,shapes,decorations.markings,fadings}

\tikzset{darrow/.style={decoration={
  markings,
  mark=at position .2 with {\arrowreversed{angle 90[width=2.5mm]}},
  mark=at position .8 with {\arrow{angle 90[width=2.5mm]}},
  }
  ,postaction={decorate}}}

\tikzset{-->-/.style={decoration={
  markings,
  mark=at position .8 with {\arrow{angle 90[width=2.5mm]}}},postaction={decorate}}}
  
\tikzset{-<--/.style={decoration={
  markings,
  mark=at position .8 with {\arrow{angle 90[width=2.5mm]}}},postaction={decorate}}}
  
\tikzset{->-/.style={decoration={
  markings,
  mark=at position .5 with {\arrow{angle 90[width=2.5mm]}} },postaction={decorate}}}

\tikzset{--->/.style={decoration={
  markings,
  mark=at position 1 with {\arrow{angle 90[width=2.5mm]}} },postaction={decorate}}}

\tikzset{
  nat/.style     = {fill=white,draw=none,ellipse,minimum size=0.3cm,inner sep=1pt},
}

\usepackage[most]{tcolorbox}

\newcommand*{\fullversion}{}%
\newcommand*{\hideappendix}{}%

\begin{document}
\title{The Three-Dimensional Stable Roommates Problem with Additively Separable Preferences\texorpdfstring{\thanks{This work was supported by the Engineering and Physical Sciences Research Council (Doctoral Training Partnership grant number EP/R513222/1 and grant number EP/P028306/1)}}{}}
\titlerunning{The 3D Stable Roommates Problem with Additively Separable Preferences}
\author{Michael~McKay \orcidID{0000-0003-1496-7434} \and David~Manlove \orcidID{0000-0001-6754-7308}}
\authorrunning{M.\ McKay and D.\ Manlove}

\institute{School of Computing Science, University of Glasgow, UK
\email{m.mckay.1@research.gla.ac.uk}, \email{david.manlove@glasgow.ac.uk}}

\maketitle
\begin{abstract}
The Stable Roommates problem involves matching a set of agents into pairs based on the agents' strict ordinal preference lists. The matching must be stable, meaning that no two agents strictly prefer each other to their assigned partners. A number of three-dimensional variants exist, in which agents are instead matched into triples. Both the original problem and these variants can also be viewed as hedonic games. We formalise a three-dimensional variant using general additively separable preferences, in which each agent provides an integer valuation of every other agent. In this variant, we show that a stable matching may not exist and that the related decision problem is $\NP$-complete, even when the valuations are binary. In contrast, we show that if the valuations are binary and symmetric then a stable matching must exist and can be found in polynomial time. We also consider the related problem of finding a stable matching with maximum utilitarian welfare when valuations are binary and symmetric. We show that this optimisation problem is $\NP$-hard and present a novel 2-approximation algorithm.
\end{abstract}

\section{Introduction}
\label{sec:intro}
The Stable Roommates problem (SR) is a classical problem in the domain of matching under preferences. It involves a set of agents that must be matched into pairs. Each agent provides a \emph{preference list}, ranking all other agents in strict order. We call a set of pairs in which each agent appears in exactly one pair a \emph{matching}. The goal is to produce a matching $M$ that admits no \emph{blocking pair}, which comprises two agents, each of whom prefers the other to their assigned partner in $M$. Such a matching is called \emph{stable}. This problem originates from a seminal paper of Gale and Shapley, published in 1962, as a generalisation of the Stable Marriage problem \cite{GS62}. They showed that an SR instance need not contain a stable matching. In 1985, Irving presented a polynomial-time algorithm to either find a stable matching or report that none exist, given an arbitrary SR instance \cite{Irv85}. Since then, many papers have explored extensions and variants of the fundamental SR problem model.

In this paper we consider the extension of SR to three dimensions (i.e., agents must be matched into triples rather than pairs). A number of different formalisms have already been proposed in the literature. The first, presented in 1991 by Ng and Hirschberg, was the \emph{3-Person Stable Assignment Problem} (3PSA) \cite{NH91}. In 3PSA, agents' preference lists are formed by ranking every pair of other agents in strict order. A matching $M$ is a partition of the agents into unordered triples. A \emph{blocking triple} $t$ of $M$ involves three agents that each prefer their two partners in $t$ to their two assigned partners in $M$. Accordingly, a \emph{stable matching} is one that admits no blocking triple. The authors showed that an instance of this model may not contain a stable matching and the associated decision problem is $\NP$-complete \cite{NH91}. In the instances constructed by their reduction, agents' preferences may be \emph{inconsistent} \cite{CCHUANG2207TR}, meaning that it is impossible to derive a logical order of individual agents from a preference list ranking pairs of agents.

In 2007, Huang considered the restriction of 3PSA to \emph{consistent} preferences. He showed that a stable matching may still not exist and the decision problem remains $\NP$-complete \cite{CCHUANG2207TR,CCHUANG2207FULLVERSION}. In his technical report, he also described another variant of 3PSA using \emph{Precedence by Ordinal Number} (PON). PON involves each agent providing a preference list ranking all other agents individually. An agent's preference over pairs is then based on the sum of the ranks of the agents in each pair. Huang left open the problem of finding a stable matching, as defined here, in the PON variant. He also proposed another problem variant involving a more general system than PON, in which agents provide arbitrary numerical ``ratings''. It is this variant that we consider in this paper. He concluded his report by asking if there exist special cases of 3PSA in which stable matchings can be found using polynomial time algorithms. This question is another motivation for our paper.

The same year, Iwama, Miyazaki and Okamoto presented another variant of 3PSA \cite{IMO07}. In this model, agents rank individual agents in strict order of preference, and an ordering over pairs is inferred using a specific \emph{set extension rule} \cite{AzizLang2016,BBP04}. The authors showed that a stable matching may not exist and that the decision problem remains $\NP$-complete. 

In 2009, Arkin et al.\ presented another variant of 3PSA called \emph{Geometric 3D-SR} \cite{ABEOMP09}. In this model, preference lists ranking pairs are derived from agents' relative positions in a metric space. Among other results, they showed that in this model a stable matching, as defined here, need not exist. In 2013, Deineko and Woeginger showed that the corresponding decision problem is $\NP$-complete \cite{DEINEKO20131837}.


All of the problem models described thus far, including SR, can be viewed as \emph{hedonic games} \cite{aziz_savani_moulin_2016}. A hedonic game is a type of \emph{coalition formation game}. In general, coalition formation games involve partitioning a set of agents into disjoint sets, or coalitions, based on agents' preferences. The term `hedonic' refers to the fact that agents are only concerned with the coalition that they belong to. The study of hedonic games and coalition formation games is broad and many different problem models have been considered in the literature \cite{Haj06}.

In particular, SR and its three-dimensional variants can be viewed as hedonic games with a constraint on permissible coalition sizes \cite{Woe13}. In the context of a hedonic game, the direct analogy of stability as described here is \emph{core stability}. In a given hedonic game, a partition is \emph{core stable} if there exists no set of agents $S$, of any size, where each agent in $S$ prefers $S$ to their assigned coalition \cite{aziz_savani_moulin_2016}.

Recently, Boehmer and Elkind considered a number of hedonic game variants, including 3PSA, which they described as \emph{multidimensional roommate games} \cite{Boe20}. In their paper they supposed that the agents have \emph{types}, and an agent's preference between two coalitions depends only on the proportion of agents of each type in each coalition. They showed that, for a number of different `solution concepts', the related problems are $\NP$-hard, although many problems are solvable in linear time when the room size is a fixed parameter. For stability in particular, they presented an integer linear programming formulation to find a stable matching in a given instance, if one exists, in linear time.

In 2020, Bredereck et al.\ considered another variation of multidimensional roommate games involving either a \emph{master list} or \emph{master poset}, a central list or poset from which all agents' preference lists are derived \cite{Bre20}. They presented two positive results relating to restrictions of the problem involving a master poset although they showed for either a master list or master poset that finding a stable matching in general remains $\NP$-hard or $\W[1]$-hard, for three very natural parameters.

Other research involving hedonic games with similar constraints has considered Pareto optimality rather than stability \cite{CSEH2019}; `flatmate games', in which any coalition contains three or fewer agents \cite{Brandt2020FindingAR}; and strategic aspects \cite{WrightV15}.

The template of a hedonic game helps us formalise the extension of SR to three dimensions. In this paper we apply the well-known system of  \emph{additively separable preferences} \cite{Aziz:2011:OPA:2283396.2283405}. In a general hedonic game, additive separable preferences are derived from each agent $\alpha_i$ assigning a numerical valuation $\mathit{val}_{\alpha_i}(\alpha_j)$ to every other agent $\alpha_j$. A preference between two sets is then obtained by comparing the sum of valuations of the agents in each set. This system formalises the system of ``ratings'' proposed by Huang \cite{CCHUANG2207TR}. In a general hedonic game with additively separable preferences, a core stable partition need not exist, and the associated decision problem is strongly $\NP$-hard  \cite{SUNG2010635}. This result holds even when preferences are \emph{symmetric}, meaning that $\mathit{val}_{\alpha_i}(\alpha_j) = \mathit{val}_{\alpha_j}(\alpha_i)$ for any two agents $\alpha_i, \alpha_j$ \cite{AZIZ2013316}.

The three-dimensional variant of SR that we consider in this paper can also be described as an additively separable hedonic game in which each coalition in a feasible partition has size three. To be consistent with previous research relating to three-dimensional variants of SR \cite{CCHUANG2207TR,IMO07}, in this paper we refer to a partition into triples as a \emph{matching} rather than a partition and write \emph{stable matching} rather than \emph{core stable partition}.  We finally remark that the usage of the terminology ``three-dimensional'' to refer to the coalition size rather than, say, the number of agent sets \cite{NH91}, is consistent with previous work in the literature \cite{ABEOMP09,Bre20,IMO07,Woe13}.


\subsubsection{Our contribution.} In this paper we use additively separable preferences to formalise the three-dimensional variant of SR first proposed by Huang in 2007 \cite{CCHUANG2207TR}. The problem model can be equally viewed as a modified hedonic game with additively separable preferences \cite{AZIZ2013316,SUNG2010635}. We show that deciding if a stable matching exists is $\NP$-complete, even when valuations are binary (Section~\ref{sec:generalbinary}). In contrast, when valuations are binary and symmetric we show that a stable matching always exists and give an $O(|N|^3)$ algorithm for finding one, where $N$ is the set of agents (Sections~\ref{sec:3dsrsasbin_prelims} -- \ref{sec:3dsrsasbin_findingarbitrarystablematching}). We believe that this restriction to binary and symmetric preferences has practical as well as theoretical significance. For example, this model could be applied to a social network graph involving a symmetric ``friendship'' relation between users. Alternatively, in a setting involving real people it might be reasonable for an administrator to remove all asymmetric valuations from the original preferences.

We also consider the notion of \emph{utility} based on agents' valuations of their partners in a given matching. This leads us to the notion of \emph{utilitarian welfare} \cite{10.5555/2832249.2832313,10.5555/3398761.3398791} which is the sum of the utilities of all agents in a given matching. We consider the problem of finding a stable matching with maximum utilitarian welfare given an instance in which valuations are binary and symmetric. We prove that this optimisation problem is $\NP$-hard and provide a novel 2-approximation algorithm (Section~\ref{sec:3dsrsasbin_utilitarianwelfare}).

We continue in the next section (Section~\ref{sec:model}) with some preliminary definitions and results.



\section{Preliminary definitions and results}
\label{sec:model}
Let $N=\{\alpha_1,\dots,\alpha_{|N|}\}$ be a set of \emph{agents}. A \emph{triple} is an unordered set of three agents. A \emph{matching} $M$ comprises a set of pairwise disjoint triples. For any agent $\alpha_i$, if some triple in $M$ contains $\alpha_i$ then we say that $\alpha_i$ is \emph{matched} and use $M(\alpha_i)$ to refer to that triple. If no triple in $M$ contains $\alpha_i$ then we say that $\alpha_i$ is \emph{unmatched} and write $M(\alpha_i)=\varnothing$. Given a matching $M$ and two distinct agents $\alpha_i, \alpha_j$, if $M(\alpha_i)=M(\alpha_j)$ then we say that $\alpha_j$ is a \textit{partner} of $\alpha_i$.

We define \emph{additively separable preferences} as follows. Each agent $\alpha_i$ supplies a \textit{valuation function} $\mathit{val}_{\alpha_i} : N\setminus \{ {\alpha_i} \} \longrightarrow \mathbb{Z}$. Given agent $\alpha_i$, let the \emph{utility} of any set $S\subseteq N$ be $u_{\alpha_i}(S) = \sum\limits_{{\alpha_j}\in S \setminus \{ \alpha_i \}} \mathit{val}_{\alpha_i}({\alpha_j})$. We say that $\alpha_i \in N$ prefers some triple $t_1$ to another triple $t_2$ if $u_{\alpha_i}(t_1) > u_{\alpha_i}(t_2)$. An agent's preference between two distinct matchings depends only on that agent's partners in each matching, so given a matching $M$ we write $u_{\alpha_i}(M)$ as shorthand for $u_{\alpha_i}(M(\alpha_i))$. Let $V=\bigcup\limits_{\alpha_i \in N} \mathit{val}_{\alpha_i}$ be the collection of all valuation functions.

Suppose we have some pair $(N, V)$ and a matching $M$ involving the agents in $N$. We say that a triple $\{\alpha_{k_1}, \alpha_{k_2}, \alpha_{k_3} \}$ \textit{blocks} $M$ in $(N, V)$ if $u_{\alpha_{k_1}}(\{\alpha_{k_2}, \alpha_{k_3} \}) > u_{\alpha_{k_1}}(M), u_{\alpha_{k_2}}(\{\alpha_{k_1}, \alpha_{k_3} \}) > u_{\alpha_{k_2}}(M)$, and $u_{\alpha_{k_3}}(\{\alpha_{k_1}, \alpha_{k_2} \}) > u_{\alpha_{k_3}}(M)$. If no triple in $N$ blocks $M$ in $(N, V)$ then we say that $M$ is \textit{stable} in $(N, V)$. We say that $(N, V)$ \emph{contains} a stable matching if at least one matching exists in $(N, V)$ that is stable.

We now define the \emph{Three-Dimensional Stable Roommates problem with Additively Separable preferences} (3D-SR-AS). An instance of 3D-SR-AS is given by the pair $(N, V)$. The problem is to either find a stable matching in $(N, V)$ or report that no stable matching exists. In this paper we consider two different restrictions of this model. The first is when preferences are \emph{binary}, meaning $\mathit{val}_{\alpha_i}(\alpha_j) \in \{ 0, 1\}$ for any $\alpha_i, \alpha_j \in N$. The second is when preferences are also \emph{symmetric}, meaning $\mathit{val}_{\alpha_i}(\alpha_j)=\mathit{val}_{\alpha_j}(\alpha_i)$ for any $\alpha_i, \alpha_j \in N$.

Lemma~\ref{lem:3dsrsasbinblockerimprovement} illustrates a fundamental property of matchings in instances of 3D-SR-AS. We shall use it extensively in the proofs. \ifdefined\fullversion \else Throughout this paper the omitted proofs can be found in the full version \cite{fullversion3dsraspaper}.
\fi

\begin{restatable}{lem}{threedsrsasbinblockerimprovement}
\label{lem:3dsrsasbinblockerimprovement}
Given an instance $(N,V)$ of 3D-SR-AS, suppose that $M$ and $M'$ are matchings in $(N, V)$. Any triple that blocks $M'$ but does not block $M$ contains at least one agent $\alpha_i \in N$ where $u_{\alpha_i}(M') < u_{\alpha_i}(M)$.
\end{restatable}
\ifdefined \fullversion
\begin{proof}
Suppose that the triple $\{ \alpha_{k_1}, \alpha_{k_2}, \alpha_{k_3} \}$ blocks $M'$. It follows that $u_{\alpha_{k_1}}(\{\alpha_{k_2},\allowbreak \alpha_{k_3} \}) >\allowbreak u_{\alpha_{k_1}}(M')$, $u_{\alpha_{k_2}}(\{\alpha_{k_1}, \alpha_{k_3} \}) >\allowbreak u_{\alpha_{k_2}}(M')$, and $u_{\alpha_{k_3}}(\{\alpha_{k_1}, \alpha_{k_2} \}) >\allowbreak u_{\alpha_{k_3}}(M')$. Suppose for a contradiction that no $\alpha_p\in \{ \alpha_{k_1}, \alpha_{k_2},\allowbreak \alpha_{k_3} \}$ exists where $u_{\alpha_p}(M') < u_{\alpha_p}(M)$ and hence $u_{\alpha_{k_r}}(M') \geq u_{\alpha_{k_r}}(M)$ for $1 \leq r \leq 3$. It follows that $u_{\alpha_{k_1}}(\{\alpha_{k_2},\allowbreak \alpha_{k_3} \}) >\allowbreak u_{\alpha_{k_1}}(M)$, $u_{\alpha_{k_2}}(\{\alpha_{k_1},\allowbreak \alpha_{k_3} \}) >\allowbreak u_{\alpha_{k_2}}(M)$, and $u_{\alpha_{k_3}}(\{\alpha_{k_1},\allowbreak \alpha_{k_2} \}) >\allowbreak u_{\alpha_{k_3}}(M)$ and thus that $\{ \alpha_{k_1}, \alpha_{k_2}, \alpha_{k_3} \}$ blocks $M$, a contradiction.
\end{proof}
\fi

We also make an observation that unmatched agents may be arbitrarily matched if required. The proof follows from Lemma~\ref{lem:3dsrsasbinblockerimprovement}.

\begin{restatable}{prop}{completematching}
\label{prop:completematching}
Suppose we are given an instance $(N,V)$ of 3D-SR-AS. Suppose $|N|=3k+l$ where $k \geq 0$ and $0\leq l < 3$. If a stable matching $M$ exists in $(N, V)$ then without loss of generality we may assume that $|M|=k$.
\end{restatable}

Finally, some notes on notation: in this paper, we use $L = \langle \dots \rangle$ to construct an ordered list of elements $L$. If $L$ and $L'$ are lists then we write $L \cdot L'$ meaning the concatenation of $L'$ to the end of $L$. We also write $L_i$ to mean the $i\textsuperscript{th}$ element of list $L$, starting from $i=1$, and $e \in L$ to describe membership of an element $e$ in $L$. When working with sets of sets, we write $\bigcup S$ to mean $\bigcup_{T\in S} T$.

\section{General binary preferences}
\label{sec:generalbinary}
Let 3D-SR-AS-BIN be the restriction of 3D-SR-AS in which preferences are binary but need not be symmetric.  In this section we establish the $\NP$-completeness of deciding whether a stable matching exists, given an instance $(N, V)$ of 3D-SR-AS-BIN.

\ifdefined\fullversion 

Given an instance $(N, V)$ of 3D-SR-AS-BIN and a matching $M$, it is straightforward to test in $O(|N|^3)$ time if $M$ is stable in $(N, V)$. This shows that the decision version of 3D-SR-AS-BIN belongs to the class $\NP$.

We present a polynomial-time reduction from Partition Into Triangles (PIT), which is the following decision problem: ``Given a simple undirected graph $G=(W,E)$ where $W=\{ w_1, w_2, \dots, w_{3q} \}$ for some integer $q$, can the vertices of $G$ be partitioned into $q$ disjoint sets $X=\{X_1, X_2, \dots, X_q\}$, each set containing exactly three vertices, such that for each $X_p=\{w_i,w_j,w_k\}\in X$ all three of the edges $\{w_i,w_j\}$, $\{w_i,w_k\}$, and $\{w_j,w_k\}$ belong to $E$?'' PIT is $\NP$-complete \cite{GJ79}.

The reduction from PIT to 3D-SR-AS-BIN is as follows (see  Figure~\ref{fig:3d_sr_as_binary_reduction}). Unless otherwise specified assume that $\mathit{val}_{\alpha_i}(\alpha_j)=0$ for any $\alpha_i, \alpha_j \in N$. For each vertex $w_i \in W$ create agents $a_i^1, a_i^2, b_i$ in $N$. Then set:
\begin{itemize}[noitemsep]
     \item $\mathit{val}_{a_i^1}(a_i^2)=\mathit{val}_{a_i^1}(b_i)=1$
     \item $\mathit{val}_{a_i^2}(a_i^1)=\mathit{val}_{a_i^2}(b_i)=1$
     \item $\mathit{val}_{b_i}(a_i^1)=\mathit{val}_{b_i}(a_i^2)=1$ and $\mathit{val}_{b_i}(b_j)=1$ if $\{v_i, v_j\} \in E$.
\end{itemize}
Next, for each $r$ where $1 \leq r \leq 6q$ create $p_r^1,p_r^2,p_r^3,p_r^4,p_r^5$ in $N$. Then set:
\begin{itemize}[noitemsep]
    \item $\mathit{val}_{p_r^1}(p_r^2)=\mathit{val}_{p_r^1}(p_r^3)=\mathit{val}_{p_r^1}(p_r^5)=1$
    \item $\mathit{val}_{p_r^2}(p_r^3)=\mathit{val}_{p_r^2}(p_r^4)=\mathit{val}_{p_r^2}(p_r^1)=1$
    \item $\mathit{val}_{p_r^3}(p_r^4)=\mathit{val}_{p_r^3}(p_r^5)=\mathit{val}_{p_r^3}(p_r^2)=1$
    \item $\mathit{val}_{p_r^4}(p_r^5)=\mathit{val}_{p_r^4}(p_r^1)=\mathit{val}_{p_r^4}(p_r^3)=1$
    \item $\mathit{val}_{p_r^5}(p_r^1)=\mathit{val}_{p_r^5}(p_r^2)=\mathit{val}_{p_r^5}(p_r^4)=1$.
\end{itemize}
We shall refer to $\{ p_r^1,\dots,p_r^5 \}$ as the \emph{$r\textsuperscript{th}$ pentagadget}. Note that $|N|=39q$.

It is straightforward to show that the reduction runs in polynomial time. To prove that the reduction is valid we show that a partition into triangles $X=\{X_1,X_2,\dots,X_q\}$ exists in $G$ if and only if a stable matching $M$ exists in $(N,V)$. In Section~\ref{sec:3dsrasbinfirstdirection} we consider the first direction and show that if a partition into triangles $X=\{X_1,X_2,\dots,X_q\}$ exists in $G$ then a stable matching $M$ exists in $(N, V)$. In Section~\ref{sec:3dsrasbinseconddirection} we consider the second direction and show that if a stable matching $M$ exists in $(N,V)$ then a partition into triangles $X = \{ X_1, X_2, \dots, X_q\}$ exists in $G$. Note that the only instance discussed is $(N, V)$ and hence we shorten ``blocks $M$ in $(N, V)$'' to simply ``blocks $M$''.

\begin{figure}[t]
  \centering
    \ifdefined \fullversion 
    \newcommand\binaryreductionscale{0.9}
\else
    \newcommand\binaryreductionscale{0.82}
\fi

\begin{tikzpicture}
    \begin{scope}[every node/.style={circle,thick,draw,inner sep=0.8mm,minimum size=1.6mm}, scale=\binaryreductionscale]
        \node[draw=none, align=center] (Pc) at (3,0.3) {$\times 6q$ pentagadgets};
        \node (pr2) at (3,5) {$p_r^2$};
        \node (pr3) at (4.9021,3.6180) {$p_r^3$};
        \node (pr4) at (4.1756,1.3820) {$p_r^4$};
        \node (pr5) at (1.8244,1.3820) {$p_r^5$};
        \node (pr1) at (1.0979,3.6180) {$p_r^1$};
        
        \node (bi) at (10,3) {$b_i$};
        \node (ai1) at (8.2,1.5) {$a_i^2$};
        \node (ai2) at (8.2,4.5) {$a_i^1$};
        \node[draw=none] (bk1) at (12.2,1.35) {};
        \node[draw=none] (bk2) at (12.2,2.45) {};
        \node[draw=none] (bkdots) at (12.5,1.9) {\dots};
        \node[draw=none, label={[shift={(0.2, -0.5)}]$b_k$}] (bj) at (12.2,3.55) {};
        \node[draw=none, label={[shift={(0.2, -0.5)}]$b_j$}] (bk) at (12.2,4.65) {};
        \node[draw=none, text width=4.5cm, align=center] (aic) at (10.2,0.2) {for each vertex $w_i\in W$ where $N(w_i)=\{ w_j, w_k, \dots \}$};
    \end{scope}
    \begin{scope}
        \foreach \from/\to in {pr1/pr2, pr2/pr3, pr3/pr4, pr4/pr5, pr5/pr1}
            \draw [thick, darrow] (\from) -- (\to);
        \foreach \from/\to in {pr1/pr3, pr3/pr5, pr5/pr2, pr2/pr4, pr4/pr1}
            \draw [thick, ->-] (\from) -- (\to);
         \foreach \from/\to in {bi/bj, bi/bk, bi/bk1, bi/bk2}
            \draw [thick, --->] (\from) -- (\to);
        \foreach \from/\to in {bi/ai1, ai1/ai2, ai2/bi}
            \draw [thick, -->-] (\from) -- (\to);
        \foreach \from/\to in {ai1/bi, ai2/ai1, bi/ai2}
            \path [thick, -<--] (\from) -- (\to);
    \end{scope}
    \end{tikzpicture}
    \vspace*{-48pt}
    \caption{The reduction from PIT to 3D-SR-AS-BIN. Each vertex represents an agent. An arc is present from agent $\alpha_i$ to agent $\alpha_j$ if $\mathit{val}_{\alpha_i}(\alpha_j) = 1$.}
    \label{fig:3d_sr_as_binary_reduction}
\end{figure}
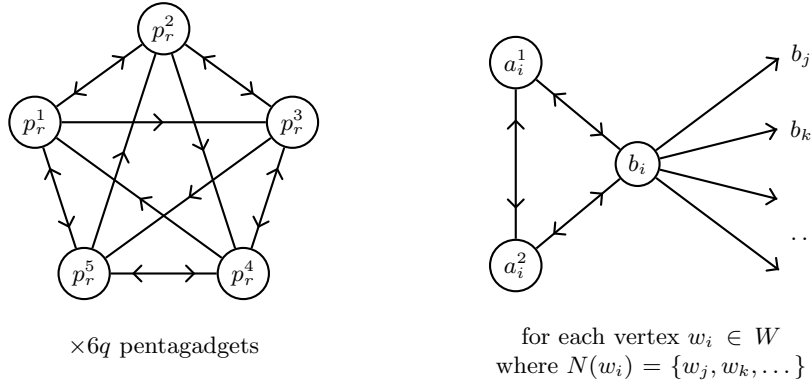

\subsection{Correctness of the reduction: first direction}\hfill
\label{sec:3dsrasbinfirstdirection}

\begin{lem}
\label{lem:3dsrasbinfirstdirection}
In the reduction, if a partition into triangles $X=\{X_1, X_2, \dots, X_q\}$ exists in $G$, then a stable matching exists in $(N, V)$.
\end{lem}
\begin{proof}
Suppose a partition into triangles $X = \{ X_1, X_2, \dots, X_q \}$ exists in $G$. We will construct a matching $M$ that is stable in $(N, V)$. For each triangle $X_p=\{ w_i, w_j, w_k \}\in W$, add $\{b_i, b_j, b_k\}$ to $M$. For each pentagadget with index $r$ where $1 \leq r \leq 6q$, add $\{p_r^1, p_r^2, p_r^3\}$ to $M$. This leaves agents $a_i^1$ and $a_i^2$ for each $1 \leq i \leq 3q$ and agents $p_r^4$ and $p_r^5$ for each $0 \leq r \leq 6q$. For each $1 \leq i \leq 3q$, add to $M$ the triples $\{a_i^1, p_{2i}^4, p_{2i}^5\}, \{a_i^2, p_{2i-1}^4, p_{2i-1}^5\}$. Now, in $M$:
\begin{itemize}[itemsep=0mm,parsep=1.2mm]
    \item For any pentagadget index $1 \leq r \leq 6q$:
        \begin{itemize}[itemsep=0mm,parsep=1.2mm]
            \item $u_{p_r^1}(M)=u_{p_r^2}(M)=2$, so neither $p_r^1$ nor $p_r^2$ belong to triples that block $M$.
            \item $u_{p_r^5}(M)=1$, so if $u_{p_r^5}$ belongs to a triple that blocks $M$ then that triple must contain two agents $\alpha_k, \alpha_l$ where $u_{p_r^5}(\{\alpha_k, \alpha_l\})=2$ and hence $\mathit{val}_{p_r^5}(\alpha_k) = \mathit{val}_{p_r^5}(\alpha_l) = 1$. Considering $\mathit{val}_{p_r^5}$, the only such agents are $p_r^1, p_r^4, p_r^2$. From above, neither $p_r^1$ nor $p_r^2$ belong to triples that block $M$. It follows that no such $\alpha_k, \alpha_l$ exist and hence $p_r^5$ does not belong to a triple that blocks $M$.
            \item $u_{p_r^4}(M)=1$, so if $u_{p_r^4}$ belongs to a triple that blocks $M$ then that triple must contain two agents $\alpha_k, \alpha_l$ where $u_{p_r^4}(\{\alpha_k, \alpha_l\})=2$ and hence $\mathit{val}_{p_r^4}(\alpha_k) = \mathit{val}_{p_r^4}(\alpha_l) = 1$. Considering $\mathit{val}_{p_r^4}$, the only such agents are $p_r^3, p_r^5, p_r^1$. From above, neither $p_r^1$ nor $p_r^5$ belong to triples that block $M$. It follows that no two such $\alpha_k, \alpha_l$ exist and hence $p_r^4$ also does not belong to a triple that blocks $M$.
            \item $u_{p_r^3}(M)=1$, so if $u_{p_r^3}$ belongs to a triple that blocks $M$ then that triple must contain two agents $\alpha_k, \alpha_l$ where $u_{p_r^3}(\{\alpha_k, \alpha_l\})=2$. Considering $\mathit{val}_{p_r^3}$, the only such agents are $p_r^2, p_r^4, p_r^5$. From above, these agents do not belong to triples that block $M$, so no such $\alpha_k, \alpha_l$ exist and hence $p_r^3$ also does not belong to a triple that blocks $M$.
        \end{itemize}
        \item $u_{b_i}(M)=2$ for any $1 \leq i \leq 3q$, so $b_i$ also does not belong to a triple that blocks $M$.
\end{itemize}
We have shown above that no pentagadget agent belongs to a triple that blocks $M$ and no $b_i$ for any $1 \leq i \leq 3q$ belongs to a triple that blocks $M$. The remaining possibility is that a blocking triple exists that contains three agents $\{a_i^{s_1}, a_j^{s_2}, a_k^{s_3}\}$ for some $1 \leq i,j,k \leq 3q$ and $s_1,s_2,s_3 \in \{1,2\}$. Since $a_i^{s_1}$ prefers this blocking triple to $M(a_i^{s_1})$, it must be that either $\mathit{val}_{a_i^{s_1}}(a_j^{s_2})=1$ or $\mathit{val}_{a_i^{s_1}}(a_k^{s_3})=1$, or both. For any $a_i^{s_1}$ where $s_1\in \{1, 2\}$, the only agent for which $\mathit{val}_{a_i^{s_1}}=1$ is $a_i^{3 - s_1}$. Assume then, without loss of generality, that the blocking triple contains $\{a_i^1, a_i^2, a_k^{s_4}\}$ for some $1 \leq i,k \leq 3q$ where $i\neq k$ and some ${s_4}\in \{1,2\}$. Note that $i\neq k$ because  ${s_4} \in \{1,2\}$. This leads to a contradiction, since $a_k^{s_4}$ must prefer this triple to $M(a_k^{s_4})$, but $u_{a_k^{s_4}}(\{a_i^1, a_i^2\})=0$ for any $1\leq k \leq 3q$ and any ${s_4}\in \{1,2\}$.
\end{proof}

\subsection{Correctness of the reduction: second direction}\hfill
\label{sec:3dsrasbinseconddirection}

In this section we assume that $M$ is a stable matching in $(N,V)$ where $|M|=|N|/3$ (by Proposition~\ref{prop:completematching}). We analyse its structure and construct a corresponding partition into triangles in $G$.

\begin{lem}
\label{lem:3dsrasbinpentagadgetagentsbelongtotwoagents}
For any $r$ where $1 \leq r \leq 6q$, the pentagadget agents $p_r^1, p_r^2, p_r^3,\allowbreak p_r^4, p_r^5$ belong to exactly two triples.
\end{lem}
\begin{proof}
Suppose for a contradiction that agents $p_r^1,\dots,p_r^5$ belong to four or five triples in $M$. It must be that three of these agents, say $p_r^{s_1}, p_r^{s_2}$, and $p_r^{s_3}$, belong to triples in $M$ each containing no other agents from the $r\textsuperscript{th}$ pentagadget. It follows that $u_{p_r^{s_1}}(M) = u_{p_r^{s_2}}(M) = u_{p_r^{s_3}}(M) = 0$. Each agent in the $r\textsuperscript{th}$ pentagadget assigns a valuation of one to exactly three other agents in the same pentagadget. It follows that either $\mathit{val}_{p_r^{s_1}}(p_r^{s_2}) = 1$ or $\mathit{val}_{p_r^{s_1}}(p_r^{s_3}) = 1$, or both. It then follows that $u_{p_r^{s_1}}(\{p_r^{s_2}, p_r^{s_3}\}) \geq 1$. A symmetric argument shows that $u_{p_r^{s_2}}(\{p_r^{s_1}, p_r^{s_3}\}) \geq 1$ and $u_{p_r^{s_3}}(\{p_r^{s_1}, p_r^{s_2}\}) \geq 1$. The triple $\{p_r^{s_1}, p_r^{s_2},p_r^{s_3}\}$ therefore blocks $M$, which is a contradiction.

Suppose then that the agents $p_r^1,\dots,p_r^5$ belong to three triples in $M$. Since there are five agents in $p_r^1,\dots,p_r^5$, there are two possibilities:
\begin{itemize}
    \item Two of the triples each contain exactly two agents in $\{ p_r^1,\dots,p_r^5 \}$ and the third triple contains exactly one agent in $\{ p_r^1,\dots,p_r^5 \}$. Due to the symmetry of the pentagadget, assume without loss of generality that $p_r^1$ is the sole agent from $p_r^1, \dots, p_r^5$ that belongs to the third triple. It follows that $u_{p_r^1}(M)=0$. The four agents $\{ p_r^2, \dots, p_r^5 \}$ each have at most one partner in $\{ p_r^1, \dots, p_r^5 \}$ in $M$. It follows that the utility in $M$ of each of these four agents is at most one. It follows that $\{ p_r^1, p_r^4, p_r^5 \}$ blocks $M$, since $u_{p_r^4}(\{p_r^5, p_r^1\}) = u_{p_r^5}(\{p_r^1, p_r^4\}) = 2$ and $u_{p_r^1}(\{p_r^4, p_r^5\}) = 1$. This is a contradiction.
    
    \item Two of the triples each contain exactly one agent in $\{ p_r^1, \dots, p_r^5 \}$ and the third triple contains exactly three agents in $\{ p_r^1, \dots, p_r^5 \}$. Suppose $p_r^{s_1}$ and $p_r^{s_2}$ are the two agents in the former two triples such that $s_1 \bmod 5 < s_2$. It follows that $u_{p_r^{s_1}}(M) = u_{p_r^{s_1}}(M) = 0$. Since there are five agents in $\{ p_r^1, \dots, p_r^5 \}$, there are two further possible cases:
    
    \begin{itemize} 
        \item Suppose $s_2 = (s_1 \bmod 5) + 1$. By the symmetry of the pentagadget, assume without loss of generality that $s_1=1$ and $s_2=2$. It follows that $\{ p_r^3, p_r^4, p_r^5 \} \in M$. Note that $u_{p_r^5}(M)=1$. The triple $\{p_r^5, p_r^1, p_r^2 \}$ blocks $M$ since $u_{p_r^5}(\{p_r^1, p_r^2\}) = u_{p_r^1}(\{p_r^2, p_r^5\}) = 2$ and $u_{p_r^2}(\{p_r^1, p_r^5\}) = 1$. This is a contradiction.
        \item Suppose $s_2 = ((s_1 + 1) \bmod 5) + 1$. By the symmetry of the pentagadget, assume without loss of generality that $s_1=1$ and $s_2=3$. It follows that $\{ p_r^2, p_r^4, p_r^5 \} \in M$. Note that $u_{p_r^2}(M)=1$. The triple $\{p_r^1, p_r^2, p_r^3 \}$ blocks $M$ since $u_{p_r^1}(\{p_r^2, p_r^3\}) = u_{p_r^2}(\{p_r^1, p_r^3\}) = 2$ and $u_{p_r^3}(\{p_r^1, p_r^2\}) = 1$. This is also a contradiction.
    \end{itemize}
\end{itemize}

In summary, we have shown that the five agents $p_r^1, \dots, p_r^5$ do not belong to three, four, or five different triples in $M$. It follows that these five agents belong to exactly two triples in $M$.
\end{proof}

\begin{lem}
\label{lem:3dsrasbinallairscores0}
For any $1 \leq i \leq 3q$ and any $s\in \{1,2\}$, $u_{a_i^s}(M)=0$.
\end{lem}
\begin{proof}
Consider an arbitrary pentagadget index $1 \leq {r_1} \leq 6q$.

By Lemma~\ref{lem:3dsrasbinpentagadgetagentsbelongtotwoagents}, the five agents $p_{r_1}^1, \dots, p_{r_1}^5$ belong to exactly two triples in $M$. It follows that one of these two triples contains exactly three agents in $\{ p_{r_1}^1, \dots, p_{r_1}^5 \}$ and the other triple contains the two remaining agents in $\{ p_{r_1}^1, \dots, p_{r_1}^5 \}$ as well as some third agent, say $\alpha_h$. Note that $u_{\alpha_h}(M) = 0$.

Suppose ${\alpha_h} = p_{r_2}^t$ for some $1 \leq {r_2} \leq 6q$ and some $1 \leq t \leq 5$. It follows that ${r_1} \neq {r_2}$. Since the triple $M(p_{r_2}^t)$ contains $p_{r_2}^t$ and two agents in $\{ p_{r_1}^1, \dots, p_{r_1}^5 \}$, it follows that the four agents in $\{ p_{r_2}^1, \dots, p_{r_2}^5\} \setminus \{ p_{r_2}^t \}$ belong to at least two triples in $M$. In total, the five agents $p_{r_2}^1, \dots, p_{r_2}^5$ belong to three or more triples, which contradicts Lemma~\ref{lem:3dsrasbinpentagadgetagentsbelongtotwoagents}. It follows that ${\alpha_h} \neq p_{r_2}^t$ for any $1 \leq {r_2} \leq 6q$ and any $1 \leq t \leq 5$.

Suppose then that ${\alpha_h} = b_j$ for some $1 \leq j \leq 3q$. Consider $a_j^1$ and $a_j^2$ and their respective valuation functions. Since $a_j^1 \notin M(b_j)$ and $a_j^2 \notin M(b_j)$, it follows that $u_{a_j^1}(M) \leq 1$ and $u_{a_j^2}(M) \leq 1$. Recalling that $u_{\alpha_h}(M) = u_{b_j}(M) = 0$, it follows that $\{b_j, a_j^1, a_j^2\}$ blocks $M$, since $u_{b_j}(\{a_j^1, a_j^2\}) = u_{a_j^1}(\{b_j, a_j^2\}) = u_{a_j^2}(\{b_j, a_j^1\}) = 2$. This is a contradiction. It follows that ${\alpha_h} \neq b_j$ for any $1 \leq j \leq 3q$.

It remains that ${\alpha_h} = a_i^s$ for some $1 \leq i \leq 3q$ and some $s \in \{1,2\}$. The initial selection of ${r_1}$ is arbitrary, and each $1 \leq {r_1} \leq 6q$ identifies a unique pentagadget $\{ p_{r_1}^1, \dots, p_{r_1}^5 \}$ and therefore a unique agent ${\alpha_h}$. There are therefore $6q$ unique agents $\alpha_h = a_i^s$ where $u_{a_i^s}(M)=0$ for some $1 \leq i \leq 3q$ and some $s \in \{1,2\}$. This shows that $u_{a_i^s}(M)=0$ for every $1 \leq i \leq 3q$ and every $s \in \{1,2\}$.
\end{proof}

\begin{lem}
\label{lem:3dsrasbinallbiscores2}
$u_{b_i}(M)=2$ for any $1 \leq i \leq 3q$.
\end{lem}
\begin{proof}
Suppose not and there exists some $1 \leq i \leq 3q$ such that $u_{b_i}(M) < 2$. Lemma~\ref{lem:3dsrasbinallairscores0} shows that $u_{a_i^1}(M) = u_{a_i^2}(M) = 0$. Considering the valuation functions of $a_i^1$, $a_i^2$, and $b_i$, it can be seen that $u_{b_i}(\{a_i^1, a_i^2\}) = u_{a_i^1}(\{b_i, a_i^2\}) = u_{a_i^2}(\{b_i, a_i^1\}) = 2$. It follows that $\{b_i, a_i^1, a_i^2\}$ blocks $M$, which is a contradiction.
\end{proof}

\begin{lem}
\label{lem:3dsrasbinallbiintriplestogether}
For any $b_i$ where $1 \leq i \leq 3q$, the triple $M(b_i)$ comprises $\{b_i, b_j, b_k\}$ for some $1 \leq j,k \leq 3q$ where $\{w_i, w_j\}, \{w_j, w_k\} \in E$.
\end{lem}
\begin{proof}
Lemma~\ref{lem:3dsrasbinallbiscores2} shows that $u_{b_i}(M)=2$. Suppose $M(b_i)=\{ b_i, \alpha_k, \alpha_l \}$ for some $\alpha_k, \alpha_l\in N$. Since $u_{b_i}(M)=2$, It must be that $\mathit{val}_{b_i}(\alpha_k)=1$ and hence either $\alpha_k = a_i^s$ for some $s\in \{1,2\}$ or $\alpha_k = b_j$ for some $1\leq j \leq 3q$ where $\{w_i, w_j\}\in E$.

Suppose first that $\alpha_k = a_i^s$ for some $s\in \{1,2\}$. Since $\mathit{val}_{a_i^s}(b_i)=1$, it follows that $u_{a_i^s}(M) \geq 1$, which contradicts Lemma~\ref{lem:3dsrasbinallairscores0}. It follows that $\alpha_k = b_j$ for some $1\leq j \leq 3q$ where $\{ w_i, w_j \} \in E$. Similarly, it can be shown that $\alpha_l = b_k$ for some $1\leq k \leq 3q$ where $\{ w_i, w_j \} \in E$. It follows that $M(b_i) = \{ b_i, b_j, b_k \}$ for some $1 \leq j,k \leq 3q$ where $\{w_i, w_j\}, \{w_j, w_k\} \in E$.
\end{proof}

\begin{lem}
\label{lem:3dsrasbinpitexists}
A partition into triangles exists in $G$.
\end{lem}
\begin{proof}
Lemma~\ref{lem:3dsrasbinallbiintriplestogether} shows that for an arbitrary $b_i$ where $1 \leq i \leq 3q$, $M(b_i)$ comprises $\{b_i, b_j, b_k\}$ for some $1 \leq j,k \leq 3q$ where $\{w_i, w_j\}\in E$ and $\{w_i, w_k\}\in E$. It follows that there are exactly $q$ triples in $M$ each containing three agents $\{b_i, b_j, b_k\}$, where the three corresponding vertices $w_i, w_j, w_k$ are pairwise adjacent in $G$. From these triples of pairwise adjacent vertices, a partition into triangles $X$ can be easily constructed.
\end{proof}

\subsection{Conclusion}\hfill

\begin{restatable}{theorem}{threedsrasbinexistence}
Given an instance of 3D-SR-AS-BIN, the problem of deciding whether a stable matching exists is $\NP$-complete. The result holds even if each agent must be matched.
\end{restatable}
\begin{proof}
We have already shown that the decision version of 3D-SR-AS-BIN belongs to $\NP$. We presented a polynomial time reduction from Partition Into Triangles (PIT) to 3D-SR-AS-BIN. If a partition into triangles exists in the PIT instance $G=(W, E)$ then a stable matching $M$ exists in $(N, V)$ where $|M|=|N|/3$ (Lemma~\ref{lem:3dsrasbinfirstdirection}). If a stable matching $M$ exists in $(N, V)$ where $|M|=|N|/3$ then a partition into triangles exists in $G$ (Lemma~\ref{lem:3dsrasbinpitexists}).
\end{proof}

\else

\begin{figure}[t!]
  \centering
    \ifdefined \fullversion 
    \newcommand\binaryreductionscale{0.9}
\else
    \newcommand\binaryreductionscale{0.82}
\fi

\begin{tikzpicture}
    \begin{scope}[every node/.style={circle,thick,draw,inner sep=0.8mm,minimum size=1.6mm}, scale=\binaryreductionscale]
        \node[draw=none, align=center] (Pc) at (3,0.3) {$\times 6q$ pentagadgets};
        \node (pr2) at (3,5) {$p_r^2$};
        \node (pr3) at (4.9021,3.6180) {$p_r^3$};
        \node (pr4) at (4.1756,1.3820) {$p_r^4$};
        \node (pr5) at (1.8244,1.3820) {$p_r^5$};
        \node (pr1) at (1.0979,3.6180) {$p_r^1$};
        
        \node (bi) at (10,3) {$b_i$};
        \node (ai1) at (8.2,1.5) {$a_i^2$};
        \node (ai2) at (8.2,4.5) {$a_i^1$};
        \node[draw=none] (bk1) at (12.2,1.35) {};
        \node[draw=none] (bk2) at (12.2,2.45) {};
        \node[draw=none] (bkdots) at (12.5,1.9) {\dots};
        \node[draw=none, label={[shift={(0.2, -0.5)}]$b_k$}] (bj) at (12.2,3.55) {};
        \node[draw=none, label={[shift={(0.2, -0.5)}]$b_j$}] (bk) at (12.2,4.65) {};
        \node[draw=none, text width=4.5cm, align=center] (aic) at (10.2,0.2) {for each vertex $w_i\in W$ where $N(w_i)=\{ w_j, w_k, \dots \}$};
    \end{scope}
    \begin{scope}
        \foreach \from/\to in {pr1/pr2, pr2/pr3, pr3/pr4, pr4/pr5, pr5/pr1}
            \draw [thick, darrow] (\from) -- (\to);
        \foreach \from/\to in {pr1/pr3, pr3/pr5, pr5/pr2, pr2/pr4, pr4/pr1}
            \draw [thick, ->-] (\from) -- (\to);
         \foreach \from/\to in {bi/bj, bi/bk, bi/bk1, bi/bk2}
            \draw [thick, --->] (\from) -- (\to);
        \foreach \from/\to in {bi/ai1, ai1/ai2, ai2/bi}
            \draw [thick, -->-] (\from) -- (\to);
        \foreach \from/\to in {ai1/bi, ai2/ai1, bi/ai2}
            \path [thick, -<--] (\from) -- (\to);
    \end{scope}
    \end{tikzpicture}
    \vspace*{-48pt}
    \caption{The reduction from PIT to 3D-SR-AS-BIN. Each vertex represents an agent. An arc is present from agent $\alpha_i$ to agent $\alpha_j$ if $\mathit{val}_{\alpha_i}(\alpha_j) = 1$.}
    \label{fig:3d_sr_as_binary_reduction}
\end{figure}

\begin{restatable}{theorem}{threedsrasbinexistence}
Given an instance of 3D-SR-AS-BIN, the problem of deciding whether a stable matching exists is $\NP$-complete. The result holds even if each agent must be matched.
\end{restatable}
\begin{proofsketch}
Given an instance $(N, V)$ of 3D-SR-AS-BIN and a matching $M$, it is straightforward to test in $O(|N|^3)$ time if $M$ is stable in $(N, V)$. This shows that the decision version of 3D-SR-AS-BIN belongs to the class $\NP$.

We present a polynomial-time reduction from Partition Into Triangles (PIT), which is the following decision problem: ``Given a simple undirected graph $G=(W,E)$ where $W=\{ w_1, w_2, \dots, w_{3q} \}$ for some integer $q$, can the vertices of $G$ be partitioned into $q$ disjoint sets $X=\{X_1, X_2, \dots, X_q\}$, each set containing exactly three vertices, such that for each $X_p=\{w_i,w_j,w_k\}\in X$ all three of the edges $\{w_i,w_j\}$, $\{w_i,w_k\}$, and $\{w_j,w_k\}$ belong to $E$?'' PIT is $\NP$-complete \cite{GJ79}.

The reduction from PIT to 3D-SR-AS-BIN is as follows (see  Figure~\ref{fig:3d_sr_as_binary_reduction}). Unless otherwise specified assume that $\mathit{val}_{\alpha_i}(\alpha_j)=0$ for any $\alpha_i, \alpha_j \in N$. For each vertex $w_i \in W$ create agents $a_i^1, a_i^2, b_i$ in $N$. Then set $\mathit{val}_{a_i^1}(a_i^2)=\mathit{val}_{a_i^1}(b_i)=1$, $\mathit{val}_{a_i^2}(a_i^1)=\mathit{val}_{a_i^2}(b_i)=1$, $\mathit{val}_{b_i}(a_i^1)=\mathit{val}_{b_i}(a_i^2)=1$, and $\mathit{val}_{b_i}(b_j)=1$ if $\{w_i, w_j\} \in E$ for any $w_j\in N \setminus \{ w_i \}$. Next, for each $r$ where $1 \leq r \leq 6q$ create $p_r^1,p_r^2,p_r^3,p_r^4,p_r^5$ in $N$. Then set $\mathit{val}_{p_r^1}(p_r^2) = \mathit{val}_{p_r^1}(p_r^3) = \mathit{val}_{p_r^1}(p_r^5)=1$, $\mathit{val}_{p_r^2}(p_r^3)=\mathit{val}_{p_r^2}(p_r^4)=\mathit{val}_{p_r^2}(p_r^1)=1$, $\mathit{val}_{p_r^3}(p_r^4)=\mathit{val}_{p_r^3}(p_r^5)=\mathit{val}_{p_r^3}(p_r^2)=1$, $\mathit{val}_{p_r^4}(p_r^5)=\mathit{val}_{p_r^4}(p_r^1)=\mathit{val}_{p_r^4}(p_r^3)=1$, and $\mathit{val}_{p_r^5}(p_r^1)=\mathit{val}_{p_r^5}(p_r^2)=\mathit{val}_{p_r^5}(p_r^4)=1$. We shall refer to $\{ p_r^1,\dots,p_r^5 \}$ as the \emph{$r\textsuperscript{th}$ pentagadget}. Note that $|N|=39q$.
In the full proof of this theorem, contained in \cite{fullversion3dsraspaper}, we show that a partition into triangles $X$ exists in $G=(W,E)$ if and only if a stable matching $M$ exists in $(N, V)$ where $|M|=|N|/3$.\qed
\end{proofsketch}

\fi

\section{Symmetric binary preferences}
\label{sec:symmetricbinary}
Consider the restriction of 3D-SR-AS in which preferences are binary and symmetric, which we call \emph{3D-SR-SAS-BIN}. In this section we show that every instance of 3D-SR-SAS-BIN admits a stable matching. We give a step-by-step constructive proof of this result between Sections~\ref{sec:3dsrsasbin_prelims} -- \ref{sec:3dsrsasbin_findingarbitrarystablematching}, leading to an $O(|N|^3)$ algorithm for finding a stable matching. In Section~\ref{sec:3dsrsasbin_utilitarianwelfare} we consider an optimisation problem related to 3D-SR-SAS-BIN.
\subsection{Preliminaries}
\label{sec:3dsrsasbin_prelims}
An instance $(N, V)$ of 3D-SR-SAS-BIN corresponds to a simple undirected graph $G=(N,E)$ where $\{ \alpha_i, \alpha_j \} \in E$ if $\mathit{val}_{\alpha_i}(\alpha_j)=1$, which we refer to as the \emph{underlying graph}. 

We introduce a restricted type of matching called a \emph{$P$\nobreakdash-matching}. Recall that by definition, $M(\alpha_p)=\varnothing$ implies that $u_{\alpha_p}(M)=0$ for any $\alpha_p \in N$ in an arbitrary matching $M$. We say that a matching $M$ in $(N, V)$ is a \emph{$P$\nobreakdash-matching} if $M(\alpha_p) \neq \varnothing$ implies $u_{\alpha_p}(M) > 0$.

It follows that a $P$\nobreakdash-matching corresponds to a $\{ K_3, P_3 \}$-packing in the underlying graph \cite{KH83}. Note that any triple in a $P$\nobreakdash-matching $M$ must contain some agent with utility two. A \emph{stable $P$\nobreakdash-matching} is a $P$\nobreakdash-matching that is also stable. We will eventually show that any instance of 3D-SR-SAS-BIN contains a stable $P$\nobreakdash-matching.

In an instance $(N, V)$ of 3D-SR-SAS-BIN, a \emph{triangle} comprises three agents $\alpha_{m_1}, \alpha_{m_2},\allowbreak \alpha_{m_3}$ such that $\mathit{val}_{\alpha_{m_1}}(\alpha_{m_2}) = \mathit{val}_{\alpha_{m_2}}(\alpha_{m_3}) = \mathit{val}_{\alpha_{m_3}}(\alpha_{m_1}) = 1$. If $(N, V)$ contains no triangle then we say it is \emph{triangle-free}. If $(N, V)$ is not triangle-free then it can be reduced by successively removing three agents that belong to a triangle until it is triangle-free. This operation corresponds to removing a \emph{maximal triangle packing} (see \cite{CHATAIGNER20091396,KH83}) in the underlying graph and can be performed in $O(|N|^3)$ time. The resulting instance is triangle-free. We summarise this observation in the following lemma.

\begin{restatable}{lem}{threedsrsasbintrianglefree}
\label{lem:threedsrsasbintrianglefree}
Given an instance $(N, V)$ of 3D-SR-SAS-BIN, we can identify an instance $(N', V')$ of 3D-SR-SAS-BIN and a set of triples $M_{\triangle}$ in $O(|N|^3)$ time such that $(N', V')$ is triangle-free, $|N'|\leq |N|$, and if $M$ is a stable $P$\nobreakdash-matching in $(N', V')$ then $M' = M \cup M_{\triangle}$ is a stable $P$\nobreakdash-matching in $(N, V)$.
\end{restatable}
\ifdefined \fullversion
\begin{proof}
The set $M_{\triangle}$ corresponds to a \emph{maximal triangle packing} in the underlying graph \cite{CHATAIGNER20091396}, and thus can be found in $O(|N|^3)$ time. Let $N' = N \setminus \bigcup M_{\triangle}$. Construct $V'$ accordingly. Since each triple in  $M_{\triangle}$ corresponds to a triangle, any agent belonging to a triple in $M_{\triangle}$ gains utility two. It follows that if $M$ is a stable $P$\nobreakdash-matching in $(N', V')$ then $M'=M\cup M_{\triangle}$ is a stable $P$\nobreakdash-matching in $(N, V)$.
\end{proof}
\fi

\subsection{Repairing a \texorpdfstring{$P$}{P}-matching in a triangle-free instance}
\label{sec:3dsrsasbin_specialcase_algorithmsection}

In this section we consider an arbitrary triangle-free instance $(N, V)$ of 3D-SR-SAS-BIN. Since the only instance referred to in this section is $(N, V)$ so here we shorten ``is stable in $(N, V)$'' to ``is stable'', or similar. 

We first define a special type of $P$\nobreakdash-matching which is `repairable'. We then present Algorithm~\texttt{repair} (Algorithm~\ref{alg:3dsrsasbin_almostthere_algo}), which, given $(N,V)$ and a `repairable' $P$\nobreakdash-matching $M$, constructs a new $P$\nobreakdash-matching $M'$ that is stable. We shall see in the next section how this relates to a more general algorithm that, given a triangle-free instance, constructs a $P$\nobreakdash-matching that is stable in that instance.

Given a triangle-free instance $(N, V)$, we say a $P$\nobreakdash-matching $M$ is \emph{repairable} if it is not stable and there exists exactly one $\alpha_i \in N$ where $u_{\alpha_i}(M)=0$ and any triple that blocks $M$ comprises $\{ \alpha_i, \alpha_{j_1}, \alpha_{j_2}\}$ for some $\alpha_{j_1}, \alpha_{j_2}\in N$ where $u_{\alpha_{j_1}}(M)=1$, $u_{\alpha_{j_2}}(M)=0$, and $\mathit{val}_{\alpha_i}(\alpha_{j_1})=\mathit{val}_{\alpha_{j_1}}(\alpha_{j_2})=1$.


\ifdefined \fullversion

We now provide some intuition behind Algorithm~\texttt{repair} and refer the reader to Figure~\ref{fig:3d_sr_sas_bin_algorithm_7_cases_original_case}. Recall that the overall goal of the algorithm is to construct a stable $P$\nobreakdash-matching $M'$. Since the given $P$\nobreakdash-matching $M$ is repairable, our aim will be to modify $M$ such that $u_{\alpha_i}(M')\geq 1$ while ensuring that no three agents that are ordered to different triples in $M'$ block $M'$. The stability of the constructed $P$\nobreakdash-matching $M'$ then follows. We note that one way to achieve this aim would be to construct $M'$ such that $u_{\alpha_i}(M')\geq 1$ and $u_{\alpha_p}(M')\geq u_{\alpha_p}(M)$ for any $\alpha_p \in N\setminus \{ \alpha_i \}$, from which it follows that $M'$ is stable.

The algorithm begins by selecting some triple $\{ \alpha_i, \alpha_{j_1}, \alpha_{j_2} \}$ that blocks $M$. The two agents in $M(\alpha_{j_1}) \setminus \{ \alpha_{j_1} \}$ are labelled $\alpha_{j_3}$ and $\alpha_{j_4}$. We present two example cases in which it is possible to construct a stable $P$\nobreakdash-matching. 

First, suppose there exists some $\alpha_{z_1}$ where $\mathit{val}_{\alpha_{j_3}}(\alpha_{z_1})=1$ and $u_{\alpha_{z_1}}(M)=0$.  Construct $M'$ from $M$ by removing $\{ \alpha_{j_1}, \alpha_{j_2}, \alpha_{j_3} \}$ and adding $\{ \alpha_i, \alpha_{j_1}, \alpha_{j_2} \}$ and $\{ \alpha_{j_3}, \alpha_{j_4}, \alpha_{z_1} \}$. Now, $u_{\alpha_i}(M')=1$ and $u_{\alpha_p}(M')\geq u_{\alpha_p}(M)$ for any $\alpha_p \in N\setminus \{ \alpha_i \}$. It follows by Lemma~\ref{lem:3dsrsasbinblockerimprovement} that $M'$ is stable. Second, suppose there exists no such $\alpha_{z_1}$ but there exists some $\alpha_{z_2}$ where $\mathit{val}_{\alpha_{j_4}}(\alpha_{z_2})=1$ and $u_{\alpha_{z_2}}(M)=0$. Now construct $M'$ from $M$ by removing $\{ \alpha_{j_1}, \alpha_{j_2}, \alpha_{j_3} \}$ and adding $\{ \alpha_i, \alpha_{j_1}, \alpha_{j_2} \}$ and $\{ \alpha_{j_3}, \alpha_{j_4}, \alpha_{z_2} \}$. Note that $u_{\alpha_i}(M')=1$ and $u_{\alpha_p}(M')\geq u_{\alpha_p}(M)$ for any $\alpha_p \in N\setminus \{ \alpha_i, \alpha_{j_3} \}$. It can be shown that $\alpha_{j_3}$ does not belong to a triple that blocks $M'$ since no $\alpha_{z_1}$ exists as described. It follows again by Lemma~\ref{lem:3dsrsasbinblockerimprovement} that $M'$ is stable.

Generalising these two example cases, the algorithm constructs a list $S$ of agents, which initially comprises $\langle \alpha_{j_1}, \alpha_{j_3},\allowbreak \alpha_{j_4} \rangle$. The list $S$ has length $3c$ for some $c\geq 1$, where $\{S_{3c-2}, S_{3c-1}, S_{3c} \} \in M$ and $\mathit{val}_{S_p}(S_{p+1})=1$ for each $p$ ($1 \leq p < 3c$). The list $S$ therefore corresponds to a path in the underlying graph. In each iteration of the main loop, three agents belonging to some triple in $M$ are appended to the end of $S$. The loop continues until $S$ satisfies at least one of six specific conditions (shown in the first \texttt{if/else} statement). We show that eventually at least one of these conditions must hold. The algorithm then constructs $M'$. The exact construction of $M'$ depends on which condition(s) caused the main loop to terminate. Two of these conditions, and the corresponding constructions of $M'$, generalise the existence of $\alpha_{z_1}$ and $\alpha_{z_2}$ as described in the example cases.

\begin{algorithm}[b!]
\textbf{Input:} a triangle-free instance $(N,V)$ of 3D-SR-SAS-BIN, repairable $P$-matching $M$ in $(N, V)$ (Section~\ref{sec:3dsrsasbin_specialcase_algorithmsection}) with some such $\alpha_i\in N$.\\
\textbf{Output:} stable $P$-matching $M'$ in $(N,V)$
\smallskip

\begin{algorithmic}
\caption{Algorithm \texttt{repair} \label{alg:3dsrsasbin_almostthere_algo}} 

\State $\{ \alpha_{j_1}, \alpha_{j_2} \} \gets$ some $\alpha_{j_1}, \alpha_{j_2}\in N$ where $\{ \alpha_i, \alpha_{j_1}, \alpha_{j_2}\}$ blocks $M$ and $u_{\alpha_{j_1}}(M)=1$

\State $\{ \alpha_{j_3}, \alpha_{j_4} \} \gets M(\alpha_{j_1}) \setminus \{ \alpha_{{j_1}} \}$ where $u_{\alpha_{{j_3}}}(M)=2$

\smallskip

\State $S \gets \langle \alpha_{j_1}, \alpha_{j_3}, \alpha_{j_4} \rangle$
\State $c \gets 1$
\State $b \gets 0$

\State $\alpha_{z_1}, \alpha_{z_2}, \alpha_{y_1}, \alpha_{y_2}, \alpha_{w_1} \gets \bot$

\smallskip

\WhileNoDo{\textbf{true}}

\State $\alpha_{z_1} \gets$ some $\alpha_{z_1}\in N \setminus \{ \alpha_i \}$ where $\mathit{val}_{\alpha_{z_1}}(S_{3c-1})=1 $ and $u_{\alpha_{z_1}}(M)=0$, else $\bot$

\smallskip

\State $\alpha_{z_2} \gets$ some $\alpha_{z_2}\in N \setminus \{ \alpha_i, \alpha_{j_2} \}$ where $\mathit{val}_{\alpha_{z_2}}(S_{3c})=1 $ and $ u_{\alpha_{z_2}}(M)=0$, else $\bot$

\smallskip

\State $\alpha_{y_1} \gets$ some $\alpha_{y_1}\in N$ where $\mathit{val}_{S_{3c}}(\alpha_i)=\mathit{val}_{\alpha_{y_1}}(\alpha_i)=1 $ and $u_{\alpha_{y_1}}(M)=0$, else $\bot$

\smallskip

\State $\alpha_{y_2} \gets$ some $\alpha_{y_2}\in N$ where $\mathit{val}_{S_{3c}}(\alpha_{j_2})=\mathit{val}_{\alpha_{y_2}}(\alpha_{j_2})=1 $ and $u_{\alpha_{y_2}}(M)=0$, else $\bot$

\smallskip

\State $b \gets$ some $1 \leq b < c$ where $\mathit{val}_{S_{3b}}(\alpha_{j_2})=\mathit{val}_{S_{3c}}(S_{3b})=1$, else $0$

\smallskip

\State $\alpha_{w_1} \gets$ some $\alpha_{w_1}\in N$ where $\mathit{val}_{S_{3c}}(\alpha_{w_1})=1$, $u_{\alpha_{w_1}}(M)=1 $ and $\alpha_{w_1} \notin S$
\State $\hphantom{\alpha_{w_1}} $ and there exists some $\alpha_{z_3}\in N \setminus \{ \alpha_i \}$ where $\mathit{val}_{\alpha_{w_1}}(\alpha_{z_3})=1 $ and $ u_{\alpha_{z_3}}(M)=0$, \State $\hphantom{\alpha_{w_1}}$ else $\bot$

\smallskip

\If{$\alpha_{z_1} \neq \bot$ or $\alpha_{z_2} \neq \bot$ or $\alpha_{y_1} \neq \bot$ or $\alpha_{y_2} \neq \bot$ or $b>0$ or $\alpha_{w_1}=\bot$}

    \State \textbf{break}

\Else

\State $\{ \alpha_{w_2}, \alpha_{w_3} \} \gets M(\alpha_{w_1}) \setminus \{ \alpha_{{w_1}} \}$ where $u_{\alpha_{{w_2}}}(M)=2$


\State $S \gets S \cdot \langle \alpha_{{w_1}}, \alpha_{{w_2}}, \alpha_{{w_3}}  \rangle$
\State $c \gets c + 1$

\EndIf
\State \textbf{end if}

\EndWhile
\State \textbf{end while}
\medskip

\medskip
\State \emph{continued overleaf}
\medskip
\algstore{alg:3dsrsasbin_almostthere_algo}
\end{algorithmic}
\end{algorithm}

\begin{algorithm}[t!]
\begin{algorithmic}
\algrestore{alg:3dsrsasbin_almostthere_algo}
\addtocounter{algorithm}{-1} 
\caption{Algorithm \texttt{repair}}

\smallskip

\State \emph{continued from previous page}

\medskip

\If{$\alpha_{z_1} \neq \bot $ and $ \alpha_{z_1} \neq \alpha_{j_2}$} 
\LineComment{Case 1}

\State $M_{\textrm{S}} \gets \{\{ \alpha_i, \alpha_{j_1}, \alpha_{j_2} \}\} \cup \bigcup\limits_{1 \leq d < c}\{ \{ S_{3d-1}, S_{3d}, S_{3d+1} \} \} \cup \{\{ \alpha_{z_1}, S_{3c-1}, S_{3c} \}\}$

\smallskip

\ElsIf{$\alpha_{z_2} \neq \bot$} 
\LineComment{Case 2}

\State $M_{\textrm{S}} \gets \{\{ \alpha_i, \alpha_{j_1}, \alpha_{j_2} \}\} \cup \bigcup\limits_{1 \leq d < c}\{ \{ S_{3d-1}, S_{3d}, S_{3d+1} \} \} \cup \{\{ S_{3c-1}, S_{3c}, \alpha_{z_2} \}\}$

\smallskip

\ElsIf{$\alpha_{z_1} \neq \bot $ and $\alpha_{z_1} = \alpha_{j_2}$} %
\LineComment{Case 3}

\State $\alpha_{z_4} \gets$ some $\alpha_{z_4} \in N\setminus \{ \alpha_i, \alpha_{j_2} \}$ where $\mathit{val}_{S_{3c-2}}(\alpha_{z_4})=1 $ and $ u_{\alpha_{z_4}}(M)=0$
\smallskip
\State $M_{\textrm{S}} \gets \{\{ \alpha_i, \alpha_{j_1}, \alpha_{j_3} \}\} \cup \bigcup\limits_{1 \leq d < c - 1}\{ \{ S_{3d}, S_{3d+1}, S_{3d+2} \} \} \cup \{\{ S_{3c-3}, S_{3c-2}, \alpha_{z_4} \}\}$
\State $\hphantom{M' \gets } \cup \{\{ S_{3c-1}, S_{3c}, \alpha_{j_2} \}\}$

\smallskip

\ElsIf{$\alpha_{y_1} \neq \bot$} 
\LineComment{Case 4}
\State $M_{\textrm{S}} \gets \{\{ \alpha_{j_2}, \alpha_{j_1}, \alpha_{j_3} \}\} \cup \bigcup\limits_{1 \leq d < c}\{ \{ S_{3d}, S_{3d+1}, S_{3d+2} \} \} \cup \{\{ S_{3c}, \alpha_i, \alpha_{y_1} \}\}$

\smallskip

\ElsIf{$\alpha_{y_2} \neq \bot$} 
 \LineComment{Case 5}
\State $M_{\textrm{S}} \gets \{\{ \alpha_{i}, \alpha_{j_1}, \alpha_{j_3} \}\} \cup \bigcup\limits_{1 \leq d < c}\{ \{ S_{3d}, S_{3d+1}, S_{3d+2} \} \} \cup \{\{ S_{3c}, \alpha_{j_2}, \alpha_{y_2} \}\}$

\smallskip

\ElsIf{$b>0$} 

\LineComment{Case 6}
\State $\alpha_{z_5} \gets$ some $\alpha_{z_5} \in N\setminus \{ \alpha_i, \alpha_{j_2} \}$ where $\mathit{val}_{S_{3b+1}}(\alpha_{z_3})=1$ and $u_{\alpha_{z_3}}(M)=0$

\smallskip

\State $M_{\textrm{S}} \gets \{\{ \alpha_i, \alpha_{j_1}, \alpha_{j_3} \}\} 
\cup \bigcup\limits_{1 \leq d < b}\{ \{ S_{3d}, S_{3d+1}, S_{3d+2} \} \}
\cup \{\{ \alpha_{z_4}, S_{3b+1}, S_{3b+2} \}\}$ \State $\hphantom{M' \gets } \cup \bigcup\limits_{b + 1 \leq d < c}\{ \{ S_{3d}, S_{3d+1}, S_{3d+2} \} \}
\cup \{\{ S_{3c}, S_{3b}, \alpha_{j_2} \}\}$


\Else 

\LineComment{Case 7. Note that $\alpha_{w_1}=\bot$.}
\State $M_{\textrm{S}} \gets \{\{ \alpha_{i}, \alpha_{j_1}, \alpha_{j_3} \}\} \cup \bigcup\limits_{1 \leq d < c}\{ \{ S_{3d}, S_{3d+1}, S_{3d+2} \} \}$

\EndIf
\State \textbf{end if}

\State \Return $M' = M_{\textrm{S}} \cup \{ r \in M \,|\, r \cap S = \varnothing \}$

\medskip
\end{algorithmic}
\end{algorithm}

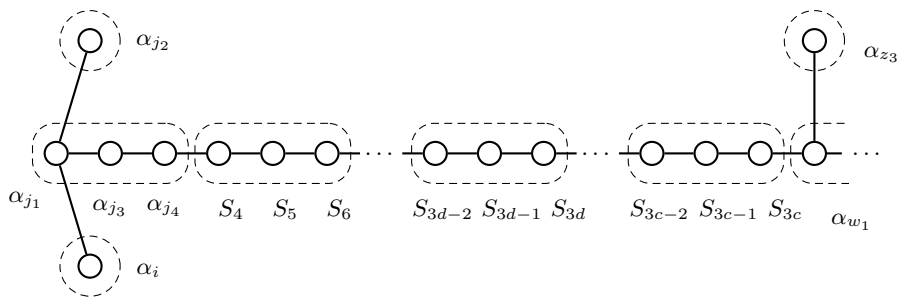
\begin{figure}[t]
    \centering
    \begin{tikzpicture}
\centering
\begin{scope}[every node/.style={circle,draw, minimum size=2.4mm}, xscale=0.9]
    \node[thick, circle, label={[shift={(0.9, -0.7)}]:$\alpha_{i\phantom{j_2}}$}] (ai) at (1,0.5) {};
    \node[thick, circle, label={[shift={(0.85,-0.6)}]:$\alpha_{j_2}$}] (aj2) at (1,3.5) {};
    \node[circle, densely dashed, minimum size=8mm] at (aj2) {};
    \node[thick, circle, label={[shift={(-0.4, -1.2)}]:$\alpha_{j_1}$}] (aj1) at (0.5,2) {};
    \node[thick, circle, label={[shift={(0, -1.3)}]:$\alpha_{j_3}$}] (aj3) at (1.3,2) {};
    \node[thick, circle, label={[shift={(0, -1.3)}]:$\alpha_{j_4}$}] (aj4) at (2.1,2) {};
    
    \node[thick, circle, label={[shift={(0.4, -1.5)}]:$S_{4\phantom{d-1}}$}] (s4) at (2.9,2) {};
    \node[thick, circle, label={[shift={(0.4, -1.5)}]$S_{5\phantom{d-1}}$}] (s5) at (3.7,2) {};
    \node[thick, circle, label={[shift={(0.4, -1.5)}]:$S_{6\phantom{d-1}}$}] (s6) at (4.5,2) {};
    
    \node[draw=none, inner sep=0.5mm] (dots1) at (5.3,2) {$\dots$};
    
    \node[thick, circle, label={[shift={(0.1, -1.5)}]:$S_{3d-2}$}] (s3d2) at (6.1,2) {};
    \node[thick, circle, label={[shift={(0.3, -1.5)}]:$S_{3d-1}$}] (s3d1) at (6.9,2) {};
    \node[thick, circle, label={[shift={(0.5, -1.5)}]:$S_{3d\phantom{-3}}$}] (s3d) at (7.7,2) {};
    
    \node[draw=none, inner sep=0.5mm] (dots2) at (8.5,2) {$\dots$};
    
    \node[thick, circle, label={[shift={(0.1, -1.5)}]:$S_{3c-2}$}] (s3c2) at (9.3,2) {};
    \node[thick, circle, label={[shift={(0.3, -1.5)}]:$S_{3c-1}$}] (s3c1) at (10.1,2) {};
    \node[thick, circle, label={[shift={(0.5, -1.5)}]:$S_{3c\phantom{-1}}$}] (s3c) at (10.9,2) {};
    
    \node[thick, circle, label={[shift={(0.5, -1.5)}]:$\alpha_{w_1}\vphantom{S_{3d-1}}$}] (aw1) at (11.7,2) {};
    \node[draw=none] (aw2) at (12.5,2) {$\dots$};
    
    \node[thick, circle, label={[shift={(0.9, -0.75)}]:$\alpha_{z_3}\vphantom{S_{3d-1}}$}] (az3) at (11.7,3.5) {};
    4
    \node[circle, densely dashed, minimum size=8mm] at (ai) {};
    \node[rectangle, inner sep=0, minimum height=8mm, minimum width=20.7mm, rounded corners=3mm, densely dashed] (triple1) at (aj3) {};
    \node[rectangle, inner sep=0, minimum height=8mm, minimum width=20.7mm, rounded corners=3mm, densely dashed] (triple1) at (s5) {};
    \node[rectangle, inner sep=0, minimum height=8mm, minimum width=20.7mm, rounded corners=3mm, densely dashed] (triple1) at (s3d1) {};
    \node[rectangle, inner sep=0, minimum height=8mm, minimum width=20.7mm, rounded corners=3mm, densely dashed] (triple1) at (s3c1) {};
    \node[circle, densely dashed, minimum size=8mm] at (az3) {};

    \begin{scope}
        \clip(0, 0) rectangle (12.2,4);
         \node[rectangle, inner sep=0, minimum height=8mm, minimum width=20.7mm, rounded corners=3mm, densely dashed] (triple2) at (aw2) {};
    \end{scope}
\end{scope}
\begin{scope}
    \foreach \from/\to in {aj2/aj1, aj1/ai, aj1/aj3, aj3/aj4, aj4/s4, s4/s5, s5/s6, s6/dots1, dots1/s3d2, s3d2/s3d1, s3d1/s3d, s3d/dots2, dots2/s3c2, s3c2/s3c1, s3c1/s3c, s3c/aw1, aw1/az3, aw1/aw2}
        \draw [thick] (\from) -- (\to);
\end{scope}
\end{tikzpicture}
    
    \vspace*{-6pt}
    \caption{Players and triples in $M$ before a new iteration of the \texttt{while} loop} 
    \label{fig:3d_sr_sas_bin_algorithm_7_cases_original_case}
\end{figure}

The six stopping conditions correspond to seven different cases, labelled Case 1 -- Case 7, in which $M'$ is constructed. Each condition corresponds to a single construction except the first condition, which corresponds to two constructions (Case 1 and Case 3). Cases 1 and 3 generalise the first example case, described above, in which some $\alpha_{z_1}$ exists. Case 2 generalises the second example case described above, in which no such $\alpha_{z_1}$ exists but some $\alpha_{z_2}$ exists as described. Cases 4 -- 7 correspond to similar scenarios. Like the two example cases, in each of Cases 1 -- 6 the algorithm identifies a number of agents divisible by three, and in each of these cases no agent identified by the algorithm, including $\alpha_i$, is unmatched in $M'$. This fact greatly simplifies the proof that $M'$ is stable in each of Cases 1 -- 6. Case 7 is unique, since the number of agents identified is not divisible by three. In Case 7, the final agent in the list $S$, labelled $S_{3c}$, for which $u_{S_{3c}}(M)=1$, is unmatched in $M'$. To show that this agent does not belong to a triple that blocks $M'$ we rely on the fact that no condition relating to previous cases held in any previous iteration of the main loop. In this way, the six stopping conditions and seven corresponding constructions of $M'$ are somewhat hierarchical. For another example, the proof that $M'$ is stable in Case 4 relies on the fact that in no iteration did the condition for Cases 1 and 3 hold. A similar reliance exists in the proofs of each of the other cases. This dependence between the cases, which is evident in the overall proof, helps show why all seven cases are required in this algorithm.

Algorithm~\texttt{repair} is presented in Algorithm~\ref{alg:3dsrsasbin_almostthere_algo} in two parts. The first part involves the construction of $S$ and exploration of the instance. The second part involves the construction of $M'$.  In order to establish the correctness and complexity of this algorithm we use a number of lemmas. The following lemma shows that the \texttt{while} loop in Algorithm~\texttt{repair} eventually terminates.

\else

We now provide some intuition behind Algorithm~\texttt{repair} and refer the reader to Figure~\ref{fig:3d_sr_sas_bin_algorithm_7_cases_original_case}. Recall that the overall goal of the algorithm is to construct a stable $P$\nobreakdash-matching $M'$. Since the given $P$\nobreakdash-matching $M$ is repairable, our aim will be to modify $M$ such that $u_{\alpha_i}(M')\geq 1$ while ensuring that no three agents that are ordered to different triples in $M'$ block $M'$. The stability of the constructed $P$\nobreakdash-matching $M'$ then follows. 

The algorithm begins by selecting some triple $\{ \alpha_i, \alpha_{j_1}, \alpha_{j_2} \}$ that blocks $M$. The two agents in $M(\alpha_{j_1}) \setminus \{ \alpha_{j_1} \}$ are labelled $\alpha_{j_3}$ and $\alpha_{j_4}$. We present two example cases in which it is possible to construct a stable $P$\nobreakdash-matching. First, suppose there exists some $\alpha_{z_1}$ where $\mathit{val}_{\alpha_{j_3}}(\alpha_{z_1})=1$ and $u_{\alpha_{z_1}}(M)=0$.  Construct $M'$ from $M$ by removing $\{ \alpha_{j_1}, \alpha_{j_2}, \alpha_{j_3} \}$ and adding $\{ \alpha_i, \alpha_{j_1}, \alpha_{j_2} \}$ and $\{ \alpha_{j_3}, \alpha_{j_4}, \alpha_{z_1} \}$. Now, $u_{\alpha_i}(M')=1$ and $u_{\alpha_p}(M')\geq u_{\alpha_p}(M)$ for any $\alpha_p \in N\setminus \{ \alpha_i \}$. It follows by Lemma~\ref{lem:3dsrsasbinblockerimprovement} that $M'$ is stable. Second, suppose there exists no such $\alpha_{z_1}$ but there exists some $\alpha_{z_2}$ where $\mathit{val}_{\alpha_{j_4}}(\alpha_{z_2})=1$ and $u_{\alpha_{z_2}}(M)=0$. Now construct $M'$ from $M$ by removing $\{ \alpha_{j_1}, \alpha_{j_2}, \alpha_{j_3} \}$ and adding $\{ \alpha_i, \alpha_{j_1}, \alpha_{j_2} \}$ and $\{ \alpha_{j_3}, \alpha_{j_4}, \alpha_{z_2} \}$. Note that $u_{\alpha_i}(M')=1$ and $u_{\alpha_p}(M')\geq u_{\alpha_p}(M)$ for any $\alpha_p \in N\setminus \{ \alpha_i, \alpha_{j_3} \}$. It can be shown that $\alpha_{j_3}$ does not belong to a triple that blocks $M'$ since no $\alpha_{z_1}$ exists as described. It follows again by Lemma~\ref{lem:3dsrsasbinblockerimprovement} that $M'$ is stable. Generalising these two example cases, the algorithm constructs a list $S$ of agents, which initially comprises $\langle \alpha_{j_1}, \alpha_{j_3},\allowbreak \alpha_{j_4} \rangle$. The list $S$ has length $3c$ for some $c\geq 1$, where $\{S_{3c-2}, S_{3c-1}, S_{3c} \} \in M$ and $\mathit{val}_{S_p}(S_{p+1})=1$ for each $p$ ($1 \leq p < 3c$). The list $S$ therefore corresponds to a path in the underlying graph. In each iteration of the main loop, three agents belonging to some triple in $M$ are appended to the end of $S$. The loop continues until $S$ satisfies at least one of six specific conditions. We show that eventually at least one of these conditions must hold. 

These six stopping conditions correspond to seven different cases, labelled Case 1 -- Case 7, in which a stable $P$\nobreakdash-matching $M'$ may be constructed. The exact construction of $M'$ depends on which condition(s) caused the main loop to terminate. Cases 1 and 3 generalise the first example case, in which some $\alpha_{z_1}$ exists as described. Case 2 generalises the second example case, in which no such $\alpha_{z_1}$ exists but some $\alpha_{z_2}$ exists as described. Cases 4 -- 7 correspond to similar scenarios. The six stopping conditions and seven corresponding constructions of $M'$ are somewhat hierarchical. For example, the proof that $M'$ is stable in Case 4 relies on the fact that in no iteration did the condition for Cases 1 and 3 hold. A similar reliance exists in the proofs of each of the other cases. The proof that $M'$ is stable in Case 7 is the most complex. It relies on the fact that no condition relating to any of the previous six cases held in the final or some previous iteration of the main loop. Further intuition for the different cases is given in the full version \cite{fullversion3dsraspaper}.

Algorithm~\texttt{repair} is presented in Algorithm~\ref{alg:3dsrsasbin_almostthere_algo} in two parts. The first part involves the construction of $S$ and exploration of the instance. The second part involves the construction of $M'$.  The following lemma establishes the correctness and complexity of this algorithm.

\fi

\ifdefined \fullversion

\begin{lem}
\label{lem:algalwaysterminates}
The {\normalfont \texttt{while}} loop in Algorithm~{\normalfont \texttt{repair}} terminates after at most $\lfloor (|N|-2) \mathbin{/} 3 \rfloor$ iterations.
\end{lem}
\begin{proof}
Any three agents $\{ \alpha_{w_1}, \alpha_{w_2}, \alpha_{w_3} \}$ added to $S$ in a single iteration comprise a triple in $M$. Just before the addition of $\langle \alpha_{w_1}, \alpha_{w_2}, \alpha_{w_3} \rangle$ to $S$, we know that $\alpha_{w_1} \notin S$. It follows that $\alpha_{w_2}, \alpha_{w_3} \notin S$, so in general $S$ contains any agent at most once. Since $\alpha_i, \alpha_{j_2} \notin S$ it follows that $|S|\leq |N| - 2$ and thus the algorithm terminates after at most $\lfloor (|N|-2) \mathbin{/} 3 \rfloor$ iterations of the \texttt{while} loop.
\end{proof}

In Case 3, the algorithm identifies some agent $\alpha_{z_4}$ in $N\setminus \{ \alpha_i \}$ such that $\mathit{val}_{S_{3c-1}}(\alpha_{z_4}) = 1$ and $u_{\alpha_{z_4}}(M) = 0$. Proposition~\ref{prop:az4exists} shows that such an agent is guaranteed to exist. 

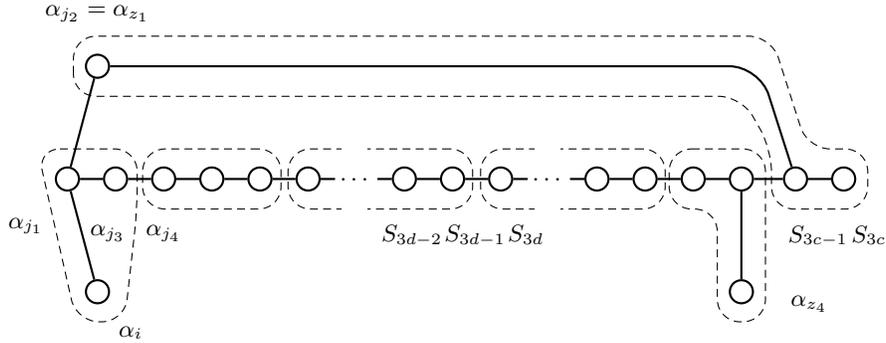
\begin{figure}
    \centering
    \begin{tikzpicture}
\begin{scope}[every node/.style={circle,draw, minimum size=2.4mm}, xscale=0.8]
    \node[thick, circle, label={[label distance=1.8mm]-45:$\alpha_{i\phantom{j_2}}$}] (ai) at (1,0.5) {};
    \node[thick, circle, label={[shift={(0.0, -0.3)}]:$\alpha_{j_2}=\alpha_{z_1}$}] (aj2) at (1,3.5) {};
    \node[thick, circle, label={[shift={(-0.55, -1.2)}]:$\alpha_{j_1}$}] (aj1) at (0.5,2) {};
    \node[thick, circle, label={[shift={(-0.1, -1.3)}]:$\alpha_{j_3}$}] (aj3) at (1.3,2) {};
    \node[thick, circle, label={[shift={(0, -1.3)}]:$\alpha_{j_4}$}] (aj4) at (2.1,2) {};
    
    \node[thick, circle, label={[shift={(0.4, -1.5)}]}] (s4) at (2.9,2) {};
    \node[thick, circle, label={[shift={(0.4, -1.5)}]}] (s5) at (3.7,2) {};
    \node[thick, circle, label={[shift={(0.4, -1.5)}]}] (s6) at (4.5,2) {};
    
    \node[draw=none, inner sep=0.5mm] (dots1) at (5.3,2) {$\dots$};
    
    \node[thick, circle, label={[shift={(0.1, -1.5)}]:$S_{3d-2}$}] (s3d2) at (6.1,2) {};
    \node[thick, circle, label={[shift={(0.3, -1.5)}]:$S_{3d-1}$}] (s3d1) at (6.9,2) {};
    \node[thick, circle, label={[shift={(0.5, -1.5)}]:$S_{3d\phantom{-3}}$}] (s3d) at (7.7,2) {};
    
    \node[draw=none, inner sep=0.5mm] (dots2) at (8.5,2) {$\dots$};
    
    \node[thick, circle, label={[shift={(0.1, -1.5)}]}] (s3c5) at (9.3,2) {};
    \node[thick, circle, label={[shift={(0.3, -1.5)}]}] (s3c4) at (10.1,2) {};
    \node[thick, circle, label={[shift={(0.5, -1.5)}]}] (s3c3) at (10.9,2) {};
    
    \node[thick, circle, label={[shift={(0.9, -0.75)}]:$\alpha_{z_4}$}] (az4) at (11.7,0.5) {};
    
    \node[thick, circle, label={[shift={(0.5, -1.5)}]}] (s3c2) at (11.7,2) {};
    \node[thick, circle, label={[shift={(0.3, -1.5)}]:$S_{3c-1}$}] (s3c1) at (12.6,2) {};
    \node[thick, circle, label={[shift={(0.5, -1.5)}]:$S_{3c\phantom{-1}}$}] (s3c) at (13.4,2) {};
    
    \node[draw=none, inner sep=0, minimum size=0] (s3by) at (12.0, 3.5]) {};
    
    \draw [thick, rounded corners=4mm] (aj2.east)--(s3by.center)--(s3c1);

    \node[rectangle, inner sep=0, minimum height=8mm, minimum width=18.4mm, rounded corners=3mm, densely dashed] (triple1) at (s4) {};
    
    \begin{scope}
        \clip(0,1) rectangle (5.05, 4.0);
        \node[rectangle, inner sep=0, minimum height=8mm, minimum width=24.54mm, rounded corners=3mm, densely dashed] (triple1) at ($(s6)!0.5!(s3d1)$) {};
    \end{scope}
    \begin{scope}
            \clip(5.5, 0.0) rectangle (12, 4.0);
            \node[rectangle, inner sep=0, minimum height=8mm, minimum width=24.54mm, rounded corners=3mm, densely dashed] (triple1) at ($(s6)!0.5!(s3d1)$) {};
    \end{scope}
    
    \begin{scope}
        \clip(0,1) rectangle (8.25, 4.0);
        \node[rectangle, inner sep=0, minimum height=8mm, minimum width=24.54mm, rounded corners=3mm, densely dashed] (triple1) at ($(s3d)!0.5!(s3c4)$) {};
    \end{scope}
    \begin{scope}
            \clip(8.7, 0.0) rectangle (12, 4.0);
            \node[rectangle, inner sep=0, minimum height=8mm, minimum width=24.54mm, rounded corners=3mm, densely dashed] (triple1) at ($(s3d)!0.5!(s3c4)$) {};
    \end{scope}

    \draw [rounded corners=3mm, densely dashed] (0.0, 2.4)--(1.66, 2.4)--(1.66, 1.6)--(1.45, 0.1)--(0.65, 0.1)--cycle;
    
    \draw [rounded corners=3mm, densely dashed] (0.6, 3.1)--(0.6, 3.9)--(12.2,3.9)--(12.8,2.4)--(13.8,2.4)--(13.8,1.6)--(12.2,1.6)--(12.2,2.4)--(11.7, 3.1)--cycle;
    
    \draw [rounded corners=3mm, densely dashed] (10.5, 2.4)--(12.1, 2.4)--(12.1, 0.1)--(11.3, 0.1)--(11.3, 1.6)--(10.5, 1.6)--cycle;
\end{scope}
\begin{scope}
    \foreach \from/\to in {aj2/aj1, aj1/ai, aj1/aj3, aj3/aj4, aj4/s4, s4/s5, s5/s6, s6/dots1, dots1/s3d2, s3d2/s3d1, s3d1/s3d, s3d/dots2, dots2/s3c5, s3c5/s3c4, s3c4/s3c3, s3c3/s3c2, s3c2/s3c1, s3c1/s3c, az4/s3c2}
        \draw [thick] (\from) -- (\to);
\end{scope}
\end{tikzpicture}
    \caption{The structure of $M'$ in Case 3} 
    \label{fig:3d_sr_sas_bin_algorithm_7_cases_case_3}
\end{figure}

\begin{prop}
\label{prop:az4exists}
In Case 3 of Algorithm~{\normalfont \texttt{repair}}, some agent $\alpha_{z_4}$ in $N\setminus \{ \alpha_i, \alpha_{j_2} \}$ exists where $\mathit{val}_{S_{3c-2}}(\alpha_{z_4})=1$ and $u_{\alpha_{z_4}}(M)=0$.
\end{prop}
\begin{proof}
Refer to Figure~\ref{fig:3d_sr_sas_bin_algorithm_7_cases_case_3}. We claim that the condition of Case 3 implies that $c\geq 1$. Suppose for a contradiction that $c=1$. The condition shows that $\alpha_{z_1}=\alpha_{j_2}$ exists where $\mathit{val}_{\alpha_{z_1}}(S_{3c-1})=1$. Since $S_{3c-1}=\alpha_{j_3}$, the triple $\{ \alpha_{j_1}, \alpha_{j_2}, \alpha_{j_3} \}$ contradicts the fact that $(N, V)$ is triangle-free.

Since $c>1$ it follows that $c'=c-1$ is the value of $c$ in the second last iteration of the \texttt{while} loop. Consider the second last iteration of the \texttt{while} loop. In this iteration $\alpha_{w_1}=S_{3c-2}$ was identified where $\mathit{val}_{S_{3c'}}(\alpha_{w_1})=1$, $\alpha_{w_1}\notin S$ and there existed $\alpha_{z_3}\in N\setminus \{ \alpha_i \}$ where $\mathit{val}_{\alpha_{w_1}}(\alpha_{z_3})=1$ and $u_{\alpha_{z_3}}(M)=0$. We refer to the agent labelled $\alpha_{z_3}$ in this iteration as $\alpha_{z_4}$. It follows that $\mathit{val}_{S_{3c-2}}(\alpha_{z_4})=1$.

We claim that $\alpha_{z_4}\neq (\alpha_{z_1} = \alpha_{j_2})$ since otherwise the triple $\{ \alpha_{z_4}, S_{3c-1}, S_{3c-2} \}$ contradicts the fact that $(N, V)$ is triangle-free. It follows that $\alpha_{z_4}\in N\setminus \{ \alpha_i, \alpha_{j_2} \}$, completing the proof.
\end{proof}

Likewise in Case 6, the algorithm identifies some agent $\alpha_{z_5}$ in $N\setminus \{ \alpha_i, \alpha_{j_2} \}$ exists where $\mathit{val}_{S_{3b+1}}(\alpha_{z_5})=1$ and $u_{\alpha_{z_5}}(M)=0$. Proposition~\ref{prop:az5exists} shows that such an agent is guaranteed to exist.

\begin{prop}
\label{prop:az5exists}
In Case 6 of Algorithm~{\normalfont \texttt{repair}}, some agent $\alpha_{z_5}$ in $N\setminus \{ \alpha_i, \alpha_{j_2} \}$ exists where $\mathit{val}_{S_{3b+1}}(\alpha_{z_5})=1$ and $u_{\alpha_{z_5}}(M)=0$.
\end{prop}
\begin{proof}
Refer to Figure~\ref{fig:3d_sr_sas_bin_algorithm_7_cases_case_6}. It follows from definition of $b$ and the condition of Case 6 that $b < c$.

\begin{figure}
    \centering
    \begin{tikzpicture}
\begin{scope}[every node/.style={circle,draw, minimum size=2.4mm}, xscale=0.8]
    \node[thick, circle, label={[label distance=1.8mm]-45:$\alpha_{i\phantom{j_2}}$}] (ai) at (1,0.5) {};
    \node[thick, circle, label={[shift={(-0.4, 0)}]:$\alpha_{j_2}$}] (aj2) at (1,3.5) {};
    \node[thick, circle, label={[shift={(-0.55, -1.2)}]:$\alpha_{j_1}$}] (aj1) at (0.5,2) {};
    \node[thick, circle, label={[shift={(-0.1, -1.3)}]:$\alpha_{j_3}$}] (aj3) at (1.3,2) {};
    \node[thick, circle, label={[shift={(0, -1.3)}]:$\alpha_{j_4}$}] (aj4) at (2.1,2) {};
    
    \node[thick, circle, label={[shift={(0.4, -1.85)}]}] (s4) at (2.9,2) {};
    \node[thick, circle, label={[shift={(0.4, -1.5)}]}] (s5) at (3.7,2) {};
    \node[thick, circle, label={[shift={(0.4, -1.5)}]}] (s6) at (4.5,2) {};
    
    \node[draw=none, inner sep=0.5mm] (dots1) at (5.3,2) {$\dots$};
    
    \node[thick, circle, label={[shift={(0.1, -1.5)}]}] (s3b2) at (6.1,2) {};
    \node[thick, circle, label={[shift={(0.1, -1.5)}]:$S_{3b-1}$}] (s3b1) at (6.9,2) {};
    \node[draw=none] (s3b) at (7.7,2) {};
    
    \node[thick, circle, label={[shift={(0.2, -0.1)}]:$S_{3b+1}$}] (s3bp1) at (8.6,2) {};
    \node[thick, circle, label={[shift={(0.3, -0.1)}]:$S_{3b+2}$}] (s3bp2) at (9.4,2) {};
    \node[thick, circle, label={[shift={(0.5, -0.1)}]}] (s3bp3) at (10.2,2) {};
    
    \node[thick, circle, label={[shift={(0.9, -0.75)}]:$\alpha_{z_5}\vphantom{S_{3d-1}}$}] (az5) at (8.6,0.5) {};
    
     \node[draw=none, inner sep=0.5mm] (dots2) at (11.0,2) {$\dots$};
    
    \node[thick, circle, label={[shift={(0.1, -1.5)}]}] (s3c2) at (11.8,2) {};
    \node[thick, circle, label={[shift={(0.3, -1.5)}]}] (s3c1) at (12.6,2) {};
    \node[thick, circle, label={[shift={(0.5, -1.5)}]:$S_{3c\phantom{-1}}$}] (s3c) at (13.4,2) {};
    
    \node[draw=none, inner sep=0, minimum size=0] (s3bx) at (7.7, 3.5) {};
    \node[draw=none, inner sep=0, minimum size=0] (s3by) at (12.8, 3.5) {};

    \node[rectangle, inner sep=0, minimum height=8mm, minimum width=18.4mm, rounded corners=3mm, densely dashed] (triple1) at (s4) {};
    
    \begin{scope}
        \clip(0,1) rectangle (5.05, 4.0);
        \node[rectangle, inner sep=0, minimum height=8mm, minimum width=24.54mm, rounded corners=3mm, densely dashed] (triple1) at ($(s6)!0.5!(s3b1)$) {};
    \end{scope}
    \begin{scope}
            \clip(5.5, 0.0) rectangle (12, 4.0);
            \node[rectangle, inner sep=0, minimum height=8mm, minimum width=24.54mm, rounded corners=3mm, densely dashed] (triple1) at ($(s6)!0.5!(s3b1)$) {};
    \end{scope}
    
    \begin{scope}
        \clip(0,1) rectangle (10.65, 4.0);
        \node[rectangle, inner sep=0, minimum height=8mm, minimum width=24.54mm, rounded corners=3mm, densely dashed] (triple1) at ($(s3bp3)!0.5!(s3c1)$) {};
    \end{scope}
    \begin{scope}
            \clip(11.1, 0.0) rectangle (14, 4.0);
            \node[rectangle, inner sep=0, minimum height=8mm, minimum width=24.54mm, rounded corners=3mm, densely dashed] (triple1) at ($(s3bp3)!0.5!(s3c1)$) {};
    \end{scope}
    
    
    \draw [rounded corners=3mm, densely dashed] (0.0, 2.4)--(1.66, 2.4)--(1.66, 1.6)--(1.45, 0.1)--(0.65, 0.1)--cycle;
    
    \draw [rounded corners=3mm, densely dashed] (0.6,3.9)--(13.0, 3.9)--(13.8, 2.4)--(13.8, 1.6)--(13.0, 1.6)--(13.0, 2.4)--(12.5, 3.1)--(8.1, 3.1)--(8.1, 1.6)--(7.3, 1.6)--(7.3, 3.1)--(0.6, 3.1)--cycle;
    
    \draw [rounded corners=3mm, densely dashed] (8.2, 2.4)--(9.8, 2.4)--(9.8, 1.6)--(9.0, 1.6)--(9.0, 0.1)--(8.2, 0.1)--cycle;
    
    \node[fill=white, draw=white, minimum size=6mm] (s3bbackground) at (8.2, 1.3) {};
    \node[thick, circle, label={[shift={(0.5, -1.5)}]:$S_{3b\phantom{-3}}$}] (s3bextra) at (7.7,2) {};
    
    \draw [thick, rounded corners=3mm] (s3b)--(s3bx.center)--(s3by.center)--(s3c.north);
    \draw [thick, rounded corners=3mm] (s3b)--(s3bx.center)--(aj2);
\end{scope}
\begin{scope}
    \foreach \from/\to in {aj2/aj1, aj1/ai, aj1/aj3, aj3/aj4, aj4/s4, s4/s5, s5/s6, s6/dots1, dots1/s3b2, s3b2/s3b1, s3b1/s3b, s3b/s3bp1, s3bp1/s3bp2, s3bp2/s3bp3, s3bp3/dots2, dots2/s3c2, s3c2/s3c1, s3c1/s3c, s3bp1/az5}
        \draw [thick] (\from) -- (\to);
\end{scope}
\end{tikzpicture}
    \caption{The structure of $M'$ in Case 6} 
    \label{fig:3d_sr_sas_bin_algorithm_7_cases_case_6}
\end{figure}
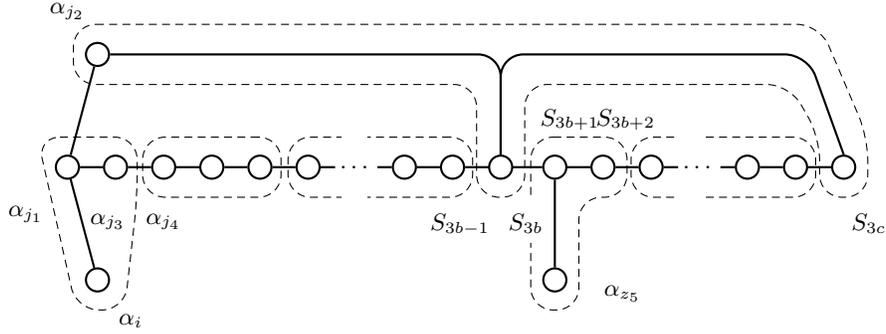

Consider the $(b+1)\textsuperscript{th}$ iteration of the \texttt{while} loop. In this iteration $\alpha_{w_1}=S_{3b+1}$ was identified and the three agents in $M(\alpha_{w_1}) = \{ S_{3b+1}, S_{3b+2}, S_{3b+3} \}$ were added to the end of $S$. By definition of $\alpha_{w_1}$, in that iteration some agent $\alpha_{z_3}\in N \setminus \{ \alpha_i \}$ was identified where $\mathit{val}_{\alpha_{w_1}}(\alpha_{z_3})=1$ and $u_{\alpha_{z_3}}(M)=0$. We refer to the agent labelled $\alpha_{z_3}$ in this iteration as $\alpha_{z_5}$. It follows that $\mathit{val}_{S_{3b+1}}(\alpha_{z_5})=1$.

We claim that $\alpha_{z_5}\neq \alpha_{j_2}$ since otherwise the triple $\{ S_{3b}. \alpha_{z_5}, S_{3b+1} \}$ contradicts the fact that $(N, V)$ is triangle-free. It follows that $\alpha_{z_5}\in N\setminus \{ \alpha_i, \alpha_{j_2} \}$, completing the proof.
\end{proof}

\begin{lem}
\label{lem:algreturnspmatching}
Algorithm~{\normalfont \texttt{repair}} returns a $P$\nobreakdash-matching.
\end{lem}
\begin{proof}
By inspection of the construction of $M'$ and Figures~ \ref{fig:3d_sr_sas_bin_algorithm_7_cases_case_1} \nobreakdash-- \ref{fig:3d_sr_sas_bin_algorithm_7_cases_case_7}.
\end{proof}

In the remainder of this section we will show that the returned $P$\nobreakdash-matching $M'$ is stable in $(N, V)$. The construction of $M'$ is slightly different in each of Cases 1 -- 7. In Lemmmas~\ref{lem:algocases1and3noalphapexists}, \ref{lem:algocases245and6noalphapexists} and~\ref{lem:algocase7noalphapexists} we show in each of the cases that no agent $\alpha_g \in N$ exists where $u_{\alpha_{g}}(M') < u_{\alpha_{g}}(M)$ and $\alpha_g$ belongs to a triple that blocks $M'$. It follows directly that $M'$ is stable (shown in Lemma \ref{lem:algoreturnsstablematching_notimecomplex}).

\begin{lem}
\label{lem:algocases1and3noalphapexists}
In Cases 1 and 3 of Algorithm~{\normalfont \texttt{repair}}, no agent $\alpha_{g}\in N$ exists where $u_{\alpha_{g}}(M') < u_{\alpha_{g}}(M)$ and $\alpha_g$ belongs to a triple that blocks $M'$.
\end{lem}
\begin{proof}
Refer to Figures~\ref{fig:3d_sr_sas_bin_algorithm_7_cases_case_3} and~\ref{fig:3d_sr_sas_bin_algorithm_7_cases_case_1}. Suppose for a contradiction that some such $\alpha_g\in N$ exists. By the construction of $M'$ in Cases 1 and 3, $u_{\alpha_p}(M')\geq u_{\alpha_p}(M)$ for any $\alpha_p\in N \setminus S$. It follows that $\alpha_g\in S$ and hence $u_{\alpha_g}(M')\geq 1$. Since $u_{\alpha_{g}}(M') < u_{\alpha_{g}}(M)$ it must be that $u_{\alpha_g}(M) = 2$. The only such agents in $S$ are $S_{3d-1}$ for $1 \leq d \leq c$.

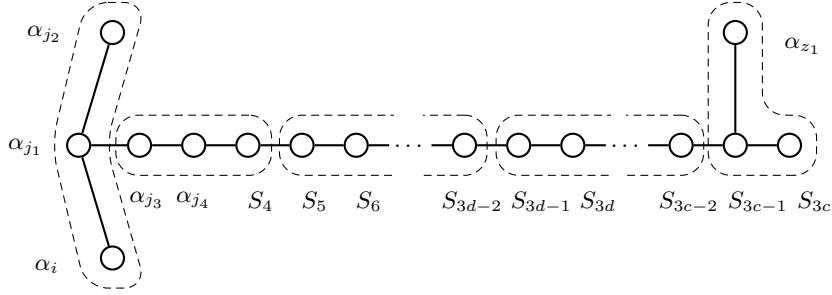
\begin{figure}
    \centering
    \begin{tikzpicture}
\begin{scope}[every node/.style={circle,draw, minimum size=2.4mm}, xscale=0.9]
    \node[thick, circle, label={[shift={(-0.75, -0.75)}]:$\alpha_{i\phantom{j_2}}$}] (ai) at (1,0.5) {};
    \node[thick, circle, label={[shift={(-0.9, -0.6)}]:$\alpha_{j_2}$}] (aj2) at (1,3.5) {};
    \node[thick, circle, label={[shift={(-0.7, -0.6)}]:$\alpha_{j_1}$}] (aj1) at (0.5,2) {};
    \node[thick, circle, label={[shift={(0.1, -1.3)}]:$\alpha_{j_3}$}] (aj3) at (1.4,2) {};
    \node[thick, circle, label={[shift={(0, -1.3)}]:$\alpha_{j_4}$}] (aj4) at (2.2,2) {};
    \node[thick, circle, label={[shift={(0.4, -1.5)}]:$S_{4\phantom{d-1}}$}] (s4) at (3.0,2) {};
    \node[thick, circle, label={[shift={(0.4, -1.5)}]$S_{5\phantom{d-1}}$}] (s5) at (3.8,2) {};
    \node[thick, circle, label={[shift={(0.4, -1.5)}]:$S_{6\phantom{d-1}}$}] (s6) at (4.6,2) {};
    \node[draw=none, inner sep=0.5mm] (dots1) at (5.4,2) {$\dots$};
    \node[thick, circle, label={[shift={(0.1, -1.5)}]:$S_{3d-2}$}] (s3d2) at (6.2,2) {};
    \node[thick, circle, label={[shift={(0.3, -1.5)}]:$S_{3d-1}$}] (s3d1) at (7.0,2) {};
    \node[thick, circle, label={[shift={(0.5, -1.5)}]:$S_{3d\phantom{-3}}$}] (s3d) at (7.8,2) {};
    \node[draw=none, inner sep=0.5mm] (dots2) at (8.6,2) {$\dots$};
    \node[thick, circle, label={[shift={(0.1, -1.5)}]:$S_{3c-2}$}] (s3c2) at (9.4,2) {};
    \node[thick, circle, label={[shift={(0.3, -1.5)}]:$S_{3c-1}$}] (s3c1) at (10.2,2) {};
    \node[thick, circle, label={[shift={(0.5, -1.5)}]:$S_{3c\phantom{-1}}$}] (s3c) at (11.0,2) {};
    \node[thick, circle, label={[shift={(0.9, -0.75)}]:$\alpha_{z_1}\vphantom{S_{3d-1}}$}] (az1) at (10.2,3.5) {};
    \draw [rounded corners=3mm, densely dashed] (0.6, 3.9)--(1.5, 3.9)--(0.9, 2)--(1.5, 0.1)--(0.6, 0.1)--(0.1, 2)--cycle;
    
    \begin{scope}[yscale=-1, xscale=-1, yshift=-4.0cm, xshift=0.1cm, xshift=-22.0cm]
        \draw [rounded corners=3mm, densely dashed] (10.5, 2.4)--(12.1, 2.4)--(12.1, 0.1)--(11.3, 0.1)--(11.3, 1.6)--(10.5, 1.6)--cycle;
    \end{scope}
    
    \node[rectangle, inner sep=0, minimum height=8mm, minimum width=20.7mm, rounded corners=3mm, densely dashed] (triple1) at (aj4) {};
    
    \begin{scope}
        \clip(0,1) rectangle (8.35, 4.0);
        \node[rectangle, inner sep=0, minimum height=8mm, minimum width=27.6mm, rounded corners=3mm, densely dashed] (triple1) at ($(s3d1)!0.5!(s3c2)$) {};
    \end{scope}
    \begin{scope}
            \clip(8.6, 0.0) rectangle (12, 4.0);
            \node[rectangle, inner sep=0, minimum height=8mm, minimum width=27.6mm, rounded corners=3mm, densely dashed] (triple1) at ($(s3d1)!0.5!(s3c2)$) {};
    \end{scope}
    
    \begin{scope}
        \clip(0,1) rectangle (5.15, 4.0);
        \node[rectangle, inner sep=0, minimum height=8mm, minimum width=27.6mm, rounded corners=3mm, densely dashed] (triple1) at ($(s5)!0.5!(s3d2)$) {};
    \end{scope}
    \begin{scope}
            
            \clip(5.6, 0.0) rectangle (12, 4.0);
            \node[rectangle, inner sep=0, minimum height=8mm, minimum width=27.6mm, rounded corners=3mm, densely dashed] (triple1) at ($(s5)!0.5!(s3d2)$) {};
    \end{scope}
\end{scope}
\begin{scope}
    \foreach \from/\to in {aj2/aj1, aj1/ai, aj1/aj3, aj3/aj4, aj4/s4, s4/s5, s5/s6, s6/dots1, dots1/s3d2, s3d2/s3d1, s3d1/s3d, s3d/dots2, dots2/s3c2, s3c2/s3c1, s3c1/s3c, s3c1/az1}
        \draw [thick] (\from) -- (\to);
\end{scope}
\end{tikzpicture}
    \caption{The structure of $M'$ in Case 1} 
    \label{fig:3d_sr_sas_bin_algorithm_7_cases_case_1}
\end{figure}

First consider $S_{3c-1}$. Since $u_{S_{3c-1}}(M')=2$ it follows that $S_{3c-1}$ does not belong to a triple that blocks $M'$ and hence $\alpha_g\neq S_{3c-1}$. Now consider $S_{3d-1}$ for $1\leq d < c$. Suppose for a contradiction that triple $\{ S_{3d-1}, \alpha_{k_1}, \alpha_{k_2} \}$ blocks $M'$ where $\alpha_{k_1}, \alpha_{k_2} \in N$. Since $u_{S_{3d-1}}(M')=1$ it follows that $u_{S_{3d-1}}(\{ \alpha_{k_1}, \alpha_{k_2} \})=2$ and hence that $\mathit{val}_{S_{3d-1}}(\alpha_{k_1})=\mathit{val}_{S_{3d-1}}(\alpha_{k_2})=1$. Consider $\alpha_{k_1}$ and $\alpha_{k_2}$. Since $(N, V)$ is triangle-free, it must be that $u_{\alpha_{k_1}}(\{ S_{3d-1}, \alpha_{k_2} \})=u_{\alpha_{k_2}}(\{ S_{3d-1}, \alpha_{k_1} \})=1$. It follows that $u_{\alpha_{k_1}}(M')=u_{\alpha_{k_2}}(M')=0$. By construction of $M'$, no agent $\alpha_p \in N$ exists where $u_{\alpha_p}(M') = 0$ and $u_{\alpha_p}(M') < u_{\alpha_p}(M)$. It follows that $u_{\alpha_{k_1}}(M)=u_{\alpha_{k_2}}(M)=0$. Recall the $d\textsuperscript{th}$ iteration of the \texttt{while} loop. We have shown that two agents $\alpha_{k_1}, \alpha_{k_2}$ exist where $\mathit{val}_{S_{3d-1}}(\alpha_{k_1})=\mathit{val}_{S_{3d-1}}(\alpha_{k_2})=1$ and $u_{\alpha_{k_1}}(M)=u_{\alpha_{k_2}}(M)=0$. It follows that some $\alpha_{z_1} \in N\setminus \{\alpha_i\}$ exists where $\mathit{val}_{\alpha_{z_1}}(S_{3d-1})=1$ and $u_{\alpha_{z_1}}(M)=0$, since either $\alpha_{z_1}=\alpha_{k_1}$ or $\alpha_{i}=\alpha_{k_1}$ and $\alpha_{z_1}=\alpha_{k_2}$. In this iteration, since $\alpha_{z_1}\neq \bot$ the break condition held and the \texttt{while} loop terminated. This is a contradiction since $d < c$. It follows, for $1\leq d \leq c$, that no triple containing $S_{3d-1}$ blocks $M'$. In summary, in Cases 1 and 3, no $\alpha_g\in N$ exists where $u_{\alpha_{g}}(M') < u_{\alpha_{g}}(M)$ and $\alpha_g$ belongs to a triple that blocks $M'$.
\end{proof}

\begin{lem}
\label{lem:algocases245and6noalphapexists}
In Cases 2, 4, 5, and 6 of Algorithm~{\normalfont \texttt{repair}}, no agent $\alpha_{g}\in N$ exists where $u_{\alpha_{g}}(M') < u_{\alpha_{g}}(M)$ and $\alpha_g$ belongs to a triple that blocks $M'$.
\end{lem}
\begin{proof}
Refer to Figures~\ref{fig:3d_sr_sas_bin_algorithm_7_cases_case_6}, \ref{fig:3d_sr_sas_bin_algorithm_7_cases_case_2}, \ref{fig:3d_sr_sas_bin_algorithm_7_cases_case_4}, and~\ref{fig:3d_sr_sas_bin_algorithm_7_cases_case_5}. Suppose for a contradiction that some such $\alpha_g\in N$ exists. As before, by the construction of $M'$ in Cases 2, 4, 5, and 6, $u_{\alpha_p}(M')\geq u_{\alpha_p}(M)$ for any $\alpha_p\in N \setminus S$. It follows that $\alpha_g\in S$ and hence $u_{\alpha_g}(M')\geq 1$. Since $u_{\alpha_{g}}(M') < u_{\alpha_{g}}(M)$ it must be that $u_{\alpha_g}(M) = 2$. The only such agents in $S$ are $S_{3d-1}$ for $1 \leq d \leq c$.

\begin{figure}
    \centering
    \begin{tikzpicture}
\begin{scope}[every node/.style={circle,draw, minimum size=2.4mm}, xscale=0.9]
    \node[thick, circle, label={[shift={(-0.75, -0.75)}]:$\alpha_{i\phantom{j_2}}$}] (ai) at (1,0.5) {};
    \node[thick, circle, label={[shift={(-0.9, -0.6)}]:$\alpha_{j_2}$}] (aj2) at (1,3.5) {};
    \node[thick, circle, label={[shift={(-0.7, -0.6)}]:$\alpha_{j_1}$}] (aj1) at (0.5,2) {};
    \node[thick, circle, label={[shift={(0.1, -1.3)}]:$\alpha_{j_3}$}] (aj3) at (1.4,2) {};
    \node[thick, circle, label={[shift={(0, -1.3)}]:$\alpha_{j_4}$}] (aj4) at (2.2,2) {};
    \node[thick, circle, label={[shift={(0.4, -1.5)}]:$S_{4\phantom{d-1}}$}] (s4) at (3.0,2) {};
    \node[thick, circle, label={[shift={(0.4, -1.5)}]$S_{5\phantom{d-1}}$}] (s5) at (3.8,2) {};
    \node[thick, circle, label={[shift={(0.4, -1.5)}]:$S_{6\phantom{d-1}}$}] (s6) at (4.6,2) {};
    \node[draw=none, inner sep=0.5mm] (dots1) at (5.4,2) {$\dots$};
    \node[thick, circle, label={[shift={(0.1, -1.5)}]:$S_{3d-2}$}] (s3d2) at (6.2,2) {};
    \node[thick, circle, label={[shift={(0.3, -1.5)}]:$S_{3d-1}$}] (s3d1) at (7.0,2) {};
    \node[thick, circle, label={[shift={(0.5, -1.5)}]:$S_{3d\phantom{-3}}$}] (s3d) at (7.8,2) {};
    \node[draw=none, inner sep=0.5mm] (dots2) at (8.6,2) {$\dots$};
    \node[thick, circle, label={[shift={(0.1, -1.5)}]:$S_{3c-2}$}] (s3c2) at (9.4,2) {};
    \node[thick, circle, label={[shift={(0.3, -1.5)}]:$S_{3c-1}$}] (s3c1) at (10.2,2) {};
    \node[thick, circle, label={[shift={(0.5, -1.5)}]:$S_{3c\phantom{-1}}$}] (s3c) at (11.0,2) {};
    
    \node[thick, circle, label={[shift={(0.9, -0.75)}]:$\alpha_{z_2}\vphantom{S_{3d-1}}$}] (az2) at (11.0,3.5) {};
    
    \draw [rounded corners=3mm, densely dashed] (0.6, 3.9)--(1.5, 3.9)--(0.9, 2)--(1.5, 0.1)--(0.6, 0.1)--(0.1, 2)--cycle;
    
    \begin{scope}[yscale=-1, yshift=-4.0cm, xshift=0.1cm]
        \draw [rounded corners=3mm, densely dashed] (9.7, 2.4)--(11.3, 2.4)--(11.3, 0.1)--(10.5, 0.1)--(10.5, 1.6)--(9.7, 1.6)--cycle;
    \end{scope}
    
    \node[rectangle, inner sep=0, minimum height=8mm, minimum width=20.7mm, rounded corners=3mm, densely dashed] (triple1) at (aj4) {};
    
    \begin{scope}
        \clip(0,1) rectangle (8.15, 4.0);
        \node[rectangle, inner sep=0, minimum height=8mm, minimum width=27.6mm, rounded corners=3mm, densely dashed] (triple1) at ($(s3d1)!0.5!(s3c2)$) {};
    \end{scope}
    \begin{scope}
            
            \clip(8.7, 0.0) rectangle (12, 4.0);
            \node[rectangle, inner sep=0, minimum height=8mm, minimum width=27.6mm, rounded corners=3mm, densely dashed] (triple1) at ($(s3d1)!0.5!(s3c2)$) {};
    \end{scope}
    
    \begin{scope}
        \clip(0,1) rectangle (5.05, 4.0);
        \node[rectangle, inner sep=0, minimum height=8mm, minimum width=27.6mm, rounded corners=3mm, densely dashed] (triple1) at ($(s5)!0.5!(s3d2)$) {};
    \end{scope}
    \begin{scope}
            \clip(5.5, 0.0) rectangle (12, 4.0);
            \node[rectangle, inner sep=0, minimum height=8mm, minimum width=27.6mm, rounded corners=3mm, densely dashed] (triple1) at ($(s5)!0.5!(s3d2)$) {};
    \end{scope}
\end{scope}
\begin{scope}
    \foreach \from/\to in {aj2/aj1, aj1/ai, aj1/aj3, aj3/aj4, aj4/s4, s4/s5, s5/s6, s6/dots1, dots1/s3d2, s3d2/s3d1, s3d1/s3d, s3d/dots2, dots2/s3c2, s3c2/s3c1, s3c1/s3c, s3c/az2}
        \draw [thick] (\from) -- (\to);
\end{scope}
\end{tikzpicture}
    \caption{The structure of $M'$ in Case 2} 
    \label{fig:3d_sr_sas_bin_algorithm_7_cases_case_2}
\end{figure}

\begin{figure}
    \centering
    \begin{tikzpicture}
\begin{scope}[every node/.style={circle,draw, minimum size=2.4mm}, xscale=0.9]
    \node[thick, circle, label={[shift={(0.2, -1.29)}]:$\alpha_{i}$}] (ai) at (1,0.5) {};
    \node[thick, circle, label={[shift={(-0.4, 0)}]:$\alpha_{j_2}$}] (aj2) at (1,3.5) {};
    \node[thick, circle, label={[shift={(-0.55, 0.0)}]:$\alpha_{j_1}$}] (aj1) at (0.5,2) {};
    \node[thick, circle, label={[shift={(-0.1, -1.3)}]:$\alpha_{j_3}$}] (aj3) at (1.3,2) {};
    \node[thick, circle, label={[shift={(0, -1.3)}]:$\alpha_{j_4}$}] (aj4) at (2.1,2) {};
    
    \node[thick, circle, label={[shift={(0.0, -1.4)}]:$\alpha_{y_1}$}] (ay1) at (0.2,0.5) {};
    
    \node[thick, circle, label={[shift={(0.4, -1.5)}]}] (s4) at (2.9,2) {};
    \node[thick, circle, label={[shift={(0.4, -1.5)}]}] (s5) at (3.7,2) {};
    \node[thick, circle, label={[shift={(0.4, -1.5)}]}] (s6) at (4.5,2) {};
    
    \node[draw=none, inner sep=0.5mm] (dots1) at (5.3,2) {$\dots$};
    
    \node[thick, circle, label={[shift={(0.1, -0.1)}]:$S_{3d-2}$}] (s3d2) at (6.1,2) {};
    \node[thick, circle, label={[shift={(0.3, -0.1)}]:$S_{3d-1}$}] (s3d1) at (6.9,2) {};
    \node[thick, circle, label={[shift={(0.5, -0.1)}]:$S_{3d\phantom{-3}}$}] (s3d) at (7.7,2) {};
    
    \node[draw=none, inner sep=0.5mm] (dots2) at (8.5,2) {$\dots$};
    
    \node[thick, circle, label={[shift={(0.1, -1.5)}]}] (s3c2) at (9.3,2) {};
    \node[thick, circle, label={[shift={(0.3, -1.5)}]}] (s3c1) at (10.1,2) {};
    \node[thick, circle, label={[shift={(0.4, -0.1)}]:$S_{3c\phantom{-1}}$}] (s3c) at (10.9,2) {};
    
    \node[draw=none, inner sep=0, minimum size=0] (s3by) at (10.4, 0.5]) {};
    
    \draw [thick, rounded corners=4mm] (ai.east)--(s3by.center)--(s3c);

    \node[rectangle, inner sep=0, minimum height=8mm, minimum width=20.7mm, rounded corners=3mm, densely dashed] (triple1) at (s4) {};
    
    \begin{scope}
        \clip(0,1) rectangle (5.05, 4.0);
        \node[rectangle, inner sep=0, minimum height=8mm, minimum width=27.6mm, rounded corners=3mm, densely dashed] (triple1) at ($(s6)!0.5!(s3d1)$) {};
    \end{scope}
    \begin{scope}
            \clip(5.5, 0.0) rectangle (12, 4.0);
            \node[rectangle, inner sep=0, minimum height=8mm, minimum width=27.6mm, rounded corners=3mm, densely dashed] (triple1) at ($(s6)!0.5!(s3d1)$) {};
    \end{scope}
    
    \begin{scope}
        \clip(0,1) rectangle (8.25, 4.0);
        \node[rectangle, inner sep=0, minimum height=8mm, minimum width=27.6mm, rounded corners=3mm, densely dashed] (triple1) at ($(s3d)!0.5!(s3c1)$) {};
    \end{scope}
    \begin{scope}
            \clip(8.7, 0.0) rectangle (12, 4.0);
            \node[rectangle, inner sep=0, minimum height=8mm, minimum width=27.6mm, rounded corners=3mm, densely dashed] (triple1) at ($(s3d)!0.5!(s3c1)$) {};
    \end{scope}
    
    \begin{scope}[yscale=-1, yshift=-4.0cm]
        \draw [rounded corners=3mm, densely dashed] (0.0, 2.4)--(1.66, 2.4)--(1.66, 1.6)--(1.45, 0.1)--(0.65, 0.1)--cycle;
        \draw [rounded corners=3mm, densely dashed] (-0.2, 3.1)--(-0.2,3.9)--(10.7,3.9)--(11.3,2.4)--(11.3,1.6)--(10.5, 1.6)--(10.5, 2.4)--(10.0, 3.1)--cycle;
    \end{scope}
\end{scope}
\begin{scope}
    \foreach \from/\to in {aj2/aj1, aj1/ai, aj1/aj3, aj3/aj4, aj4/s4, s4/s5, s5/s6, s6/dots1, dots1/s3d2, s3d2/s3d1, s3d1/s3d, s3d/dots2, dots2/s3c2, s3c2/s3c1, s3c1/s3c, ay1/ai}
        \draw [thick] (\from) -- (\to);
\end{scope}
\end{tikzpicture}
    \caption{The structure of $M'$ in Case 4} 
    \label{fig:3d_sr_sas_bin_algorithm_7_cases_case_4}
\end{figure}

\begin{figure}
    \centering
    \begin{tikzpicture}
\begin{scope}[every node/.style={circle,draw, minimum size=2.4mm}, xscale=0.9]
    \node[thick, circle, label={[label distance=1.8mm]-45:$\alpha_{i\phantom{j_2}}$}] (ai) at (1,0.5) {};
    \node[thick, circle, label={[shift={(0.0, 0.1)}]:$\alpha_{j_2}$}] (aj2) at (1,3.5) {};
    \node[thick, circle, label={[shift={(-0.55, -1.2)}]:$\alpha_{j_1}$}] (aj1) at (0.5,2) {};
    \node[thick, circle, label={[shift={(-0.1, -1.3)}]:$\alpha_{j_3}$}] (aj3) at (1.3,2) {};
    \node[thick, circle, label={[shift={(0, -1.3)}]:$\alpha_{j_4}$}] (aj4) at (2.1,2) {};
    
    \node[thick, circle, label={[shift={(0.0, 0.1)}]:$\alpha_{y_2}$}] (ay2) at (0.2,3.5) {};
    
    \node[thick, circle, label={[shift={(0.4, -1.5)}]}] (s4) at (2.9,2) {};
    \node[thick, circle, label={[shift={(0.4, -1.5)}]}] (s5) at (3.7,2) {};
    \node[thick, circle, label={[shift={(0.4, -1.5)}]}] (s6) at (4.5,2) {};
    
    \node[draw=none, inner sep=0.5mm] (dots1) at (5.3,2) {$\dots$};
    
    \node[thick, circle, label={[shift={(0.1, -1.5)}]:$S_{3d-2}$}] (s3d2) at (6.1,2) {};
    \node[thick, circle, label={[shift={(0.3, -1.5)}]:$S_{3d-1}$}] (s3d1) at (6.9,2) {};
    \node[thick, circle, label={[shift={(0.5, -1.5)}]:$S_{3d\phantom{-3}}$}] (s3d) at (7.7,2) {};
    
    \node[draw=none, inner sep=0.5mm] (dots2) at (8.5,2) {$\dots$};
    
    \node[thick, circle, label={[shift={(0.1, -1.5)}]}] (s3c2) at (9.3,2) {};
    \node[thick, circle, label={[shift={(0.3, -1.5)}]}] (s3c1) at (10.1,2) {};
    \node[thick, circle, label={[shift={(0.5, -1.5)}]:$S_{3c\phantom{-1}}$}] (s3c) at (10.9,2) {};
    
    \node[draw=none, inner sep=0, minimum size=0] (s3by) at (10.4, 3.5]) {};
    
    \draw [thick, rounded corners=3mm] (aj2.east)--(s3by.center)--(s3c);

    \node[rectangle, inner sep=0, minimum height=8mm, minimum width=20.7mm, rounded corners=3mm, densely dashed] (triple1) at (s4) {};
    
    \begin{scope}
        \clip(0,1) rectangle (5.05, 4.0);
        \node[rectangle, inner sep=0, minimum height=8mm, minimum width=27.6mm, rounded corners=4mm, densely dashed] (triple1) at ($(s6)!0.5!(s3d1)$) {};
    \end{scope}
    \begin{scope}
            \clip(5.5, 0.0) rectangle (12, 4.0);
            \node[rectangle, inner sep=0, minimum height=8mm, minimum width=27.6mm, rounded corners=3mm, densely dashed] (triple1) at ($(s6)!0.5!(s3d1)$) {};
    \end{scope}
    
    \begin{scope}
        \clip(0,1) rectangle (8.25, 4.0);
        \node[rectangle, inner sep=0, minimum height=8mm, minimum width=27.6mm, rounded corners=3mm, densely dashed] (triple1) at ($(s3d)!0.5!(s3c1)$) {};
    \end{scope}
    \begin{scope}
            \clip(8.7, 0.0) rectangle (12, 4.0);
            \node[rectangle, inner sep=0, minimum height=8mm, minimum width=27.6mm, rounded corners=3mm, densely dashed] (triple1) at ($(s3d)!0.5!(s3c1)$) {};
    \end{scope}
    
    \draw [rounded corners=3mm, densely dashed] (0.0, 2.4)--(1.66, 2.4)--(1.66, 1.6)--(1.45, 0.1)--(0.65, 0.1)--cycle;
    \draw [rounded corners=3mm, densely dashed] (-0.2, 3.1)--(-0.2,3.9)--(10.7,3.9)--(11.3,2.4)--(11.3,1.6)--(10.5, 1.6)--(10.5, 2.4)--(10.0, 3.1)--cycle;
\end{scope}
\begin{scope}
    \foreach \from/\to in {aj2/aj1, aj1/ai, aj1/aj3, aj3/aj4, aj4/s4, s4/s5, s5/s6, s6/dots1, dots1/s3d2, s3d2/s3d1, s3d1/s3d, s3d/dots2, dots2/s3c2, s3c2/s3c1, s3c1/s3c, ay2/aj2}
        \draw [thick] (\from) -- (\to);
\end{scope}
\end{tikzpicture}
    \caption{The structure of $M'$ in Case 5} 
    \label{fig:3d_sr_sas_bin_algorithm_7_cases_case_5}
\end{figure}
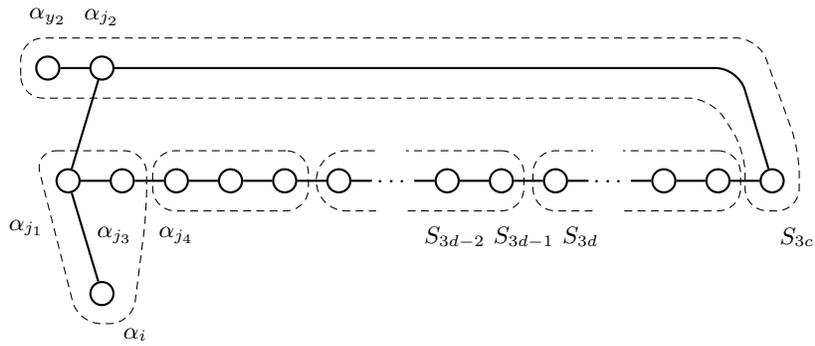

Consider $S_{3d-1}$ for $1\leq d \leq c$. Note that $u_{S_{3d-1}}(M)=2$ and $u_{S_{3d-1}}(M')=1$. Suppose for a contradiction that triple $\{ S_{3d-1}, \alpha_{k_1}, \alpha_{k_2} \}$ blocks $M'$ where $\alpha_{k_1}, \alpha_{k_2} \in N$. As before, since $u_{S_{3d-1}}(M')=1$ it follows that $u_{S_{3d-1}}(\{ \alpha_{k_1},\allowbreak \alpha_{k_2} \})=2$. Consider $\alpha_{k_1}$ and $\alpha_{k_2}$. Since $(N, V)$ is triangle-free, it must be that $u_{\alpha_{k_1}}(\{ S_{3d-1}, \alpha_{k_2} \}) = u_{\alpha_{k_2}}(\{ S_{3d-1}, \alpha_{k_1} \}) = 1$. It follows that $u_{\alpha_{k_1}}(M') = u_{\alpha_{k_2}}(M') = 0$. By construction of $M'$, no agent $\alpha_p \in N$ exists where $u_{\alpha_p}(M') = 0$ and $u_{\alpha_p}(M') < u_{\alpha_p}(M)$. It follows that $u_{\alpha_{k_1}}(M)=u_{\alpha_{k_2}}(M)=0$. Recall the $d\textsuperscript{th}$ iteration of the \texttt{while} loop. We have shown that two agents $\alpha_{k_1}, \alpha_{k_2}$ exist where $\mathit{val}_{S_{3d-1}}(\alpha_{k_1})=\mathit{val}_{S_{3d-1}}(\alpha_{k_2})=1$ and $u_{\alpha_{k_1}}(M)=u_{\alpha_{k_2}}(M)=0$. It follows that some $\alpha_{z_1} \in N\setminus \{\alpha_i\}$ exists where $\mathit{val}_{\alpha_{z_1}}(S_{3d-1})=1$ and $u_{\alpha_{z_1}}(M)=0$, since either $\alpha_{z_1}=\alpha_{k_1}$ or $\alpha_{i}=\alpha_{k_1}$ and $\alpha_{z_1}=\alpha_{k_2}$. In this iteration, since $\alpha_{z_1}\neq \bot$ the break condition held, the \texttt{while} loop terminated, and the condition for either Case 1 or Case 3 was true. This is a contradiction. It follows that no triple containing $S_{3d-1}$ blocks $M'$ for $1\leq d \leq c$. In summary, in Cases 2, 4, 5, and 6, no $\alpha_g\in N$ exists where $u_{\alpha_{g}}(M') < u_{\alpha_{g}}(M)$ and $\alpha_g$ belongs to a triple that blocks $M'$.
\end{proof}

\begin{lem}
\label{lem:algocase7noalphapexists}
In Case 7 of Algorithm~{\normalfont \texttt{repair}}, no agent $\alpha_{g}\in N$ exists where $u_{\alpha_{g}}(M') < u_{\alpha_{g}}(M)$ and $\alpha_g$ belongs to a triple that blocks $M'$.
\end{lem}
\begin{proof}

Refer to Figure~\ref{fig:3d_sr_sas_bin_algorithm_7_cases_case_7}. Suppose for a contradiction that some $\alpha_g$ exists as above.

\begin{figure}[ht]
    \centering
    \begin{tikzpicture}
\begin{scope}[every node/.style={circle,draw, minimum size=2.4mm}, xscale=0.9]
    \node[thick, circle, label={[label distance=1.8mm]-45:$\alpha_{i\phantom{j_2}}$}] (ai) at (1,0.5) {};
    \node[thick, circle, label={[shift={(-0.4, 0)}]:$\alpha_{j_2}$}] (aj2) at (1,3.5) {};
    \node[circle, densely dashed, minimum size=8mm] at (aj2) {};
    \node[thick, circle, label={[shift={(-0.55, -1.2)}]:$\alpha_{j_1}$}] (aj1) at (0.5,2) {};
    \node[thick, circle, label={[shift={(-0.1, -1.3)}]:$\alpha_{j_3}$}] (aj3) at (1.3,2) {};
    \node[thick, circle, label={[shift={(0, -1.3)}]:$\alpha_{j_4}$}] (aj4) at (2.1,2) {};
    
    \node[thick, circle, label={[shift={(0.4, -1.5)}]}] (s4) at (2.9,2) {};
    \node[thick, circle, label={[shift={(0.4, -1.5)}]}] (s5) at (3.7,2) {};
    \node[thick, circle, label={[shift={(0.4, -1.5)}]}] (s6) at (4.5,2) {};
    
    \node[draw=none, inner sep=0.5mm] (dots1) at (5.3,2) {$\dots$};
    
    \node[thick, circle, label={[shift={(0.1, -1.5)}]:$S_{3d-2}$}] (s3d2) at (6.1,2) {};
    \node[thick, circle, label={[shift={(0.3, -1.5)}]:$S_{3d-1}$}] (s3d1) at (6.9,2) {};
    \node[thick, circle, label={[shift={(0.5, -1.5)}]:$S_{3d\phantom{-3}}$}] (s3d) at (7.7,2) {};
    
    \node[draw=none, inner sep=0.5mm] (dots2) at (8.5,2) {$\dots$};
    
    \node[thick, circle, label={[shift={(0.1, -1.5)}]}] (s3c2) at (9.3,2) {};
    \node[thick, circle, label={[shift={(0.3, -1.5)}]}] (s3c1) at (10.1,2) {};
    \node[thick, circle, label={[shift={(0.5, -1.5)}]:$S_{3c\phantom{-1}}$}] (s3c) at (11.0,2) {};

    \node[circle, densely dashed, minimum size=8mm] at (s3c) {};
    \node[rectangle, inner sep=0, minimum height=8mm, minimum width=20.7mm, rounded corners=3mm, densely dashed] (triple1) at (s4) {};
    
    \begin{scope}
        \clip(0,1) rectangle (5.05, 4.0);
        \node[rectangle, inner sep=0, minimum height=8mm, minimum width=27.6mm, rounded corners=3mm, densely dashed] (triple1) at ($(s6)!0.5!(s3d1)$) {};
    \end{scope}
    \begin{scope}
            \clip(5.5, 0.0) rectangle (12, 4.0);
            \node[rectangle, inner sep=0, minimum height=8mm, minimum width=27.6mm, rounded corners=3mm, densely dashed] (triple1) at ($(s6)!0.5!(s3d1)$) {};
    \end{scope}
    
    \begin{scope}
        \clip(0,1) rectangle (8.25, 4.0);
        \node[rectangle, inner sep=0, minimum height=8mm, minimum width=27.6mm, rounded corners=3mm, densely dashed] (triple1) at ($(s3d)!0.5!(s3c1)$) {};
    \end{scope}
    \begin{scope}
            \clip(8.7, 0.0) rectangle (12, 4.0);
            \node[rectangle, inner sep=0, minimum height=8mm, minimum width=27.6mm, rounded corners=3mm, densely dashed] (triple1) at ($(s3d)!0.5!(s3c1)$) {};
    \end{scope}
    
    \draw [rounded corners=3mm, densely dashed] (0.0, 2.4)--(1.66, 2.4)--(1.66, 1.6)--(1.45, 0.1)--(0.65, 0.1)--cycle;

\end{scope}
\begin{scope}
    \foreach \from/\to in {aj2/aj1, aj1/ai, aj1/aj3, aj3/aj4, aj4/s4, s4/s5, s5/s6, s6/dots1, dots1/s3d2, s3d2/s3d1, s3d1/s3d, s3d/dots2, dots2/s3c2, s3c2/s3c1, s3c1/s3c}
        \draw [thick] (\from) -- (\to);
\end{scope}
\end{tikzpicture}
    \caption{The structure of $M'$ in Case 7} 
    \label{fig:3d_sr_sas_bin_algorithm_7_cases_case_7}
\end{figure}
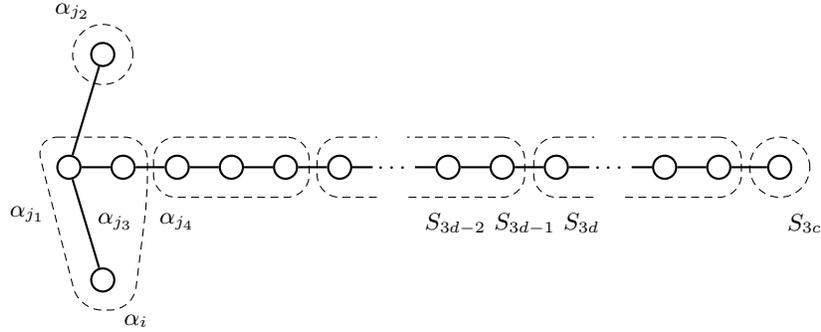

First, consider any agent $\alpha_p\in N$ where $\alpha_p\notin S \cup \{ \alpha_{j_2}, \alpha_{i} \}$. By the construction of $M'$, it can be seen that $M(\alpha_p)=M'(\alpha_p)$ so $u_{\alpha_p}(M)=u_{\alpha_p}(M')$ and hence $\alpha_g \notin S \cup \{ \alpha_{j_2}, \alpha_{i} \}$.

Now consider $\alpha_i$ and $\alpha_{j_2}$. Since $u_{\alpha_i}(M)=0$ and $u_{\alpha_i}(M')=1$ it follows that $\alpha_g \neq \alpha_i$.  Similarly, since $u_{\alpha_{j_2}}(M)=u_{\alpha_{j_2}}(M')=0$ it follows that $\alpha_g \neq \alpha_{j_2}$.

It remains that $\alpha_g \in S$.

Consider $S_{3d-2}$ for $1\leq d\leq c$. By construction of $M'$ it follows that $u_{S_{3d-2}}(M')=2$ so $\alpha_{p}\neq S_{3d-2}$.

Consider $S_{3d-1}$ for $1\leq d \leq c$. Suppose for a contradiction that triple $\{ S_{3d-1}, \alpha_{k_1}, \alpha_{k_2} \}$ blocks $M'$ where $\alpha_{k_1}, \alpha_{k_2} \in N$. Since $u_{S_{3d-1}}(M')=1$ it follows that $u_{S_{3d-1}}(\{ \alpha_{k_1}, \alpha_{k_2} \})=2$. Consider $\alpha_{k_1}$ and $\alpha_{k_2}$. Since $(N, V)$ is triangle-free, it must be that $u_{\alpha_{k_1}}(\{ S_{3d-1}, \alpha_{k_2} \})=u_{\alpha_{k_2}}(\{ S_{3d-1}, \alpha_{k_1} \})=1$. It follows that $u_{\alpha_{k_1}}(M')=u_{\alpha_{k_2}}(M')=0$. By construction of $M'$ it can be seen that $\alpha_p=S_{3c}$ is the only $\alpha_p\in N$ where $u_{\alpha_p}(M')=0$ and $u_{\alpha_p}(M') < u_{\alpha_p}(M)$. It follows that either $u_{\alpha_{k_1}}(M)=0$, $u_{\alpha_{k_2}}(M)=0$, or both. Suppose without loss of generality that $u_{\alpha_{k_1}}(M)=0$. Since $u_{\alpha_{k_1}}(M')=0$ it follows that $\alpha_{k_1}\neq \alpha_i$. Recall the $d\textsuperscript{th}$ iteration of the \texttt{while} loop. Since $\mathit{val}_{S_{3d-1}}(\alpha_{k_1})=1$, $u_{\alpha_{k_1}}(M)=0$, and $\alpha_{k_1}\neq \alpha_i$, it follows that there exists some $\alpha_{z_1} \in N\setminus \{\alpha_i\}$, namely $\alpha_{k_1}$, where $\mathit{val}_{\alpha_{z_1}}(S_{3d-1})=1$ and $u_{\alpha_{z_1}}(M)=0$. In this iteration, since $\alpha_{z_1}\neq \bot$ the break condition held, the \texttt{while} loop terminated, and either the condition for Case 1 was true or the condition for Case 3 was true. This is a contradiction. It follows that $S_{3d-1}$ does not belong to a triple that blocks $M'$ for any $1\leq d \leq c$ and hence that $\alpha_{g}\neq S_{3d-1}$. 

Consider $S_{3d}$ for $1\leq d < c$. By construction of $M'$ it follows that $u_{S_{3d}}(M')=u_{S_{3d}}(M)=1$ so $\alpha_{g}\neq S_{3d}$.

It remains to consider $S_{3c}$. As before, suppose for a contradiction that $\alpha_{k_1},\allowbreak \alpha_{k_2}\in N$ exist where $\{ S_{3c}, \alpha_{k_1}, \alpha_{k_2} \}$ blocks $M'$.  Since $\{ S_{3c}, \alpha_{k_1}, \alpha_{k_2} \}$ blocks $M'$ and $u_{S_{3c}}(M')=0$ it must be that either $\mathit{val}_{S_{3c}}(\alpha_{k_1})=1$ or $\mathit{val}_{S_{3c}}(\alpha_{k_2})=1$ or both.

Suppose that both $\mathit{val}_{S_{3c}}(\alpha_{k_1})=1$ and $\mathit{val}_{S_{3c}}(\alpha_{k_2})=1$ and hence $u_{S_{3c}}(\{ \alpha_{k_1},\allowbreak \alpha_{k_2} \})=2$. Since $(N, V)$ is triangle-free, it must be that $u_{\alpha_{k_1}}(\{ S_{3c}, \alpha_{k_2} \}) = u_{\alpha_{k_2}}(\{ S_{3c}, \alpha_{k_1} \}) = 1$. Since this triple blocks $M'$ it must be that $u_{\alpha_{k_1}}(M')=u_{\alpha_{k_2}}(M')=0$. By the construction of $M'$ it can be seen that $\alpha_p=S_{3c}$ is the only $\alpha_p\in N$ where $u_{\alpha_p}(M')=0$ and $u_{\alpha_p}(M') < u_{\alpha_p}(M)$. It follows that $u_{\alpha_{k_1}}(M)=u_{\alpha_{k_2}}(M)=0$. Note that since $u_{\alpha_i}(M')=1$ it follows that $\alpha_{k_1}\neq \alpha_i$ and $\alpha_{k_2}\neq \alpha_i$. It follows that either $\alpha_{k_1}\in N \setminus \{ \alpha_i, \alpha_{j_2} \}$, $\alpha_{k_2}\in N \setminus \{ \alpha_i, \alpha_{j_2} \}$, or both. Without loss of generality assume that $\alpha_{k_1}\in N \setminus \{ \alpha_i, \alpha_{j_2} \}$. In summary, there exists some $\alpha_{z_2} \in N\setminus \{\alpha_i, \alpha_{j_2} \}$, namely $\alpha_{k_1}$, where $\mathit{val}_{\alpha_{z_2}}(S_{3c})=1$ and $u_{\alpha_{z_2}}(M)=0$. In the algorithm, since $\alpha_{z_2}\neq \bot$ the condition of Case 2 holds. This is a contradiction.

The remaining possibility is that either $\mathit{val}_{S_{3c}}(\alpha_{k_1})=1$ or $\mathit{val}_{S_{3c}}(\alpha_{k_2})=1$ but not both. Suppose without loss of generality that $\mathit{val}_{S_{3c}}(\alpha_{k_1})=1$ and $\mathit{val}_{S_{3c}}(\alpha_{k_2})=0$. It follows that $u_{\alpha_{k_2}}(\{ S_{3c}, \alpha_{k_1} \})=1$ and hence $u_{\alpha_{k_2}}(M')=0$. Since $\alpha_p=S_{3c}$ is the only $\alpha_p\in N$ where $u_{\alpha_p}(M')=0$ and $u_{\alpha_p}(M') < u_{\alpha_p}(M)$, it follows that $u_{\alpha_{k_2}}(M)=0$. It must be that $\mathit{val}_{\alpha_{k_1}}(\alpha_{k_2})=1$ since $u_{\alpha_{k_2}}(\{ S_{3c}, \alpha_{k_1} \})=1$ and $\mathit{val}_{S_{3c}}(\alpha_{k_2})=0$. In summary, since $\mathit{val}_{S_{3c}}(\alpha_{k_1})=1$ and $\mathit{val}_{\alpha_{k_1}}(\alpha_{k_2})=1$ it follows that $u_{\alpha_{k_1}}(\{ S_{3c}, \alpha_{k_2} \})=2$.

We have shown that $u_{\alpha_{k_1}}(\{ S_{3c}, \alpha_{k_2} \})=2$. Either $u_{\alpha_{k_1}}(M')=1$ or $u_{\alpha_{k_1}}(M')=0$. Suppose for a contradiction that $u_{\alpha_{k_1}}(M')=0$. Since $\alpha_p=S_{3c}$ is the only $\alpha_p\in N$ where $u_{\alpha_p}(M')=0$ and $u_{\alpha_p}(M') < u_{\alpha_p}(M)$ it follows that $u_{\alpha_{k_1}}(M)=0$. Consider two further possibilities. First, that $\alpha_{k_1}=\alpha_{j_2}$. Second, that $\alpha_{k_1} \neq \alpha_{j_2}$. In the first, since $\alpha_{k_1}=\alpha_{j_2}$ then there exists some $\alpha_{y_2}\in N$, namely $\alpha_{k_2}$, where $\mathit{val}_{\alpha_{S_{3c}}}(\alpha_{j_2})=\mathit{val}_{\alpha_{y_2}}(\alpha_{j_2})=1$ and $u_{\alpha_{y_2}}(M)=0$. In the algorithm, since $\alpha_{y_2}\neq \bot$ the condition of Case 5 holds. This is a contradiction. Consider the second possibility that $\alpha_{k_1} \neq \alpha_{j_2}$. Since $u_{\alpha_i}(M')=1$ it follows that $\alpha_i\neq \alpha_{k_1}$ and hence there exists some $\alpha_{z_2} \in N\setminus \{\alpha_i, \alpha_{j_2} \}$, namely $\alpha_{k_1}$, where $\mathit{val}_{\alpha_{z_2}}(S_{3c})=1$ and $u_{\alpha_{z_2}}(M)=0$. In the algorithm, since $\alpha_{z_2}\neq \bot$ the condition of Case 2 holds. This is also a contradiction. It remains that $u_{\alpha_{k_1}}(M')=1$.

In summary, we supposed that a triple $\{ S_{3c}, \alpha_{k_1}, \alpha_{k_2} \}$ blocks $M'$. We showed that $\mathit{val}_{S_{3c}}(\alpha_{k_1})=\mathit{val}_{\alpha_{k_1}}(\alpha_{k_2})=1$, $u_{\alpha_{k_2}}(M')=u_{\alpha_{k_2}}(M)=0$, and $u_{\alpha_{k_1}}(M')=1$. This is illustrated in Figure~\ref{fig:3d_sr_sas_bin_algorithm_proof_case_7_explanation_1}.

\begin{figure}[ht]
    \centering
    \begin{tikzpicture}
\begin{scope}[every node/.style={circle,draw, minimum size=2.4mm}]
    
    \node[draw=none, inner sep=0.5mm] (dots1) at (0.5,2) {$\dots$};
    
    \node[thick, circle, label={[shift={(0.1, -1.7)}]:$S_{3c-2}$}] (s3c2) at (1.3,2) {};
    \node[thick, circle, label={[shift={(0.3, -1.7)}]:$S_{3c-1}$}] (s3c1) at (2.1,2) {};
    \node[thick, circle, label={[shift={(0.5, -1.7)}]:$S_{3c\phantom{-3}}$}] (s3c) at (2.9,2) {};
    
    \node[thick, circle, label={[shift={(-0.5, 0.1)}]:$\alpha_{k_1}$}] (ak1) at (4.6,2) {};
    
     \node[thick, circle] (ak12) at (5.232,1.367) {};
     
     \node[thick, circle, draw=none, rotate=-45] (ak13) at (5.864, 0.734) {$\dots$};
    
    \node[thick, circle, label={[shift={(0.5, 0.1)}]:$\alpha_{k_2}$}] (ak2) at (5.6,3.6) {};
    \node[circle, densely dashed, minimum size=8mm] at (s3c) {};
     \node[circle, densely dashed, minimum size=8mm] at (ak2) {};
    \begin{scope}[rotate=-45]
        \clip(0,1) rectangle (3.6, 6.0);
        \node[rectangle, inner sep=0, minimum height=8mm, minimum width=27mm, rounded corners=4mm, densely dashed, transform shape] (triple1) at (ak12) {};
    \end{scope}
    
    \begin{scope}
        \clip(0,1) rectangle (5.05, 4.0);
        \node[rectangle, inner sep=0, minimum height=8mm, minimum width=30.67mm, rounded corners=4mm, densely dashed] (triple2) at ($(dots1)!0.5!(s3c2)$) {};
    \end{scope}
    \begin{scope}
            \clip(5.5, 0.0) rectangle (5, 4.0);
            \node[rectangle, inner sep=0, minimum height=8mm, minimum width=30.67mm, rounded corners=4mm, densely dashed] (triple3) at ($(dots1)!0.5!(s3c2)$) {};
    \end{scope}
\end{scope}
\begin{scope}
    \foreach \from/\to in {dots1/s3c2, s3c2/s3c1, s3c1/s3c, s3c/ak1, ak1/ak2, ak1/ak12, ak12/ak13}
        \draw [thick] (\from) -- (\to);
\end{scope}
\end{tikzpicture}
    \vspace*{-42pt}
    \caption{In Lemma~\ref{lem:algocase7noalphapexists} we consider $M'$ in Case 7. We suppose for a contradiction that some triple $\{ S_{3c}, \alpha_{k_1}, \alpha_{k_2} \}$ blocks $M'$ where $\alpha_{k_1}, \alpha_{k_2}\in N$. We then show that $\mathit{val}_{S_{3c}}(\alpha_{k_1})=\mathit{val}_{\alpha_{k_1}}(\alpha_{k_2})=1$, $u_{\alpha_{k_2}}(M')=u_{\alpha_{k_2}}(M)=0$, and $u_{\alpha_{k_1}}(M')=1$. We then show that this is a contradiction, and conclude that no such $\alpha_{k_1}, \alpha_{k_2}$ exist. This shows that $S_{3c}$ does not belong to a triple that blocks $M'$.}
    \label{fig:3d_sr_sas_bin_algorithm_proof_case_7_explanation_1}
\end{figure}
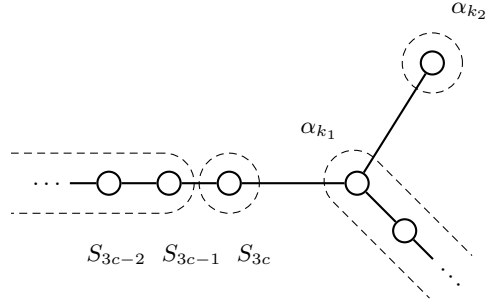

By the condition of Case 7, in the algorithm $\alpha_{w_1}=\bot$. This means that no $\alpha_{w_1}\in N$ exists where $\mathit{val}_{S_{3c}}(\alpha_{w_1})=1$, $u_{\alpha_{w_1}}(M)=1$,  $\alpha_{w_1}\notin S$, and there exists $\alpha_{z_3}\in N\setminus \{ \alpha_i \}$ where $\mathit{val}_{\alpha_{z_3}}(\alpha_{w_1})=1$ and $u_{\alpha_{z_3}}(M)=0$. If $u_{\alpha_{k_1}}(M)=1$ and $\alpha_{k_1}\notin S$ then some $\alpha_{w_1},\alpha_{z_3}$, namely $\alpha_{k_1}, \alpha_{k_2}$, exist, which is a contradiction. It must be that either $u_{\alpha_{k_1}}(M) \neq 1$ or $\alpha_{k_1}\in S$ or both.

Suppose that $\alpha_{k_1}\notin S$ and hence $u_{\alpha_{k_1}}(M)\neq 1$. Recall that $u_{\alpha_{k_1}}(M')=1$. By construction of $M'$ in Case 7, $u_{\alpha_p}(M')=u_{\alpha_p}(M)$ for any $\alpha_p \in N \setminus (S \cup \{ \alpha_i \})$. It follows that $\alpha_{k_1} \in S \cup \{ \alpha_i \}$. By assumption, $\alpha_{k_1}\notin S$, so it must be that $\alpha_{k_1}=\alpha_i$. In this case, there exists some $\alpha_{y_1}\in N$, namely $\alpha_{k_2}$, where $\mathit{val}_{S_{3c}}(\alpha_i)=\mathit{val}_{\alpha_i}(\alpha_{y_1})=1$ and $u_{\alpha_{y_1}}(M)=0$. It follows that, in the algorithm, the condition for Case 4 is true. This is a contradiction. It must be that $\alpha_{k_1}\in S$.

From above, $u_{\alpha_{k_1}}(M')=1$. Since $u_{S_{3d-2}}(M')=2$ for any $1\leq d\leq c$ it follows that $\alpha_{k_1}\neq S_{3d-2}$ for any $1\leq d\leq c$. It follows that either $\alpha_{k_1}=S_{3d-1}$ or $\alpha_{k_1}=S_{3d}$ for some $1\leq d\leq c$.

Suppose that $\alpha_{k_1}=S_{3d-1}$ for some $1\leq d\leq c$. Recall the $d\textsuperscript{th}$ iteration of the \texttt{while} loop. There exists some $\alpha_{z_1} \in N\setminus \{ \alpha_i \}$, namely $\alpha_{k_2}$, where $\mathit{val}_{\alpha_{z_1}}(S_{3d-1})=1$ and $u_{\alpha_{z_1}}(M)=0$. It follows that in the algorithm, after the $d\textsuperscript{th}$ iteration of the \texttt{while} loop, $\alpha_{z_1}\neq \bot$, hence $d=c$ and either the condition for Case 1 was true or the condition for Case 3 was true. Both cases are a contradiction. It follows that no such $\alpha_{k_1}\neq S_{3d-1}$ for any $1\leq d\leq c$.

Finally, suppose that $\alpha_{k_1}=S_{3d}$ for some $1\leq d \leq c$. From above, $S_{3c}\neq \alpha_{k_1}$ so $d < c$. Recall the $d\textsuperscript{th}$ iteration of the \texttt{while} loop. Since $\mathit{val}_{S_{3d}}(\alpha_{k_2})=1$, $u_{\alpha_{k_2}}(M)=0$, and $\alpha_{k_2}\neq \alpha_i$, since $u_{\alpha_i}(M')=1$, it follows that there exists some $\alpha_{z_2} \in N\setminus \{ \alpha_i \}$, namely $\alpha_{k_2}$, where $\mathit{val}_{\alpha_{z_2}}(S_{3d})=1$ and $u_{\alpha_{z_2}}(M)=0$. There are two possibilities. The first is that $\alpha_{k_2}\neq \alpha_{j_2}$. The second is that $\alpha_{k_2}=\alpha_{j_2}$. Suppose first that $\alpha_{k_2}\neq \alpha_{j_2}$. There exists some $\alpha_{z_2} \in N\setminus \{ \alpha_i, \alpha_{j_2} \}$, namely $\alpha_{k_2}$, where $\mathit{val}_{\alpha_{z_2}}(S_{3d})=1$ and $u_{\alpha_{z_2}}(M)=0$. It follows that in the algorithm $\alpha_{z_2}\neq \bot$, the break condition held, and the \texttt{while} loop terminated after this iteration. This is a contradiction since $d < c$. Suppose then that $\alpha_{k_2}=\alpha_{j_2}$. It follows that there exists some $b=d$ where $1\leq b < c$ and $\mathit{val}_{S_{3b}}(\alpha_{j_2})=\mathit{val}_{S_{3c}}(S_{3b})=1$. It follows that, after the final iteration of the \texttt{while} loop, the condition for Case 6 is true, which is a contradiction. In summary, we have shown that $\alpha_{k_1}\neq S_{3d}$ for every $1\leq d\leq c$.

To recap, we supposed that some $\alpha_g$ exists where $u_{\alpha_{g}}(M') < u_{\alpha_{g}}(M)$ and $\alpha_g$ belongs to a triple that blocks $M'$. We first showed that $\alpha_g\in S$. We then showed that $\alpha_g \neq S_{3d-2}$ and $\alpha_g \neq S_{3d-1}$ for any $1\leq d\leq c$. We then showed that $\alpha_g \neq S_{3d}$ for any $1\leq d < c$. We concluded that $\alpha_g=S_{3c}$. We supposed that some $\alpha_{k_1}, \alpha_{k_2}\in N$ exist where $\{ S_{3c}, \alpha_{k_1}, \alpha_{k_2} \}$ blocks $M'$. We then showed that $\alpha_{k_1}\in S$. Finally we showed that $\alpha_{k_1}\neq S_{3d-2}$ for any $1\leq d\leq c$, that $\alpha_{k_1}\neq S_{3d-1}$ for any $1\leq d\leq c$, and that $\alpha_{k_1}\neq S_{3d}$ for any $1\leq d < c$. This contradicts $\alpha_{k_1}\in S$ and it follows that no such $\alpha_{k_1}$ exists where $\{ S_{3c}, \alpha_{k_1}, \alpha_{k_2} \}$ blocks $M'$. This shows that $S_{3c}$ does not belong to a triple that blocks $M'$ and hence $\alpha_g \neq S_{3c}$.
\end{proof}

\begin{lem}
\label{lem:algoreturnsstablematching_notimecomplex}
Algorithm~{\normalfont \texttt{repair}} returns a stable $P$\nobreakdash-matching $M'$.
\end{lem}
\begin{proof}
By Lemma~\ref{lem:algalwaysterminates} the algorithm terminates after at most $\lfloor {(n-2)}/{3} \rfloor$ iterations of the main loop. By Lemma~\ref{lem:algreturnspmatching} the algorithm returns a $P$\nobreakdash-matching.

Suppose $M'$ is a $P$\nobreakdash-matching returned by the algorithm. By Lemmas~\ref{lem:algocases1and3noalphapexists}, \ref{lem:algocases245and6noalphapexists}, and~\ref{lem:algocase7noalphapexists}, in Cases 1, 2, 3, 4, 5, 6, and 7, no $\alpha_g \in N$ exists where $u_{\alpha_{g}}(M') < u_{\alpha_{g}}(M)$ and $\alpha_g$ belongs to a triple that blocks $M'$.

Suppose for a contradiction that $M'$ is not stable and some triple $\{ \alpha_{k_1}, \alpha_{k_2},\allowbreak \alpha_{k_3} \}$ blocks $M'$. It follows that $u_{\alpha_{k_r}}(M') \geq u_{\alpha_{k_r}}(M)$ for $1 \leq r \leq 3$, otherwise $\alpha_g$ exists as described above. By Lemma~\ref{lem:3dsrsasbinblockerimprovement}, it follows that $\{ \alpha_{k_1}, \alpha_{k_2}, \alpha_{k_3} \}$ also blocks $M$. By assumption, any triple that blocks $M$ contains $\alpha_i$ so assume without loss of generality that $\alpha_{k_1}=\alpha_i$.

In Case 4, $u_{\alpha_i}(M')=2$ and hence $\alpha_i$ does not belong to a triple that blocks $M'$. This is a contradiction. It follows that no triple blocks $M'$ and that $M'$ is stable in Case 4.

In Cases 1, 2, 3, 5, 6, and 7, $u_{\alpha_i}(M')=1$. It follows that $u_{\alpha_i}(\{ \alpha_{k_2}, \alpha_{k_3} \})=2$ so $\mathit{val}_{\alpha_i}(\alpha_{k_2})=\mathit{val}_{\alpha_i}(\alpha_{k_3})=1$. Since $(N, V)$ is triangle-free, it must be that $\mathit{val}_{\alpha_{k_2}}(\alpha_{k_3})=0$ and hence $u_{\alpha_{k_2}}(\{ \alpha_{i}, \alpha_{k_3} \})=u_{\alpha_{k_3}}(\{ \alpha_{i}, \alpha_{k_2} \})=1$. Since $\{ \alpha_{i}, \alpha_{k_2}, \alpha_{k_3} \}$ blocks $M$, It follows that $u_{\alpha_{k_2}}(M)=u_{\alpha_{k_3}}(M)=0$, and thus that $\{ \alpha_i, \alpha_{k_1}, \alpha_{k_2} \}$ blocks $M$. This contradicts our original assumption that any triple that blocks $M$ contains $\{ \alpha_i, \alpha_{j_1}, \alpha_{j_2}\}$ where $\alpha_{j_1}, \alpha_{j_2}\in N$, $u_{\alpha_{j_1}}(M)=1$, $u_{\alpha_{j_2}}(M)=0$, and $\mathit{val}_{\alpha_i}(\alpha_{j_1})=\mathit{val}_{\alpha_i}(\alpha_{j_2})=1$. It follows that no triple blocks $M'$.
\end{proof}

\begin{lem}
\label{lem:3dsrsasbin_almosttherealgo_runningtime}
Algorithm~{\normalfont \texttt{repair}} has running time $O(|N|^2)$.
\end{lem}
\begin{proof}
The pseudocode above outlines the algorithm at a high level. To analyse the worst-case time complexity we describe a suitable system of data structures, which we combine with a preprocessing step. Relying on the unit cost of standard operations in these data structures, we analyse the worst case time complexity of Algorithm~{\normalfont \texttt{repair}} in terms of $|N|$.

Suppose that $(N, V)$ is stored such that, for a given $\alpha_{p}\in N$, the algorithm can iterate through the set $\{ \alpha_{q} \in N : \mathit{val}_{\alpha_{p}}(\alpha_{q})=1 \}$ in $O(|N|)$ time. Suppose that $M$ is stored such that the algorithm can iterate through each triple in $O(|N|)$ time. For example, $(N, V)$ could be stored graphically using adjacency lists. It follows that, given three agents $\alpha_{h_1}, \alpha_{h_2}, \alpha_{h_3}\in N$ the algorithm can compute $u_{\alpha_{h_1}}(\{ \alpha_{h_2}, \alpha_{h_3} \}), u_{\alpha_{h_2}}(\{ \alpha_{h_1}, \alpha_{h_3} \}),$ and $u_{\alpha_{h_3}}(\{ \alpha_{h_1}, \alpha_{h_2} \})$ in $O(|N|)$ time.

The preprocessing step involves constructing two lookup tables. Each lookup table contains exactly $|N|$ entries and is indexed by some $\alpha_{p}\in N$. Each entry in each table contains some integer less than or equal to $|N|$. It follows that finding an entry given its index requires constant time.  
Each entry in $L_1$ will contain either zero, one, or two. For each agent $\alpha_{p}\in N$, the algorithm constructs $L_1$ so that the ${p}\textsuperscript{th}$ entry contains $u_{\alpha_p}(M)$. By assumption, the algorithm can compute $u_{\alpha_p}(M)$ for any $\alpha_p\in N$ in $O(|N|)$ time. It follows that $L_1$ can be constructed in $O(|N|^2)$ time by iterating through $M$ and computing $u_{\alpha_{h_1}}(M), u_{\alpha_{h_2}}(M), u_{\alpha_{h_3}}(M)$ for each $\{ \alpha_{h_1}, \alpha_{h_2}, \alpha_{h_3} \} \in M$. Since $|M|=O(|N|)$ in total this step takes $O(|N|^2)$ time. It follows that we can use $L_1$ to look up $u_{\alpha_{p}}(M)$ for any $\alpha_{p}\in N$ in constant time. Each entry in $L_2$ contains either the label of some agent or $\bot$. Construct $L_2$ such that for any $\alpha_p\in N$, the ${p}\textsuperscript{th}$ entry either contains some $\alpha_{q}\in N \setminus \{ \alpha_i \}$ where $\mathit{val}_{\alpha_{p}}(\alpha_{q})=1$ and $u_{\alpha_{q}}(M)=0$ if it exists and otherwise $\bot$. The algorithm will use $L_2$ primarily in the body of the loop to identify $\alpha_{w_1}$, if it exists, using $S_{3c}$. The lookup table $L_2$ can be constructed in $O(|N|^2)$ time, as follows. For each $\alpha_{p}\in N$, look up $u_{\alpha_{p}}(M)$ in $L_1$. If $u_{\alpha_{p}}(M)=0$ then consider each $\alpha_{q}\in N$ where $\mathit{val}_{\alpha_{p}}(\alpha_{q})=1$ and $\alpha_{q}\neq \alpha_i$. If the $q\textsuperscript{th}$ entry of $L_2$ is currently $\bot$ then set that entry to $\alpha_{p}$.

The list $S$ can be stored using a linked list or any data structure in which a new element can be appended to the end of $S$ in constant time and the iteration through $S$ takes $O(|N|)$ time. The list $S$ will be supplemented with a lookup table $L_S$. For any $\alpha_p \in N$, the table $L_S$ can be used to test membership in $S$ and look up the position of any agent in $S$ in constant time. This is possible because the only modification that the algorithm makes to $S$ is appending a single agent to the end of $S$ in each iteration. As noted in Lemma~\ref{lem:algalwaysterminates}, any agent is added to $S$ at most than once. Like the tables $L_1$ and $L_2$, the table $L_S$ contains exactly $|N|$ entries and is indexed by each $\alpha_{p}\in N$. Each entry in $L_S$ contains some integer position less than or equal to $|S|$. Before the algorithm appends an element $\alpha_p\in N$ to the end of $S$, it can maintain $L_S$ in constant time by setting the $p\textsuperscript{th}$ entry to $|S|$.

The first step in the algorithm involves identifying agents $\alpha_{j_1}, \alpha_{j_2}$ where $\{ \alpha_i, \alpha_{j_1}, \alpha_{j_2}\}$ blocks $M$ in $(N, V)$ and $u_{\alpha_{j_1}}(M)=1$ as follows. Given any $\alpha_{j_1}, \alpha_{j_2}\in N$ where $u_{\alpha_{j_1}}(M)=1$, $u_{\alpha_{j_2}}(M)=0$ and $\mathit{val}_{\alpha_i}(\alpha_{j_1})=\mathit{val}_{\alpha_i}(\alpha_{j_2})=1$, the triple $\{ \alpha_i, \alpha_{j_1}, \alpha_{j_2}\}$ blocks $M$ in $(N, V)$. It follows that some $\alpha_{j_1}, \alpha_{j_2}\in N$ can be found in $O(|N|)$ time, as follows. Consider each agent $\alpha_{p}$ for which $\mathit{val}_{\alpha_i}(\alpha_{p})=1$, and look up $u_{\alpha_{p}}(M)$ in $L_1$. If $u_{\alpha_{p}}(M)=1$ then look up the ${p}\textsuperscript{th}$ entry of $L_2$. By the construction of $L_2$, if this entry is not equal to $\bot$ then it contains some $\alpha_{q}\in N\setminus \{ \alpha_i \}$ where $\mathit{val}_{\alpha_{p}}(\alpha_{q})=1$ and $u_{\alpha_{q}}(M)=0$. In this case let $\alpha_{j_1}=\alpha_{p}$ and $\alpha_{j_2}=\alpha_{q}$. Since $M$ is not stable in $(N, V)$, by the condition of $M$ there must exist some such $\alpha_{j_1}, \alpha_{j_2}$.

 The second step in the algorithm involves identifying agents $\alpha_{j_3}, \alpha_{j_4}$ where $\alpha_{j_3}, \alpha_{j_4} \in M(\alpha_{j_1}) \setminus \{ \alpha_{j_1} \}$ and $u_{\alpha_{j_3}}(M)=2$. This can be done in $O(|N|)$ time, as follows. Consider each triple in $M$ until $M(\alpha_{j_1})$ is found. This takes $O(|N|)$ time. Use $L_1$ to identify $\alpha_{j_3}$ and $\alpha_{j_4}$.
 
 The initialisation of $S, c, \alpha_{z_1}, \alpha_{z_2}, \alpha_{y_1}, \alpha_{y_2}$ and $\alpha_{w_1}$ in the algorithm takes constant time.
 
 Consider the \texttt{while} loop. By Lemma~\ref{lem:algalwaysterminates}, there are at most $\lfloor (|N|-2) \mathbin{/} 3 \rfloor = O(|N|)$ iterations. Setting up the lookup tables allows us to ensure that each iteration takes $O(|N|)$ time. It follows that the loop terminates in $O(|N|^2)$ time.
 
 To identify $\alpha_{z_1}$ as described, first identify $S_{3c-1}$, in constant time. Consider each $\alpha_{p}\in N$ for which $\mathit{val}_{S_{3c-1}}(\alpha_{p})=1$. This takes $O(|N|)$ time. For each such $\alpha_{p}$, if $\alpha_{p}=\alpha_i$ then continue. If $\alpha_{p}\neq \alpha_i$ then look up $u_{\alpha_{p}}(M)$ in $L_1$. If $u_{\alpha_{p}}(M)=0$ then set $\alpha_{z_1}=\alpha_{p}$. If no such $\alpha_{p}$ is found then no such $\alpha_{z_1}$ exists so set $\alpha_{z_1}=\bot$.
 
 Similarly, to identify some $\alpha_{z_2}$ as described, first identify $S_{3c}$. Consider each $\alpha_{l_1}\in N$ for which $\mathit{val}_{S_{3c}}(\alpha_{l_1})=1$. This takes $O(|N|)$ time. For each such $\alpha_{l_1}$, if $\alpha_{l_1}=\alpha_i$ or $\alpha_{l_1}=\alpha_{j_2}$ then continue. If not, look up $u_{\alpha_{l_1}}(M)$ in $L_1$. If $u_{\alpha_{l_1}}(M)=0$ then set $\alpha_{z_2}=\alpha_{l_1}$. If no such $\alpha_{p}$ is found then no such $\alpha_{z_2}$ exists so set $\alpha_{z_2}=\bot$.
 
 To identify $\alpha_{y_1}$ as described, test if $\mathit{val}_{S_{3c}}(\alpha_{i})=1$. This takes $O(|N|)$ time. If $\mathit{val}_{S_{3c}}(\alpha_{i})=0$ then no such $\alpha_{y_1}$ exists. If $\mathit{val}_{S_{3c}}(\alpha_{i})=1$ then consider each $\alpha_{p}\in N$ for which $\mathit{val}_{\alpha_{i}}(\alpha_{p})=1$. Note that $\alpha_p\neq \alpha_{j_2}$ since otherwise $\mathit{val}_{\alpha_{j_2}}(\alpha_{i})=1$, from which it follows that $\{ \alpha_i, \alpha_{j_1}, \alpha_{j_2} \}$ is a triangle in $(N, V)$. Look up $u_{\alpha_{p}}(M)$ in $L_1$. If $u_{\alpha_{p}}(M)=0$ then set $\alpha_{y_1}=\alpha_{p}$. If no such $\alpha_{p}$ where $u_{\alpha_{l_1}}(M)=0$ is found then no such $\alpha_{y_1}$ exists so set $\alpha_{y_1}=\bot$. The identification of $\alpha_{y_2}$, if it exists, can be performed similarly in $O(|N|)$ time.
 
 To compute $1 \leq b < c$ as described, if there exists some such $S_{3b}$ where $\mathit{val}_{S_{3b}}(\alpha_{j_2})=\mathit{val}_{S_{3c}}(S_{3b})=1$, consider each $\alpha_{p}\in N$ for which $\mathit{val}_{S_{3c}}(\alpha_{p})=1$. This takes $O(|N|)$ time. For each such $\alpha_{p}$, determine its position $b'$ in $S$ if it belongs to $S$. If $\alpha_{p}$ belongs to $S$ and $b'$ is divisible by three and less than $c$ then set $b=b'$. Otherwise, no such $S_{3b}$ exists so set $b=0$.
 
To identify some $\alpha_{w_1}$ as described, first identify $S_{3c}$ in constant time. Consider each $\alpha_{p}\in N$ for which $\mathit{val}_{S_{3c}}(\alpha_{p})=1$. This takes $O(|N|)$ time. For each such $\alpha_{p}$, test if $\alpha_{p}$ belongs to $S$ using $L_S$. If so, then continue. If not, then look up the $p\textsuperscript{th}$ entry in $L_2$. If this entry is $\bot$ then continue. If not, then suppose this entry is $\alpha_{q}$. By the construction of $L_2$, it follows that $\alpha_{q}\in N \setminus \{ \alpha_i \}$, $\mathit{val}_{\alpha_{p}}(\alpha_{q})=1$ and $u_{\alpha_{q}}(M)=0$. Accordingly, set $\alpha_{w_1}$ to $\alpha_{p}$ since the algorithm has identified $\alpha_{z_3}=\alpha_{q}\in N \setminus \{ \alpha_i \}$ for which $\mathit{val}_{\alpha_{w_1}}(\alpha_{z_3})=1$ and $u_{\alpha_{z_3}}(M)=0$.

Evaluating the \texttt{break} condition in the loop can be performed in constant time. If the break condition is true then $\alpha_{w_1}$ exists. The identification of $\alpha_{w_2}$ and $\alpha_{w_3}$ can be accomplished in $O(|N|)$ time, using the same process as for $\alpha_{j_3}$ and $\alpha_{j_4}$. From above, adding three elements to $S$ requires constant time.

Now consider the final \texttt{if/else} statement and the seven possible constructions of $M'$. In each of the seven cases, $M'$ contains each triple in $\{ r \in M | r \cap S = \varnothing \}$. This set can be constructed in $O(|N|)$ time by considering each triple in $M$ and the three corresponding entries in $L_S$. In Cases 3 and 6, the agents $\alpha_{z_4}$ and $\alpha_{z_5}$ can each be identified in $O(|N|)$ time, using a similar process as for $\alpha_{z_1}$ in the loop body as described above. The remaining triples in $M'$ can be constructed after one scan of $S$ in $O(|N|)$ time.
\end{proof}

\begin{restatable}{lem}{algoreturnsstablematching}
\label{lem:algoreturnsstablematching}
Algorithm~{\normalfont \texttt{repair}} returns a stable $P$\nobreakdash-matching in $O(|N|^2)$ time.
\end{restatable}
\begin{proof}
By Lemmas~\ref{lem:algoreturnsstablematching_notimecomplex} and~\ref{lem:3dsrsasbin_almosttherealgo_runningtime}.
\end{proof}

\else

\begin{restatable}{lem}{algoreturnsstablematching}
\label{lem:algoreturnsstablematching}
Algorithm~{\normalfont \texttt{repair}} returns a stable $P$\nobreakdash-matching in $O(|N|^2)$ time.
\end{restatable}

\fi

\subsection{Finding a stable \texorpdfstring{$P$}{P}-matching in a triangle-free instance}
\label{sec:3dsrsasbin_trianglefree_section}
In the previous section we supposed that $(N, V)$ was a triangle-free instance of 3D-SR-SAS-BIN and considered a $P$\nobreakdash-matching $M$ that was repairable (Section~\ref{sec:3dsrsasbin_specialcase_algorithmsection}). We presented Algorithm~\texttt{repair}, which can be used to construct a stable $P$\nobreakdash-matching $M'$ in $O(|N|^2)$ time (Lemma~\ref{lem:algoreturnsstablematching}). In this section we present Algorithm~\texttt{findStableInTriangleFree} (Algorithm~\ref{alg:3dsrsasbin_find_stable_pmatching_in_triangle_free_instance}), which, given a triangle-free instance $(N, V)$, constructs a $P$\nobreakdash-matching $M'$ that is stable in $(N, V)$. Algorithm~\texttt{findStableInTriangleFree} is recursive. The algorithm first removes an arbitrary agent $\alpha_i$ to construct a smaller instance $(N', V')$. It then uses a rec ursive call to construct a $P$\nobreakdash-matching $M$ that is stable in $(N', V')$. By Lemma~\ref{lem:3dsrsasbinblockerimprovement}, any triple that blocks $M$ in the larger instance $(N, V)$ must contain $\alpha_i$ or block $M$ in $(N', V')$. There are then three cases involving types of triple that block $M$ in $(N', V')$. In two out of three cases, $M'$ can be constructed by adding to $M$ a new triple containing $\alpha_i$ and two players unmatched in $M$. In the third case, $M$ is not stable in $(N, V)$ but, by design, is repairable (see Section~\ref{sec:3dsrsasbin_specialcase_algorithmsection}). It follows that Algorithm~\texttt{repair} can be used to construct a $P$\nobreakdash-matching that is stable in $(N, V)$ (Lemma~\ref{lem:algoreturnsstablematching}). It is relatively straightforward to show that the running time of Algorithm~\texttt{findStableInTriangleFree} is $O(|N|^3)$.

\newcommand\contentsOfFindStableInTriangleFreeAlgorithm{
\textbf{Input:} an instance $(N,V)$ of 3D-SR-SAS-BIN\\
\textbf{Output:} stable $P$-matching $M'$ in $(N,V)$
\smallskip 
\begin{algorithmic}
\caption{Algorithm \texttt{findStableInTriangleFree} \label{alg:3dsrsasbin_find_stable_pmatching_in_triangle_free_instance}} 

\If{$|N|=2$} \Return $\varnothing$
\EndIf
\smallskip

\State $\alpha_i \gets$ an arbitrary agent in $N$
\State $(N', V') \gets (N \setminus \{ \alpha_i \}, V \setminus \{ \mathit{val}_{\alpha_i} \} )$
\State $M \gets \texttt{findStableInTriangleFree}((N', V'))$

\smallskip

\If{some $\alpha_{l_1}, \alpha_{l_2}\in N$ exist where $u_{\alpha_{l_1}}(M)=u_{\alpha_{l_2}}(M)=0$
\State and $\mathit{val}_{\alpha_i}(\alpha_{l_1})=\mathit{val}_{\alpha_i}(\alpha_{l_2})=1$}
\smallskip

    \State \Return $M \cup \{ \{ \alpha_{i}, \alpha_{l_1}, \alpha_{l_2} \} \}$
    
\ElsIf{some $\alpha_{l_3}, \alpha_{l_4}\in N$ exist where $u_{\alpha_{l_3}}(M)=u_{\alpha_{l_4}}(M)=0$
\State and $\mathit{val}_{\alpha_i}(\alpha_{l_3})=\mathit{val}_{\alpha_{l_3}}(\alpha_{l_4})=1$}
    
    \smallskip

    \State \Return $M \cup \{ \{ \alpha_{i}, \alpha_{l_3}, \alpha_{l_4} \} \}$
    
\ElsIf{some $\alpha_{l_5}, \alpha_{l_6}\in N$ exist where $u_{\alpha_{l_5}}(M)=1$, $u_{\alpha_{l_6}}(M)=0$ 
\State and $\mathit{val}_{\alpha_i}(\alpha_{l_5})=\mathit{val}_{\alpha_{l_5}}(\alpha_{l_6})=1$}

    \smallskip
    \LineComment{$M$ is repairable in $(N, V)$ (see Section~\ref{sec:3dsrsasbin_specialcase_algorithmsection}). Note that $\alpha_{j_1}=\alpha_{l_5}$ and $\alpha_{j_2}=\alpha_{l_6}$.}
    \State \Return $\texttt{repair}((N, V), M, \alpha_i)$
    
\Else

    \State \Return $M$

\EndIf
\State \textbf{end if}


\medskip
\end{algorithmic}
}

\ifdefined \fullversion
\begin{algorithm}[b!]
\contentsOfFindStableInTriangleFreeAlgorithm
\end{algorithm}
\else
\begin{algorithm}[t!]
\contentsOfFindStableInTriangleFreeAlgorithm
\end{algorithm}
\fi

\ifdefined \fullversion

\begin{lem}
\label{lem:3dsrsasbin_algfindsstablepmatching_notimecomplex}
Given a triangle-free instance $(N, V)$, Algorithm~{\normalfont \texttt{findStableInTriangleFree}} returns a stable $P$\nobreakdash-matching in $(N, V)$.
\end{lem}
\begin{proof}
The proof is by induction on $|N|$. When $|N|=2$, the returned matching $\varnothing$ is trivially stable in $(N, V)$. Suppose then that Algorithm~{\normalfont \texttt{findStableInTriangleFree}} returns a stable $P$\nobreakdash-matching $M$ given $(N', V')$ where $|N'|<|N|$. It follows that the recursive call to Algorithm~{\normalfont \texttt{findStableInTriangleFree}} returns a $P$\nobreakdash-matching $M$ that is stable in $(N', V')$.

Consider the first branch of the \texttt{if/else} statement. By construction, $u_{\alpha_{i}}(M')\allowbreak =2$ and $u_{\alpha_{l_1}}(M')=u_{\alpha_{l_2}}(M')=1$. Since $M$ is a $P$\nobreakdash-matching, it follows that $M'$ is also a $P$\nobreakdash-matching. Suppose for a contradiction that some triple blocks the returned $P$\nobreakdash-matching $M'$ in $(N, V)$. Since $u_{\alpha_i}(M')=2$, such a triple does not contain $\alpha_i$. By construction, $u_{\alpha_p}(M')\geq u_{\alpha_p}(M)$ for any $\alpha_p\in N$, so it follows that such a triple also blocks $M$ in $(N', V')$, a contradiction.

Consider the second branch of the \texttt{if/else} statement. By construction, $u_{\alpha_{l_3}}(M')=2$ and $u_{\alpha_{i}}(M')=u_{\alpha_{l_3}}(M')=1$. Since $M$ is a $P$\nobreakdash-matching, it follows that $M'$ is also a $P$\nobreakdash-matching. Suppose for a contradiction that some triple blocks $M'$ in $(N, V)$. By construction, $u_{\alpha_p}(M')\geq u_{\alpha_p}(M)$ for any $\alpha_p\in N$. It follows that any such triple that blocks $M'$ in $(N, V)$ contains $\alpha_i$, otherwise that triple blocks $M$ in $(N', V')$, a contradiction. Suppose that some triple $\{ \alpha_i, \alpha_{k_1}, \alpha_{k_2} \}$ blocks $M'$ in $(N, V)$. By construction, $u_{\alpha_i}(M')=1$ so it must be that $u_{\alpha_i}(\{ \alpha_{k_1}, \alpha_{k_2} \})=2$ and hence $\mathit{val}_{\alpha_i}(\alpha_{k_1})=\mathit{val}_{\alpha_i}(\alpha_{k_2})=1$. Since $(N, V)$ is triangle-free, it follows that $u_{\alpha_{k_1}}(\{ \alpha_{i}, \alpha_{k_2} \})=u_{\alpha_{k_2}}(\{ \alpha_{i}, \alpha_{k_1} \})=1$. It follows that $u_{\alpha_{k_1}}(M')=u_{\alpha_{k_2}}(M')=0$. Since $u_{\alpha_p}(M')\geq u_{\alpha_p}(M)$ for any $\alpha_p\in N$, it must be that $u_{\alpha_{k_1}}(M)=u_{\alpha_{k_2}}(M)=0$. This contradicts the condition of the first branch of the \texttt{if/else} statement, since two agents $\alpha_{l_1}, \alpha_{l_2}$, namely $\alpha_{k_1}, \alpha_{k_2}$, exist where $u_{\alpha_{l_1}}(M)=u_{\alpha_{l_2}}(M)=0$ and $\mathit{val}_{\alpha_i}(\alpha_{l_1})=\mathit{val}_{\alpha_i}(\alpha_{l_2})=1$.

Consider the third branch of the \texttt{if/else} statement. It must be that the conditional expressions in the first and second branches of the \texttt{if/else} statement do not hold. It follows from this that every triple that blocks $M$ in $(N', V')$ comprises $\{ \alpha_i, \alpha_{l_5}, \alpha_{l_6}\}$ where $\alpha_{l_5}, \alpha_{l_6}\in N$, $u_{\alpha_{l_5}}(M)=1$, $u_{\alpha_{l_6}}(M)=0$, and $\mathit{val}_{\alpha_i}(\alpha_{l_5})=\mathit{val}_{\alpha_{l_5}}(\alpha_{l_6})=1$. Note that $u_{\alpha_i}(M)=0$ and hence this is exactly the condition required by Algorithm~{\normalfont \texttt{repair}} (see Section~\ref{sec:3dsrsasbin_specialcase_algorithmsection}). By Lemma~\ref{lem:algoreturnsstablematching}, Algorithm~{\normalfont \texttt{repair}} returns a $P$\nobreakdash-matching $M'$ that is stable in $(N, V)$.

Consider the fourth branch of the \texttt{if/else} statement. It must be that the conditional expressions in the first, second, and third branches of the \texttt{if/else} statement do not hold. Suppose for a contradiction that some triple blocks $M'=M$ in $(N, V)$. By construction, $u_{\alpha_p}(M')=u_{\alpha_p}(M)$ for any $\alpha_p\in N$. It follows that any such triple that blocks $M'$ in $(N, V)$ contains $\alpha_i$, otherwise that triple blocks $M$ in $(N', V')$, a contradiction. Suppose that some triple $\{ \alpha_i, \alpha_{k_1}, \alpha_{k_2} \}$ blocks $M'$ in $(N, V)$.

Suppose first that $u_{\alpha_i}(\{ \alpha_{k_1}, \alpha_{k_2} \})=2$. Since $(N, V)$ is triangle-free, it follows that $u_{\alpha_{k_1}}(\{ \alpha_{i}, \alpha_{k_2} \})=u_{\alpha_{k_2}}(\{ \alpha_{i}, \alpha_{k_1} \})=1$. It follows that $u_{\alpha_{k_1}}(M')=u_{\alpha_{k_2}}(M')=0$. Since $u_{\alpha_p}(M')\geq u_{\alpha_p}(M)$ for any $\alpha_p\in N$, it must be that $u_{\alpha_{k_1}}(M)=u_{\alpha_{k_2}}(M)=0$. This contradicts the condition of the first branch of the \texttt{if/else} statement.

Suppose then that $u_{\alpha_i}(\{ \alpha_{k_1}, \alpha_{k_2} \})=1$. It must be that either $u_{\alpha_{k_1}}(\{ \alpha_{i}, \alpha_{k_2} \})=2$ or $u_{\alpha_{k_2}}(\{ \alpha_{i}, \alpha_{k_1} \})=2$. Suppose without loss of generality that $u_{\alpha_{k_1}}(\{ \alpha_{i}, \alpha_{k_2} \})=2$. It follows that $\mathit{val}_{\alpha_{k_1}}(\alpha_{i})=\mathit{val}_{\alpha_{k_1}}(\alpha_{k_2})=1$. There are two possibilities: either $u_{\alpha_{k_1}}(M)=1$ or $u_{\alpha_{k_1}}(M)=0$. The first possibility implies that the conditional expression of the second \texttt{if/else} branch holds, a contradiction. The second possibility implies that the conditional expression of the third \texttt{if/else} branch holds, also a contradiction. It follows that no such triple $\{ \alpha_i, \alpha_{k_1}, \alpha_{k_2} \}$ blocks $M'$ in $(N, V)$.\end{proof}

Algorithm~{\normalfont \texttt{findStableInTriangleFree}} is recursive. We consider its asymptotic time complexity and prove that it has running time $O(|N|^3)$.

\begin{lem}
\label{lem:3dsrsasbin_algfindstablepmatchingrunningtime}
Algorithm~{\normalfont \texttt{findStableInTriangleFree}} has running time $O(|N|^3)$.
\end{lem}
\begin{proof}
The pseudocode for Algorithm~{\normalfont \texttt{findStableInTriangleFree}} gives an outline of the algorithm at a high level. As before, to analyse the worst-case time complexity we provide a more detailed description of certain steps in terms of the unit cost of operations in standard data structures. This description suffices to show that the running time of the algorithm is $O(|N|^3)$. Let $T(|N|)$ be the running time of the algorithm given an instance $(N, V)$. We will prove inductively that $T(|N|)=O(|N|^3)$.

Suppose that the input $(N, V)$ is given such that, for a given $\alpha_{p}\in N$, the algorithm can iterate through the set $\{ \alpha_{q} \in N : \mathit{val}_{\alpha_{p}}(\alpha_{q})=1 \}$ in $O(|N|)$ time. For example, $(N, V)$ could be stored graphically using adjacency lists. It follows that, given three agents $\alpha_{h_1}, \alpha_{h_2}, \alpha_{h_3}\in N$ the algorithm can compute $u_{\alpha_{h_1}}(\{ \alpha_{h_2}, \alpha_{h_3} \}), u_{\alpha_{h_2}}(\{ \alpha_{h_1}, \alpha_{h_3} \}),$ and $u_{\alpha_{h_3}}(\{ \alpha_{h_1}, \alpha_{h_2} \})$ in $O(|N|)$ time. In any case, algorithm will return a $P$\nobreakdash-matching $M'$ stored as a linked list or similar data structure that allows a triple to be appended to the end of list in constant time.

By inspection, when $|N|=2$ the algorithm returns immediately and hence $T(2)=O(1)$. In this case the algorithm will return an empty linked list or similar data structure.

The constructed instance $(N',V')$ can be stored using adjacency lists or an equivalent data structure. A straightforward procedure to select $\alpha_i$ and construct $(N', V')$ takes $O(|N|)$ time. By assumption, the recursive call to construct $M'$ takes $T(|N|-1)$ time.

After this call, the algorithm constructs a supplementary lookup table $L_1$, with exactly $|N|-1$ entries indexed by each $\alpha_p\in N'$. Each entry will contain either zero, one, or two. For each agent $\alpha_p\in N$, the algorithm constructs $L_1$ so that the $p\textsuperscript{th}$ entry contains $u_{\alpha_p}(M)$. By assumption, the algorithm can compute $u_{\alpha_p}(M)$ for any $\alpha_p\in N$ in $O(|N|)$ time. It follows that $L_1$ can be constructed in $O(|N|^2)$ time by iterating through $M$ and computing $u_{\alpha_{h_1}}(M), u_{\alpha_{h_2}}(M), u_{\alpha_{h_3}}(M)$ for each $\{ \alpha_{h_1}, \alpha_{h_2}, \alpha_{h_3} \} \in M$. Since $|M|=O(|N|)$ in total this step takes $O(|N|^2)$ time. It follows that we can use $L_1$ to look up $u_{\alpha_{p}}(M)$ for any $\alpha_{p}\in N$ in constant time.

The construction of $L_1$ allows the algorithm to identify some $\alpha_{l_1}, \alpha_{l_2}\in N$ exist where $u_{\alpha_{l_1}}(M)=u_{\alpha_{l_2}}(M)=0$ and $\mathit{val}_{\alpha_i}(\alpha_{l_1})=\mathit{val}_{\alpha_i}(\alpha_{l_2})=1$, if two such agents exist, in $O(|N|^2)$ time. One way to do this is to consider each pair $(\alpha_{l_1}, \alpha_{l_2})\in N^2$ and look up $u_{\alpha_{l_1}}(M)$ and $u_{\alpha_{l_2}}(M)$ in $L_1$. Since $M$ is stored using a linked list or similar data structure, if such $\alpha_{l_1}, \alpha_{l_2}\in N$ exist then $M'$ can be constructed by adding the triple $\{ \alpha_i, \alpha_{l_1}, \alpha_{l_2} \}$ to $M$ in constant time. Similarly, the identification of $\alpha_{l_3}, \alpha_{l_4}\in N$ where $u_{\alpha_{l_3}}(M)=u_{\alpha_{l_4}}(M)=0$ and $\mathit{val}_{\alpha_i}(\alpha_{l_3})=\mathit{val}_{\alpha_{l_3}}(\alpha_{l_4})=1$ can be performed in $O(|N|^2)$ time and the corresponding construction of $M'$ in constant time. In the third branch of the \texttt{if/else} statement, the identification of $\alpha_{l_5}, \alpha_{l_6}\in N$ where $u_{\alpha_{l_3}}(M)=1$, $u_{\alpha_{l_4}}(M)=0$ and $\mathit{val}_{\alpha_i}(\alpha_{l_3})=\mathit{val}_{\alpha_{l_3}}(\alpha_{l_4})=1$ can be similarly performed in $O(|N|^2)$ time. By Lemma~\ref{lem:3dsrsasbin_almosttherealgo_runningtime}, the call to Algorithm~{\normalfont \texttt{repair}} also takes $O(|N|^2)$ time. It follows that the overall running time of Algorithm~{\normalfont \texttt{findStableInTriangleFree}} is $O(|N|^3)$.
\end{proof}

\begin{restatable}{lem}{threedsrsasbinalgfindsstablepmatching}
\label{lem:threedsrsasbin_algfindsstablepmatching}
Algorithm~{\normalfont \texttt{findStableInTriangleFree}} returns a stable $P$\nobreakdash-matching in $(N, V)$ in $O(|N|^3)$ time.
\end{restatable}
\begin{proof}
By Lemmas~\ref{lem:3dsrsasbin_algfindsstablepmatching_notimecomplex} and~\ref{lem:3dsrsasbin_algfindstablepmatchingrunningtime}.
\end{proof}

\else

\begin{restatable}{lem}{threedsrsasbinalgfindsstablepmatching}
\label{lem:threedsrsasbin_algfindsstablepmatching}
Algorithm~{\normalfont \texttt{findStableInTriangleFree}} returns a stable $P$\nobreakdash-matching in $(N, V)$ in $O(|N|^3)$ time.
\end{restatable}


\fi
\subsection{Finding a stable \texorpdfstring{$P$}{P}-matching in an arbitrary instance}
\label{sec:3dsrsasbin_findingarbitrarystablematching}

\ifdefined \fullversion

In the previous section we considered instances of 3D-SR-SAS-BIN that are triangle-free. We showed that, given such an instance, Algorithm~\texttt{findStableInTriangleFree} can be used to find a stable $P$\nobreakdash-matching in $O(|N|^3)$ time (Lemma~\ref{lem:threedsrsasbin_algfindsstablepmatching}). In Section~\ref{sec:3dsrsasbin_prelims}, we showed that an arbitrary instance can be reduced in $O(|N|^3)$ time to construct a corresponding triangle-free instance (Lemma~\ref{lem:threedsrsasbintrianglefree}). We define a subroutine, \texttt{eliminateTriangles}, which reduces an arbitrary instance in this way, and returns a pair containing the reduced instance and a set of triangles $M_{\triangle}$. Algorithm~\texttt{findStable} therefore comprises two steps. First, the instance is reduced with a call to \texttt{eliminateTriangles}. Then, Algorithm~\texttt{findStableInTriangleFree} is called to construct a $P$\nobreakdash-matching $M'$ that is stable in the reduced, triangle-free instance.

\begin{algorithm}[H]
\textbf{Input:} an instance $(N,V)$ of 3D-SR-SAS-BIN\\
\textbf{Output:} stable $P$-matching $M'$ in $(N,V)$
\smallskip 
\begin{algorithmic}
\caption{Algorithm~\texttt{findStable} \label{alg:3dsrsasbin_overall_algo}} 

\State $(N', V'), M_{\triangle} \gets \texttt{eliminateTriangles}((N, V))$
\State $M' \gets \texttt{findStableInTriangleFree}((N', V'))$

\smallskip

\State \Return $M' \cup M_{\triangle}$

\medskip
\end{algorithmic}
\end{algorithm}

\begin{lem}
\label{lem:3dsrsasbinconstruction_norunningtime}
Given an instance $(N, V)$ of 3D-SR-SAS-BIN, Algorithm~{\normalfont \texttt{findStable}} returns a stable $P$\nobreakdash-matching.
\end{lem}
\begin{proof} 
A suitable implementation of the subroutine {\normalfont \texttt{eliminateTriangles}} returns a pair $((N', V'), M_{\triangle})$ where $|N'|\leq |N|$ and if $M$ is a stable $P$\nobreakdash-matching in $(N', V')$ then $M' = M \cup M_{\triangle}$ is a stable $P$\nobreakdash-matching in $(N, V)$ (Lemma~\ref{lem:threedsrsasbintrianglefree}). By Lemma~\ref{lem:threedsrsasbin_algfindsstablepmatching}, Algorithm~{\normalfont \texttt{findStableInTriangleFree}} returns $P$\nobreakdash-matching $M'$ that is stable in in $(N', V')$. It follows that $M' \cup M_{\triangle}$ is a $P$\nobreakdash-matching that is stable in $(N, V)$.
\end{proof}

\begin{lem}
\label{lem:3dsrsasbinconstruction_runningtime}
Algorithm~{\normalfont \texttt{findStable}} has running time $O(|N|^3)$.
\end{lem}
\begin{proof}
By definition, Algorithm~{\normalfont \texttt{eliminateTriangles}} has running time $O(|N|^3)$ (Lemma~\ref{lem:threedsrsasbintrianglefree}). By Lemma~\ref{lem:3dsrsasbin_algfindstablepmatchingrunningtime}, Algorithm~{\normalfont \texttt{findStableInTriangleFree}} also has running time $O(|N|^3)$. It follows that Algorithm~{\normalfont \texttt{findStable}} has total running time $O(|N|^3)$.
\end{proof}

\begin{restatable}{theorem}{threedsrsasbinconstruction}
\label{lem:3dsrsasbinconstruction}
Given an instance $(N, V)$ of 3D-SR-SAS-BIN, a stable $P$\nobreakdash-matching, and hence a stable matching, must exist and can be found in $O(|N|^3)$ time. Moreover, if $|N|$ is a multiple of three then, if required, every agent can be matched in the returned stable matching.
\end{restatable}
\begin{proof}
By Lemmas~\ref{lem:3dsrsasbinconstruction_norunningtime} and~\ref{lem:3dsrsasbinconstruction_runningtime}. If $|N|$ is a multiple of three, then if required the agents unmatched in $M' \cup M_{\triangle}$ can be arbitrarily matched into triples. By Lemma~\ref{lem:3dsrsasbinblockerimprovement}, the resulting matching is still stable in $(N, V)$. 
\end{proof}

\else

In the previous section we considered instances of 3D-SR-SAS-BIN that are triangle-free. We showed that, given such an instance, Algorithm~\texttt{findStableInTriangleFree} can be used to find a stable $P$\nobreakdash-matching in $O(|N|^3)$ time (Lemma~\ref{lem:threedsrsasbin_algfindsstablepmatching}). In Section~\ref{sec:3dsrsasbin_prelims}, we showed that an arbitrary instance can be reduced in $O(|N|^3)$ time to construct a corresponding triangle-free instance (Lemma~\ref{lem:threedsrsasbintrianglefree}). Algorithm \texttt{findStable} therefore comprises two steps. First, the instance is reduced by removing a maximal set of triangles. Call this set $M_{\triangle}$. Then, Algorithm~\texttt{findStableInTriangleFree} is called to construct a $P$\nobreakdash-matching $M'$ that is stable in the reduced, triangle-free instance. It is straightforward to show that $M_{\triangle} \cup M'$ is a stable $P$\nobreakdash-matching. The running time of Algorithm~\texttt{findStable} is thus $O(|N|^3)$. A pseudocode description of Algorithm~\texttt{findStable} can be found in the full version of this paper \cite{fullversion3dsraspaper}. We arrive at the following result.

\begin{restatable}{theorem}{threedsrsasbinconstruction}
\label{lem:3dsrsasbinconstruction}
Given an instance $(N, V)$ of 3D-SR-SAS-BIN, a stable $P$\nobreakdash-matching, and hence a stable matching, must exist and can be found in $O(|N|^3)$ time. Moreover, if $|N|$ is a multiple of three then, if required, every agent can be matched in the returned stable matching.
\end{restatable}

\fi
\subsection{Stability and utilitarian welfare}
\label{sec:3dsrsasbin_utilitarianwelfare}

Given an instance $(N,V)$ of 3D-SR-SAS-BIN and matching $M$, let the \emph{utilitarian welfare} \cite{10.5555/2832249.2832313,10.5555/3398761.3398791} of a set $S\subseteq N$, denoted $u_{S}(M)$, be $\sum\limits_{\alpha_i\in S} u_{\alpha_i}(M)$. Let $u(M)$ be short for $u_{N}(M)$. Given a matching $M$ in an arbitrary instance $(N, V)$ of 3D-SR-SAS-BIN, it follows that $0 \leq u(M) \leq 2|N|$. It is natural to then consider the optimisation problem of finding a stable matching with maximum utilitarian welfare, which we refer to as \emph{3D-SR-SAS-BIN-MAXUW}. This problem is closely related to Partition Into Triangles (PIT, defined in Section~\ref{sec:generalbinary}), which we reduce from in the proof that 3D-SR-SAS-BIN-MAXUW is $\NP$-hard.

\ifdefined \fullversion

\begin{restatable}{theorem}{threedsrsasbinmaxutilstablehard}
\label{thm:3dsrsasbin_maxutilstable_hard}
3D-SR-SAS-BIN-MAXUW is $\NP$-hard.
\end{restatable}
\begin{proof}
A trivial reduction exists from Partition Into Triangles (defined in Section~\ref{sec:generalbinary}) to the problem of deciding if a given instance of 3D-SR-SAS-BIN-MAXUW contains a stable matching $M$ with $u(M)=2|N|$.
\end{proof}

\else

\begin{restatable}{theorem}{threedsrsasbinmaxutilstablehard}
\label{thm:3dsrsasbin_maxutilstable_hard}
3D-SR-SAS-BIN-MAXUW is $\NP$-hard.
\end{restatable}

\fi

We note that the reduction from PIT to 3D-SR-SAS-BIN-MAXUW also shows that the problem of finding a (not-necessarily stable) matching with maximum utilitarian welfare, given an instance of 3D-SR-SAS-BIN, is also $\NP$-hard.

In Section~\ref{sec:3dsrsasbin_findingarbitrarystablematching} we showed that, given an arbitrary instance $(N, V)$ of 3D-SR-SAS-BIN, a stable $P$\nobreakdash-matching exists and can be found in $O(|N|^3)$ time. We now present Algorithm~\texttt{findStableUW} (Algorithm~\ref{alg:3dsrsasbin_approximationalgo}) as an approximation algorithm for 3D-SR-SAS-BIN-MAXUW. This algorithm first calls Algorithm~\texttt{findStable} to construct a stable $P$\nobreakdash-matching. It then orders the unmatched agents into triples such that the produced matching is still stable in $(N,V)$ (by Lemma~\ref{lem:3dsrsasbinblockerimprovement}) but is not necessarily a $P$\nobreakdash-matching. The pseudocode description of Algorithm~\texttt{findStableUW} includes a call to \texttt{maximum2DMatching}. Given an instance $(N, V)$ and some set $U\subseteq N$, this subroutine returns a (two-dimensional) \emph{maximum cardinality matching} $Y$ in the subgraph of $G$, the underlying graph of $(N,V)$, induced by $U$. From $Y$, Algorithm~\texttt{findStableUW} constructs a set $X$ of pairs with cardinality $\lfloor |U|/3 \rfloor$. It also constructs a set $Z$ from the remaining agents, also with cardinality $\lfloor |U|/3 \rfloor$. Finally, it constructs the matching $M_2$ such that each triple in $M_2$ is union of a pair of agents in $X$ and a single agent in $Z$.
\newcommand\contentsOfFindStableUWAlgorithm{
\textbf{Input:} an instance $(N,V)$ of 3D-SR-SAS-BIN\\
\textbf{Output:} stable matching $M_A$ in $(N,V)$
\smallskip 
\begin{algorithmic}
\caption{Algorithm \texttt{findStableUW}\label{alg:3dsrsasbin_approximationalgo}} 
\State $M_1 \gets \texttt{findStable}((N, V))$
\State $U \gets \text{agents in $N$ unmatched in $M_1$}$
\State $Y \gets \texttt{maximum2DMatching}((N, V),\, U)$
\smallskip

\If {$|Y| \geq \lfloor |U|/3 \rfloor $}

\State $X \gets \text{any $\lfloor |U|/3 \rfloor$ elements of $Y$}$

\Else

\LineComment{Note that since $Y$ is a set of disjoint pairs, it follows that}
\LineCommentPhantom{$|U \setminus \bigcup Y|=|U|-2|Y|\geq \lfloor |U|/3 \rfloor - |Y|$.}
\smallskip

\State $W \gets \text{ an arbitrary set of } \lfloor |U|/3 \rfloor - |Y| \text{ pairs of elements in $U \setminus \bigcup Y$}$
\State $X \gets Y \cup W$

\EndIf
\State \textbf{end if}
\smallskip

\State $Z \gets U \setminus \bigcup X$

\smallskip

\LineComment{Suppose $X=\{x_1, x_2, \dots, x_{\lfloor |U|/3 \rfloor}\}$ and $Z=\{z_1, z_2, \dots, z_{\lfloor |U|/3 \rfloor}\}$.}
\LineComment{Note that $x_i$ is a pair of agents and $z_i$ is a single agent for each $1\leq i \leq {\lfloor |U|/3 \rfloor}$.}
\smallskip

\State $M_2 \gets \{ x_i \cup \{ z_i \} \text{ for each } 1 \leq i \leq \lfloor |U|/3 \rfloor \}$

\State \Return $M_1 \cup M_2$

\medskip
\end{algorithmic}
}
\ifdefined \fullversion
\begin{algorithm}
\contentsOfFindStableUWAlgorithm
\end{algorithm}
\else
\begin{algorithm}[t!]
\contentsOfFindStableUWAlgorithm
\end{algorithm}
\fi

\ifdefined \fullversion

We consider Algorithm~{\normalfont \texttt{findStableUW}} with an arbitrary input instance $(N, V)$. The goal is to show that $2u(M_{\textrm{A}}) \geq u(M_\textrm{opt})$, where $M_{\textrm{A}}$ is the stable matching returned by the algorithm, and $M_{\textrm{opt}}$ is a stable matching in $(N, V)$ with maximum utilitarian welfare. Recall that $|N|=3k+l$ for some $k \geq 0$ and $0 \leq l < 3$ and by Proposition~\ref{prop:completematching} we assume that $|M_{\textrm{opt}}|=k$.

The proof is broken down into two cases. The first case is proved in Lemma~\ref{lem:3dsrsasbin_no000exists_lem}. In this case, the utilitarian welfare of any triple in $M_{\textrm{A}}$ is at least two. The second case, in which some triple in $M_\textrm{A}$ has utilitarian welfare zero, is considered in Lemmas~\ref{lem:3dsrsasbin_some000exists_r1_lem} -- and~\ref{lem:3dsrsasbin_some000exists_final_lem}. At a high level, the proof in both cases is similar, and involves placing a lower bound on the welfare in $M_{\textrm{A}}$ of the agents in each triple in $M_{\textrm{opt}}$.

Let $T_\textrm{opt}^y$ and $T_\textrm{A}^y$ be the set of triples each with utilitarian welfare $y$ in $M_{\textrm{opt}}$ and $M_{\textrm{A}}$ respectively. Recall that since the valuation functions are symmetric, $u_t(M)\in \{0, 2, 4, 6\}$ for any triple $t$ in an arbitrary matching $M$. It follows that
\begin{alignat}{7}
    M_{\textrm{opt}} &= \matheqbox{RARoptbox}{T_{\textrm{opt}}^6} &\,& \cup &\,& \matheqbox{RARoptbox}{T_{\textrm{opt}}^4} &\,& \cup &\,& \matheqbox{RARoptbox}{T_{\textrm{opt}}^2} &\,& \cup &\,& \matheqbox{RARoptbox}{T_{\textrm{opt}}^0} \label{eq:moptcomposition}\\
    M_{\textrm{A}} &= 
    \matheqbox{RARoptbox}{T_{\textrm{A}}^6} &&\cup&& \matheqbox{RARoptbox}{T_{\textrm{A}}^4} &&\cup&& \matheqbox{RARoptbox}{T_{\textrm{A}}^2} &&\cup&& \matheqbox{RARoptbox}{T_{\textrm{A}}^0} \label{eq:macomposition}
\end{alignat}
and hence
\begin{align}
u(M_\textrm{opt}) &= 6|T_{\textrm{opt}}^6| + 4|T_{\textrm{opt}}^4| + 2|T_{\textrm{opt}}^2| \label{eq:welfareofmopt} \\
u(M_\textrm{A}) &= 6|T_{\textrm{A}}^6| + 4|T_{\textrm{A}}^4| + 2|T_{\textrm{A}}^2|\enspace. \label{eq:welfareofma}
\end{align}

Lemma~\ref{lem:3dsrsasbin_ma_is_complete} shows that, by design, there are exactly $l$ agents in $N$ that are unmatched in $M_{\textrm{A}}$.

\begin{lem}
\label{lem:3dsrsasbin_ma_is_complete}
$|M_{\textrm{A}}|=k$.
\end{lem}
\begin{proof}
Recall that $|N|=3k+l$ for $k\geq 1$ and $0\leq l < 3$. Since $U$ contains the agents unmatched in $|M_1|$,
\begin{align}
    |U| &= |N| - 3|M_1| \nonumber\\
    &= 3k + l - 3|M_1| \label{eq:sizeofu}\enspace.
\end{align}
It then follows that
\begin{align}
    \lfloor |U| / 3 \rfloor &= \lfloor (3k + l - 3|M_1|)/3 \rfloor && \mbox{by Equation~\ref{eq:sizeofu}} \nonumber\\
    &= \lfloor k + l/3 - |M_1| \rfloor \nonumber\\
    &= k + \lfloor l/3 \rfloor - |M_1|\nonumber\\
    &= k - |M_1| && \mbox{since $l<3$ by definition.} \label{eq:sizeofuover3floor}
\end{align}
Now consider $|M_{\textrm{A}}|$. By construction, $M_\textrm{A} = M_1 \cup M_2$ so it follows that
\begin{align}
    |M_{\textrm{A}}| &= |M_1| + |M_2|\nonumber\\
    &= |M_1| + \lfloor |U| / 3 \rfloor && \mbox{by construction} \nonumber\\
    &= |M_1| + k - |M_1| && \mbox{by Equation~\ref{eq:sizeofuover3floor}}\nonumber\\
    &= k\enspace. \nonumber
\end{align}
\end{proof}

Lemma~\ref{lem:3dsrsasbin_tau_a_geq_tau_opt_over_3} demonstrates a relationship between $T_{\textrm{A}}^6$ and $T_{\textrm{opt}}^6$.

\begin{lem}
\label{lem:3dsrsasbin_tau_a_geq_tau_opt_over_3}
$|T_{\textrm{A}}^6| \geq |T_{\textrm{opt}}^6|/3$.
\end{lem}
\begin{proof}
Consider an arbitrary triple $\{ \alpha_{h_1}, \alpha_{h_2}, \alpha_{h_3} \} \in T_{\textrm{opt}}^6$. This triple is a triangle, meaning $\mathit{val}_{\alpha_{h_1}}(\alpha_{h_2})=\mathit{val}_{\alpha_{h_2}}(\alpha_{h_3})=\mathit{val}_{\alpha_{h_3}}(\alpha_{h_1})=1$. Recall that the first step of Algorithm~\texttt{findStable} involved selecting a maximal set of triangles. In the pseudocode description of Algorithm~\texttt{findStable}, we described this operation using Algorithm~\texttt{eliminateTriangles}, which we refer to here. Since  $\{ \alpha_{h_1}, \alpha_{h_2}, \alpha_{h_3} \}$ is a triangle in $(N, V)$, either Algorithm~\texttt{eliminateTriangles} selected this triple, and $\{ \alpha_{h_1}, \alpha_{h_2}, \alpha_{h_3} \} \in T_{\textrm{A}}^6$, or at least one of $\alpha_{h_1}, \alpha_{h_2}, \alpha_{h_3}$ was added to a different triple in $T_{\textrm{A}}^6$. In either case, any triple in $T_{\textrm{opt}}^6$ contains at least one agent that belongs to some triple in $T_{\textrm{A}}^6$. Triples in $T_{\textrm{A}}^6$ are disjoint, hence the number of agents in triples in $T_{\textrm{A}}^6$ is at least $|T_{\textrm{opt}}^6|$. It follows that $|T_{\textrm{A}}^6| \geq |T_{\textrm{opt}}^6|/3$.
\end{proof}

In Lemma~\ref{lem:3dsrsasbin_no000exists_lem} we consider the case when $T_{\textrm{A}}^0=\varnothing$.

\begin{lem}
\label{lem:3dsrsasbin_no000exists_lem}
If $T_{\textrm{A}}^0=\varnothing$ then $2{u(M_{\textrm{A}})} \geq u(M_\textrm{opt})$.
\end{lem}
\begin{proof}
We start by presenting an upper bound on $|T_{\textrm{opt}}^4| + |T_{\textrm{opt}}^2|$ in terms of $k$ and $|T_{\textrm{opt}}^6|$. Recall that $|M_{\textrm{opt}}| = k$ by Proposition~\ref{prop:completematching}.
\begin{align}
    |T_{\textrm{opt}}^6| + |T_{\textrm{opt}}^4| + |T_{\textrm{opt}}^2| + |T_{\textrm{opt}}^0| &= |M_{\textrm{opt}}| = k && \mbox{by Equation~\ref{eq:moptcomposition}} \nonumber\\
    |T_{\textrm{opt}}^6| + |T_{\textrm{opt}}^4| + |T_{\textrm{opt}}^2| &\leq k \nonumber\\
    |T_{\textrm{opt}}^4| + |T_{\textrm{opt}}^2| &\leq k - |T_{\textrm{opt}}^6| \label{eq:maxsizeofr4optr2opt}\enspace.
\end{align}

We now place an upper bound on $u(M_\textrm{opt})$ only in terms of $|T_{\textrm{opt}}^6|$ and $k$.
\begin{align}
u(M_{\textrm{opt}}) &= 6|T_{\textrm{opt}}^6| + 4|T_{\textrm{opt}}^4| + 2|T_{\textrm{opt}}^2| && \mbox{(Equation~\ref{eq:welfareofmopt})} \nonumber\\
&\leq 6|T_{\textrm{opt}}^6| + 4(|T_{\textrm{opt}}^4| + |T_{\textrm{opt}}^2|) \nonumber\\
&\leq 6|T_{\textrm{opt}}^6| + 4(k - |T_{\textrm{opt}}^6|) && \mbox{by Inequality~\ref{eq:maxsizeofr4optr2opt}} \nonumber\\
&\leq 6|T_{\textrm{opt}}^6| + 4k - 4|T_{\textrm{opt}}^6| \nonumber\\
&\leq 2|T_{\textrm{opt}}^6| + 4k \label{eq:lowerboundumopt}\enspace.
\end{align}

Considering $M_{\textrm{A}}$, the following equalities hold:
\begin{align}
    |T_{\textrm{A}}^6| + |T_{\textrm{A}}^4| + |T_{\textrm{A}}^2| + |T_{\textrm{A}}^0| &= |M_{\textrm{A}}|  && \mbox{by Equation~\ref{eq:macomposition}} \nonumber \\
    |T_{\textrm{A}}^6| + |T_{\textrm{A}}^4| + |T_{\textrm{A}}^2| &= |M_{\textrm{A}}| && \mbox{since $|T_{\textrm{A}}^0|=\varnothing$} \nonumber\\
    |T_{\textrm{A}}^6| + |T_{\textrm{A}}^4| + |T_{\textrm{A}}^2| &= k && \mbox{by Lemma~\ref{lem:3dsrsasbin_ma_is_complete}} \nonumber\\
    |T_{\textrm{A}}^4| + |T_{\textrm{A}}^2| &= k - |T_{\textrm{A}}^6|\enspace. \label{eq:sizeofr4ar2a}
\end{align}

Placing a lower bound on $u(M_{\textrm{A}})$,
\begin{align}
    u(M_{\textrm{A}}) &= 6|T_{\textrm{A}}^6| + 4|T_{\textrm{A}}^4| + 2|T_{\textrm{A}}^2| && \mbox{(Equation~\ref{eq:welfareofma})} \nonumber\\
    &\geq 6|T_{\textrm{A}}^6| + 2(|T_{\textrm{A}}^4| + |T_{\textrm{A}}^2|) \nonumber\\
    &\geq 6|T_{\textrm{A}}^6| + 2(k - |T_{\textrm{A}}^6|) && \mbox{by Equation~\ref{eq:sizeofr4ar2a}} \nonumber\\
    &\geq 6|T_{\textrm{A}}^6| + 2k - 2|T_{\textrm{A}}^6| \nonumber\\
    &\geq 4|T_{\textrm{A}}^6| + 2k\enspace. \label{eq:upperbounduma}
\end{align}
By Lemma~\ref{lem:3dsrsasbin_tau_a_geq_tau_opt_over_3} and Inequality~\ref{eq:upperbounduma} we obtain the following lower bound for  $u(M_{\textrm{A}})$ in terms of $|T_{\textrm{opt}}^6|$ and $k$:
\begin{align}
    u(M_{\textrm{A}}) &\geq 4|T_{\textrm{A}}^6| + 2k && \mbox{(Inequality~\ref{eq:upperbounduma})} \nonumber\\
    &\geq 4(|T_{\textrm{opt}}^6|/3) + 2k && \mbox{by Lemma~\ref{lem:3dsrsasbin_tau_a_geq_tau_opt_over_3}} \nonumber\\
    &\geq 4|T_{\textrm{opt}}^6|/3 + 2k\enspace. \label{eq:final}
\end{align}
Thus, by Inequality~\ref{eq:final}:
\begin{align}
    2u(M_{\textrm{A}}) &\geq 8|T_{\textrm{opt}}^6|/3 + 4k \nonumber\\
    &\geq 2|T_{\textrm{opt}}^6| + 4k \nonumber\\
    &\geq u(M_{\textrm{opt}}) && \mbox{by Inequality~\ref{eq:lowerboundumopt}.} \nonumber
\end{align}
\end{proof}

We now consider the case when $|T_{\textrm{A}}^0|>0$.
\begin{lem}
\label{lem:3dsrsasbin_approx_if000thentisnotcomplete}
If $|T_{\textrm{A}}^0|>0$ then $|Y|<\lfloor |U|/3 \rfloor$.
\end{lem}
\begin{proof}
We prove the contrapositive. Suppose $|Y| \geq \lfloor |U|/3 \rfloor$. By construction, $X\subseteq Y$ is a set of pairs where $\mathit{val}_{\alpha_p}(\alpha_q)=1$ for each pair $\{\alpha_p, \alpha_q\} \in X$. It follows that each triple in $M_2$ contains two agents $\alpha_p, \alpha_q$ for which $\{ \alpha_p, \alpha_q \} \in X$ and hence $\mathit{val}_{\alpha_p}(\alpha_q)=1$. We thus obtain $u_{t}(M_{\textrm{A}})\geq 2$ for any triple $t \in M_2$. Since $M_1$ is a $P$\nobreakdash-matching, it also holds that $u_{t}(M_{\textrm{A}}) \geq 2$ for any $t \in M_1$. This shows that $|T_{\textrm{A}}^0|=\varnothing$.
\end{proof}

\begin{lem}
\label{lem:3dsrsasbin_approx_if000theneveryagentintgets1}
If $|T_{\textrm{A}}^0|>0$ then $u_{\alpha_p}(M_{\textrm{A}})\geq 1$ for any $\alpha_p \in \bigcup Y$.
\end{lem}
\begin{proof}
Suppose $|T_{\textrm{A}}^0|>0$. Consider an arbitrary $\alpha_p \in \bigcup Y$. It follows that some $\alpha_q\in N$ exists where $\{ \alpha_p, \alpha_q \}\in Y$ and hence $\mathit{val}_{\alpha_p}(\alpha_q)=1$, by the definition of $Y$. 

By Lemma~\ref{lem:3dsrsasbin_approx_if000thentisnotcomplete}, $|Y| < \lfloor |U|/3 \rfloor$. It follows that $\{ \alpha_p, \alpha_q \} \in X$. It follows that there exists some $i$ where $1 \leq i \leq \lfloor |U|/3 \rfloor$ for which $X_i = \{ \alpha_p, \alpha_q \}$ and hence, by construction of $M_2$, the triple $x_i \cup \{ z_i \}$ belongs to $M_2$. It follows that $\alpha_q\in M_2(\alpha_p)$ and hence $u_{\alpha_p}(M_{\textrm{A}})\geq 1$.
\end{proof}

\begin{lem}
\label{lem:3dsrsasbin_000exists_stlemma}
Suppose $|T_{\textrm{A}}^0|>0$. For any $\alpha_r, \alpha_s \in N$, if $\mathit{val}_{\alpha_r}(\alpha_s)=1$ then $u_{\{ \alpha_r, \alpha_s\}}(M_{\textrm{A}})\geq 1$.
\end{lem}
\begin{proof}
Suppose $|T_{\textrm{A}}^0|>0$.

Suppose for a contradiction that some $\alpha_r, \alpha_s \in N$ exist where $\mathit{val}_{\alpha_r}(\alpha_s)=1$ and $u_{\{ \alpha_r, \alpha_s\}}(M_{\textrm{A}})=0$. It follows that $u_{\alpha_r}(M_{\textrm{A}}) = u_{\alpha_s}(M_{\textrm{A}})=0$. It follows, by the definition of a $P$\nobreakdash-matching, that $\alpha_r, \alpha_s$ are unmatched in $M_1$ and hence $\alpha_r, \alpha_s \in U$. By Lemma~\ref{lem:3dsrsasbin_approx_if000theneveryagentintgets1} it follows that $\alpha_r \notin \bigcup Y$ and $\alpha_s \notin \bigcup Y$. It follows that $Y'=Y \cup \{ \alpha_r, \alpha_s \}$ is a disjoint set of pairs of agents in $U$ where $\mathit{val}_{\alpha_p}(\alpha_q)=1$ for each pair $\{ \alpha_p, \alpha_q \} \in Y'$. Since $|Y'|>|Y|$, this contradicts the maximality of $Y$ returned by Algorithm~\texttt{maximum2DMatching}. It follows that no such $\alpha_r, \alpha_s$ exist where $\mathit{val}_{\alpha_r}(\alpha_s)=1$ and $u_{\{ \alpha_r, \alpha_s\}}(M_{\textrm{A}})=0$.
\end{proof}

\begin{lem}
\label{lem:3dsrsasbin_some000exists_r1_lem}
If $|T_{\textrm{A}}^0|>0$ then $u_t(M_{\textrm{A}})\geq 3$ for any $t\in T_{\textrm{opt}}^6$.
\end{lem}
\begin{proof}
Suppose $|T_{\textrm{A}}^0|>0$. Consider an arbitrary $\{ \alpha_{h_1}, \alpha_{h_2}, \alpha_{h_3} \} \in T_{\textrm{opt}}^6$. By definition, $\mathit{val}_{\alpha_{h_1}}(\alpha_{h_2})=\mathit{val}_{\alpha_{h_2}}(\alpha_{h_3})=\mathit{val}_{\alpha_{h_3}}(\alpha_{h_1})=1$. Since $M_\textrm{A}$ is a stable matching, the triple $\{ \alpha_{h_1}, \alpha_{h_2}, \alpha_{h_3} \}$ does not block $M_\textrm{A}$. It follows that at least one of the following holds: $u_{\alpha_{h_1}}(M_{\textrm{A}})=2$, $u_{\alpha_{h_2}}(M_{\textrm{A}})=2$, or $u_{\alpha_{h_3}}(M_{\textrm{A}})=2$. Suppose without loss of generality that $u_{\alpha_{h_1}}(M_{\textrm{A}})=2$. By Lemma~\ref{lem:3dsrsasbin_000exists_stlemma}, it must be that $u_{\{ \alpha_{h_2}, \alpha_{h_3}\}}(M_{\textrm{A}})\geq 1$. In total, $u_{\{ \alpha_{h_1}, \alpha_{h_2},\allowbreak \alpha_{h_3} \}}(M_{\textrm{A}}) \geq 3$.
\end{proof}

\begin{lem}
\label{lem:3dsrsasbin_some000exists_r2_lem}
If $|T_{\textrm{A}}^0|>0$ then $u_t(M_{\textrm{A}})\geq 2$ for any $t\in T_{\textrm{opt}}^4$.
\end{lem}
\begin{proof}
Suppose $|T_{\textrm{A}}^0|>0$. Consider an arbitrary $\{ \alpha_{h_1}, \alpha_{h_2}, \alpha_{h_3} \} \in T_{\textrm{opt}}^4$ where $\mathit{val}_{\alpha_{h_1}}(\alpha_{h_2}) = \mathit{val}_{\alpha_{h_2}}(\alpha_{h_3}) = 1$ and $\mathit{val}_{\alpha_{h_1}}(\alpha_{h_3})=0$. Suppose for a contradiction that $u_{\{ \alpha_{h_1}, \alpha_{h_2}, \alpha_{h_3} \}}(M_{\textrm{A}}) < 2$.

If $u_{\{ \alpha_{h_1}, \alpha_{h_2}, \alpha_{h_3} \}}(M_{\textrm{A}}) = 0$, then $\{ \alpha_{h_1}, \alpha_{h_2}, \alpha_{h_3} \}$ blocks $M_{\textrm{A}}$ in $(N,V)$. It must be that $u_{\{ \alpha_{h_1}, \alpha_{h_2}, \alpha_{h_3} \}}(M_{\textrm{A}})=1$. By Lemma~\ref{lem:3dsrsasbin_000exists_stlemma}, it must be that $u_{\{ \alpha_{h_1}, \alpha_{h_2}\}}(M_{\textrm{A}})\geq 1$ and also that $u_{\{ \alpha_{h_2}, \alpha_{h_3}\}}(M_{\textrm{A}})\geq 1$. It follows that $u_{\alpha_{h_1}}(M_{\textrm{A}}) = u_{\alpha_{h_3}}(M_{\textrm{A}}) = 0$ and $u_{\alpha_{h_2}}(M_{\textrm{A}}) = 1$. In this case, $\{ \alpha_{h_1}, \alpha_{h_2}, \alpha_{h_3} \}$ blocks $M_{\textrm{A}}$ in $(N, V)$, which is a contradiction. It follows that $u_{\{ \alpha_{h_1}, \alpha_{h_2}, \alpha_{h_3} \}}(M_{\textrm{A}}) \geq 2$.
\end{proof}

\begin{lem}
\label{lem:3dsrsasbin_some000exists_r3_lem}
If $|T_{\textrm{A}}^0|>0$ then $u_t(M_{\textrm{A}})\geq 1$ for any $t\in T_{\textrm{opt}}^2$.
\end{lem}
\begin{proof}
Suppose $|T_{\textrm{A}}^0|>0$. Consider an arbitrary $\{ \alpha_{h_1}, \alpha_{h_2}, \alpha_{h_3} \} \in T_{\textrm{opt}}^2$ where $\mathit{val}_{\alpha_{h_1}}(\alpha_{h_2})=1$ and $\mathit{val}_{\alpha_{h_1}}(\alpha_{h_3})=\mathit{val}_{\alpha_{h_2}}(\alpha_{h_3})=0$. By Lemma~\ref{lem:3dsrsasbin_000exists_stlemma}, it must be that $u_{\{ \alpha_{h_1}, \alpha_{h_2}\}}(M_{\textrm{A}})\geq 1$ and hence $u_{\{ \alpha_{h_1}, \alpha_{h_2}, \alpha_{h_3} \}}(M_{\textrm{A}})\geq 1$.
\end{proof}

\begin{lem}
\label{lem:3dsrsasbin_some000exists_final_lem}
If $|T_{\textrm{A}}^0|>0$ then $2{w(M_{\textrm{A}})} \geq w(M_\textrm{opt})$.
\end{lem}
\begin{proof}
Suppose that $|T_{\textrm{A}}^0|>0$. Intuitively, in this lemma the utilitarian welfare in $M_{\textrm{A}}$ is apportioned by considering the utilitarian welfare in $M_{\textrm{A}}$ of each triple in $M_{\textrm{opt}}$. By definition,
\begin{align}
    u(M_{\textrm{A}}) &= \sum\limits_{t\in M_{\textrm{opt}}} u_{t}(M_{\textrm{A}}) \nonumber\\
    &= \sum\limits_{t \in T_{\textrm{opt}}^6} u_{t}(M_{\textrm{A}}) +  \sum\limits_{t \in T_{\textrm{opt}}^4} u_{t}(M_{\textrm{A}}) + \sum\limits_{t \in T_{\textrm{opt}}^2} u_{t}(M_{\textrm{A}}) \nonumber \\ 
    & \phantom{=} + 
    \sum\limits_{t \in T_{\textrm{opt}}^0} u_{t}(M_{\textrm{A}}) && \mbox{by Equation~\ref{eq:moptcomposition}} \nonumber \\
    &\geq \sum\limits_{t \in T_{\textrm{opt}}^6} u_{t}(M_{\textrm{A}}) +  \sum\limits_{t \in T_{\textrm{opt}}^4} u_{t}(M_{\textrm{A}})  + 
    \sum\limits_{t \in T_{\textrm{opt}}^2} u_{t}(M_{\textrm{A}}) \nonumber \\
    &\geq 3|T_{\textrm{opt}}^6| + 2|T_{\textrm{opt}}^4| + |T_{\textrm{opt}}^2| && \mbox{by Lemmas~\ref{lem:3dsrsasbin_some000exists_r1_lem},} \nonumber \\
    & && \mbox{ \ref{lem:3dsrsasbin_some000exists_r2_lem}, and~\ref{lem:3dsrsasbin_some000exists_r3_lem}.} \label{eq:final2}
\end{align}
Thus, by Inequality~\ref{eq:final2}:
\begin{align*}
    2u(M_{\textrm{A}}) &\geq 6|T_{\textrm{opt}}^6| + 4|T_{\textrm{opt}}^4| + 2|T_{\textrm{opt}}^2|\\
    &\geq u(M_{\textrm{opt}}) && \mbox{by Equation~\ref{eq:welfareofmopt}.}
\end{align*}
\end{proof}

\begin{lem}
\label{lem:3dsrsasbin_absoluteapproxratio_runningtime}
Algorithm~{\normalfont \texttt{findStableUW}} has running time $O(|N|^3)$.
\end{lem}
\begin{proof}
Since the time complexity of Algorithm~\texttt{findStable} is $O(|N|^3)$ and the time complexity of Algorithm~\texttt{maximum2DMatching} is $O(|N|^2)$. 
\end{proof}

\begin{restatable}{theorem}{threedsrsasbinabsoluteapproxratio}
\label{thm:threedsrsasbin_absoluteapproxratio}
Algorithm~{\normalfont \texttt{findStableUW}} is a 2-approximation algorithm for 3D-SR-SAS-BIN-MAXUW.
\end{restatable}
\begin{proof}
The absolute approximation ratio is shown in Lemmas~\ref{lem:3dsrsasbin_no000exists_lem} and~\ref{lem:3dsrsasbin_some000exists_final_lem}. The running time is shown in Lemma~\ref{lem:3dsrsasbin_absoluteapproxratio_runningtime}. 
\end{proof}

\else

Let $M_{\textrm{A}}$ be an arbitrary matching returned by Algorithm~\texttt{findStableUW} given $(N,V)$. Suppose $M_{\textrm{opt}}$ is a stable matching in $(N, V)$ with maximum utilitarian welfare. To prove the performance guarantee of Algorithm~{\normalfont \texttt{findStableUW}} we show that $2{u(M_{\textrm{A}})} \geq u(M_\textrm{opt})$. The proof involves apportioning the welfare of agents in $M_{\textrm{A}}$ by the triples of those agents in $M_{\textrm{opt}}$.



\begin{restatable}{theorem}{threedsrsasbinabsoluteapproxratio}
\label{thm:threedsrsasbin_absoluteapproxratio}
Algorithm~{\normalfont \texttt{findStableUW}} is a 2-approximation algorithm for 3D-SR-SAS-BIN-MAXUW.
\end{restatable}



\fi

In the instance of 3D-SR-SAS-BIN shown in Figure~\ref{fig:3d_sr_sas_bin_algorithm_example_before_scenarios}, Algorithm~\texttt{findStableUW} always returns $M_\textrm{A}=\{ \{ \alpha_3, \alpha_5, \alpha_6 \} \}$ while $M_\textrm{opt}=\{\{ \alpha_1, \alpha_2, \alpha_3 \},\allowbreak \{ \alpha_4, \alpha_5, \alpha_8 \}, \{ \alpha_6, \alpha_7, \alpha_9 \}\}$. Since $u(M_\textrm{A})=6$ and $u(M_{\textrm{opt}})=12$ it follows that $u(M_{\textrm{opt}})=2u(M_{\textrm{A}})$. This shows that the analysis of Algorithm~{\normalfont \texttt{findStableUW}} is tight. Moreover, this particular instance shows that any approximation algorithm with a better performance ratio than $2$ should not always begin, like Algorithm~\texttt{findStableUW} does, by selecting a maximal set of triangles.

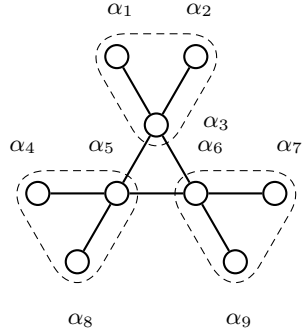
\begin{figure}[t]
    \centering
    \begin{tikzpicture}
\begin{scope}[every node/.style={circle,draw, minimum size=2.4mm}, scale=0.7]
    \begin{scope}
        \begin{scope}[rotate=-60]
            \node[thick, circle, label={[shift={(-0.2, 0.1)}]:$\alpha_5$}] (a1) at (0,0) {};
            \node[thick, circle, label={[shift={(-0.2, 0.1)}]:$\alpha_4$}] (a2) at (-0.75,-1.3) {};
            \node[thick, circle, label={[shift={(0.05, -1.3)}]:$\alpha_8$}] (a3) at (0.75,-1.3) {};
            
            \begin{scope}[scale=2, shift={(0.0, 0.433)}]
            \draw [rounded corners=6.5mm, densely dashed] (0.0, 0.0)--(-0.75, -1.3)--(0.75, -1.3)--cycle;
            \end{scope}
        \end{scope}
        
        \begin{scope}[shift={(1.5, 0.0)}]
            \begin{scope}[rotate=60]
                \node[thick, circle, label={[shift={(0.2, 0.1)}]:$\alpha_6$}] (a4) at (0,0) {};
                \node[thick, circle, label={[shift={(0.05, -1.3)}]:$\alpha_9$}] (a5) at (-0.75,-1.3) {};
                \node[thick, circle, label={[shift={(0.2, 0.1)}]:$\alpha_7$}] (a6) at (0.75,-1.3) {};
                
                \begin{scope}[scale=2, shift={(0.0, 0.433)}]
                \draw [rounded corners=6.5mm, densely dashed] (0.0, 0.0)--(-0.75, -1.3)--(0.75, -1.3)--cycle;
                \end{scope}
            \end{scope}
        \end{scope}
        
        \begin{scope}[shift={(0.75, 1.3)}]
            \begin{scope}[rotate=180]
                \node[thick, circle, label={[shift={(0.8, -0.5)}]:$\alpha_3$}] (a7) at (0,0) {};
                \node[thick, circle, label={[shift={(0.05, 0.1)}]:$\alpha_2$}] (a8) at (-0.75,-1.3) {};
                \node[thick, circle, label={[shift={(0.05, 0.1)}]:$\alpha_1$}] (a9) at (0.75,-1.3) {};
                
                \begin{scope}[scale=2, shift={(0.0, 0.433)}]
                \draw [rounded corners=6.5mm, densely dashed] (0.0, 0.0)--(-0.75, -1.3)--(0.75, -1.3)--cycle;
                \end{scope}
            \end{scope}
        \end{scope}
    \end{scope}

\end{scope}

\begin{scope}
    \foreach \from/\to in {a1/a2, a3/a1, a7/a8, a9/a7, a4/a5, a6/a4, a1/a7, a7/a4, a4/a1}
        \draw [thick] (\from) -- (\to);

\end{scope}
\end{tikzpicture}
    \caption{An instance in which $u(M_{\textrm{opt}})=2u(M_{\textrm{A}})$.}
    \label{fig:3d_sr_sas_bin_algorithm_example_before_scenarios}
\end{figure}

\section{Open questions}
\label{sec:conclusion}
In this paper we have considered the three-dimensional stable roommates problem with additively separable preferences. We considered the special cases in which preferences are binary but not necessarily symmetric, and both binary and symmetric.  There are several interesting directions for future research.

\begin{itemize}
    \item[\textbullet] Does there exist an approximation algorithm for 3D-SR-SAS-BIN-MAXUW (Section~\ref{sec:3dsrsasbin_utilitarianwelfare}) with a better performance guarantee than $2$?
    \item[\textbullet] In 3D-SR-AS, there are numerous possible restrictions besides symmetric and binary preferences. Do any other restrictions ensure that a stable matching exists? For example, we could consider the restriction in which preferences are symmetric and $\mathit{val}_{\alpha_i}\in \{0,1,2\}$ for each $\alpha_i \in N$.
    \item[\textbullet] Additively separable preferences are one possible structure of agents' preferences that can be applied in a model of three-dimensional SR. Are there other systems of preferences that result in new models in which a stable matching can be found in polynomial time?
    \item[\textbullet] The 3D-SR-AS problem model can be generalised to higher dimensions. It would be natural to ask if the algorithm for 3D-SR-SAS-BIN can be generalised to the same problem in $k\geq 3$ dimensions, in which a $k$-set of agents $S$ is blocking if, for each of the $k$ agents in $S$, the utility of $S$ is strictly greater than that agent's utility in the matching. We conjecture that when $k\geq 4$, a stable matching need not exist, and that the associated decision problem is $\NP$-complete, even when preferences are both binary and symmetric.
\end{itemize}

\bibliographystyle{splncs04}
\bibliography{michael}

\ifdefined\hideappendix
\else
\ifdefined\fullversion \else

\clearpage
\appendix
\section{Proofs in detail}
\label{sec:appendix}
\setcounter{secnumdepth}{3} 

\subsection{Preliminaries}

\threedsrsasbinblockerimprovement*
\begin{proof}
Suppose that the triple $\{ \alpha_{k_1}, \alpha_{k_2}, \alpha_{k_3} \}$ blocks $M'$. It follows that $u_{\alpha_{k_1}}(\{\alpha_{k_2},\allowbreak \alpha_{k_3} \}) >\allowbreak u_{\alpha_{k_1}}(M')$, $u_{\alpha_{k_2}}(\{\alpha_{k_1}, \alpha_{k_3} \}) >\allowbreak u_{\alpha_{k_2}}(M')$, and $u_{\alpha_{k_3}}(\{\alpha_{k_1}, \alpha_{k_2} \}) >\allowbreak u_{\alpha_{k_3}}(M')$. Suppose for a contradiction that no $\alpha_p\in \{ \alpha_{k_1}, \alpha_{k_2},\allowbreak \alpha_{k_3} \}$ exists where $u_{\alpha_p}(M') < u_{\alpha_p}(M)$ and hence $u_{\alpha_{k_r}}(M') \geq u_{\alpha_{k_r}}(M)$ for $1 \leq r \leq 3$. It follows that $u_{\alpha_{k_1}}(\{\alpha_{k_2},\allowbreak \alpha_{k_3} \}) >\allowbreak u_{\alpha_{k_1}}(M)$, $u_{\alpha_{k_2}}(\{\alpha_{k_1},\allowbreak \alpha_{k_3} \}) >\allowbreak u_{\alpha_{k_2}}(M)$, and $u_{\alpha_{k_3}}(\{\alpha_{k_1},\allowbreak \alpha_{k_2} \}) >\allowbreak u_{\alpha_{k_3}}(M)$ and thus that $\{ \alpha_{k_1}, \alpha_{k_2}, \alpha_{k_3} \}$ blocks $M$, a contradiction.
\end{proof}

\subsection{General binary preferences}

The reduction in Section~\ref{sec:generalbinary}, shown in Figure~\ref{fig:3d_sr_as_binary_reduction}, constructs an instance $(N, V)$ of 3D-SR-AS-BIN from an arbitrary instance $G=(W, E)$ of PIT. In Section~\ref{sec:3dsrasbinfirstdirection} we consider the first direction and show that if a partition into triangles $X=\{X_1,X_2,\dots,X_q\}$ exists in $G$ then a stable matching $M$ exists in $(N, V)$. In Section~\ref{sec:3dsrasbinseconddirection} we consider the second direction and show that if a stable matching $M$ exists in $(N,V)$ then a partition into triangles $X = \{ X_1, X_2, \dots, X_q\}$ exists in $G$. Note that the only instance discussed is $(N, V)$ and hence we shorten ``blocks $M$ in $(N, V)$'' to simply ``blocks $M$''.

\subsubsection{Correctness of the reduction: first direction}\hfill
\label{sec:3dsrasbinfirstdirection}

\subsubsection{Correctness of the reduction: second direction}\hfill
\label{sec:3dsrasbinseconddirection}

In this section assume that $M$ is a stable matching in $(N,V)$. We analyse its structure and construct a corresponding partition into triangles in $G$.

\subsubsection{Conclusion}\hfill

\threedsrasbinexistence*
\begin{proof}
We have already shown that the decision version of 3D-SR-AS-BIN belongs to $\NP$. We presented a polynomial time reduction from Partition Into Triangles (PIT) to 3D-SR-AS-BIN. If a partition into triangles exists in the PIT instance $G=(W, E)$ then a stable matching $M$ exists in $(N, V)$ where $|M|=|N|/3$ (Lemma~\ref{lem:3dsrasbinfirstdirection}). If a stable matching $M$ exists in $(N, V)$ where $|M|=|N|/3$ then a partition into triangles exists in $G$ (Lemma~\ref{lem:3dsrasbinpitexists}).
\end{proof}

\subsection{Symmetric binary preferences}\hfill

The layout of the full proofs in this section reflects the layout of the paper.

\subsubsection{Preliminaries}\hfill

\threedsrsasbintrianglefree*
\begin{proof}
The set $M_{\triangle}$ corresponds to a \emph{maximal triangle packing} in the underlying graph \cite{CHATAIGNER20091396}, and thus can be found in $O(|N|^3)$ time. Let $N' = N \setminus \bigcup M_{\triangle}$. Construct $V'$ accordingly. Since each triple in  $M_{\triangle}$ corresponds to a triangle, any agent belonging to a triple in $M_{\triangle}$ gains utility two. It follows that if $M$ is a stable $P$\nobreakdash-matching in $(N', V')$ then $M'=M\cup M_{\triangle}$ is a stable $P$\nobreakdash-matching in $(N, V)$.
\end{proof}

\subsubsection{Repairing a \texorpdfstring{$P$}{P}-matching in a triangle-free instance}\hfill\\

Before presenting the full proof of Lemma~\ref{lem:algoreturnsstablematching} we prove some intermediary lemmas.

\algoreturnsstablematching*
\begin{proof}
By Lemmas~\ref{lem:algoreturnsstablematching_notimecomplex} and~\ref{lem:3dsrsasbin_almosttherealgo_runningtime}.
\end{proof}
\subsubsection{Finding a stable \texorpdfstring{$P$}{P}-matching in a triangle-free instance}\hfill

\threedsrsasbinalgfindsstablepmatching*
\begin{proof}
By Lemmas~\ref{lem:3dsrsasbin_algfindsstablepmatching_notimecomplex} and~\ref{lem:3dsrsasbin_algfindstablepmatchingrunningtime}.
\end{proof}

\subsubsection{Finding a stable \texorpdfstring{$P$}{P}-matching in an arbitrary instance}\hfill\\
\label{sec:findingpmatchingingeneralappendix}

For completeness, we provide a sample implementation of Algorithm~{\normalfont \texttt{findStable}}. Suppose a suitable implementation of Algorithm~{\normalfont \texttt{eliminateTriangles}} is supplied, which runs in $O(|N|^3)$ time (Lemma~\ref{lem:threedsrsasbintrianglefree}).

\begin{algorithm}[H]
\textbf{Input:} an instance $(N,V)$ of 3D-SR-SAS-BIN\\
\textbf{Output:} stable $P$-matching $M'$ in $(N,V)$
\smallskip 
\begin{algorithmic}
\caption{Algorithm~\texttt{findStable} \label{alg:3dsrsasbin_overall_algo}} 

\State $(N', V'), M_{\triangle} \gets \texttt{eliminateTriangles}((N, V))$
\State $M' \gets \texttt{findStableInTriangleFree}((N', V'))$

\smallskip

\State \Return $M' \cup M_{\triangle}$

\medskip
\end{algorithmic}
\end{algorithm}

\begin{lem}
\label{lem:3dsrsasbinconstruction_norunningtime}
Given an instance $(N, V)$ of 3D-SR-SAS-BIN, Algorithm~{\normalfont \texttt{findStable}} returns a stable $P$\nobreakdash-matching.
\end{lem}
\begin{proof} 
A suitable implementation of the subroutine {\normalfont \texttt{eliminateTriangles}} returns a pair $((N', V'), M_{\triangle})$ where $|N'|\leq |N|$ and if $M$ is a stable $P$\nobreakdash-matching in $(N', V')$ then $M' = M \cup M_{\triangle}$ is a stable $P$\nobreakdash-matching in $(N, V)$ (Lemma~\ref{lem:threedsrsasbintrianglefree}). By Lemma~\ref{lem:threedsrsasbin_algfindsstablepmatching}, Algorithm~{\normalfont \texttt{findStableInTriangleFree}} returns $P$\nobreakdash-matching $M'$ that is stable in in $(N', V')$. It follows that $M' \cup M_{\triangle}$ is a $P$\nobreakdash-matching that is stable in $(N, V)$.
\end{proof}

\begin{lem}
\label{lem:3dsrsasbinconstruction_runningtime}
Algorithm~{\normalfont \texttt{findStable}} has running time $O(|N|^3)$.
\end{lem}
\begin{proof}
By definition, Algorithm~{\normalfont \texttt{eliminateTriangles}} has running time $O(|N|^3)$ (Lemma~\ref{lem:threedsrsasbintrianglefree}). By Lemma~\ref{lem:3dsrsasbin_algfindstablepmatchingrunningtime}, Algorithm~{\normalfont \texttt{findStableInTriangleFree}} also has running time $O(|N|^3)$. It follows that Algorithm~{\normalfont \texttt{findStable}} has total running time $O(|N|^3)$.
\end{proof}

\threedsrsasbinconstruction*
\begin{proof}
By Lemmas~\ref{lem:3dsrsasbinconstruction_norunningtime} and~\ref{lem:3dsrsasbinconstruction_runningtime}. If $|N|$ is a multiple of three, then if required the agents unmatched in $M' \cup M_{\triangle}$ can be arbitrarily matched into triples. By Lemma~\ref{lem:3dsrsasbinblockerimprovement}, the resulting matching is still stable in $(N, V)$. 
\end{proof}
\subsubsection{Stability and utilitarian welfare}


\threedsrsasbinmaxutilstablehard*
\begin{proof}
A trivial reduction exists from Partition Into Triangles (PIT, see Section~\ref{sec:generalbinary}) to this decision problem. Let $N=W$ and let $\mathit{val}_{\alpha_i}(\alpha_j)=1$ if $\{w_i, w_j\}\in E$ for every $v_i,v_j \in V$. If a partition into triangles $X=\{X_1,X_2,\dots,X_q\}$ exists then the corresponding matching $M=X$ with welfare $u(M)=2|N|$ is stable, since $u_{\alpha_l}(M)=2$ for any $\alpha_l \in N$. Conversely, if the matching $M$ has utility $u(M)=2|N|$ then it is stable, and a partition into triangles exists.
\end{proof}

\threedsrsasbinabsoluteapproxratio*
\begin{proof}
The absolute approximation ratio is shown in Lemmas~\ref{lem:3dsrsasbin_no000exists_lem} and~\ref{lem:3dsrsasbin_some000exists_final_lem}. The running time is shown in Lemma~\ref{lem:3dsrsasbin_absoluteapproxratio_runningtime}. 
\end{proof}

\fi
\fi
\end{document}